\documentclass[twoside,12pt]{article} 
\usepackage[margin=1.0in]{geometry}

\usepackage{amssymb,amsmath,amsthm}
\usepackage{mathtools}
\usepackage{bm,lscape}
\usepackage{bbm}
\usepackage{graphicx}
\usepackage{array}
\usepackage{caption}
\usepackage{algpseudocode}
\usepackage[ruled,vlined]{algorithm2e}
\usepackage{mathrsfs,dsfont}
\usepackage{wasysym}
\RequirePackage[
    colorlinks=true,
    linkcolor=blue,
    urlcolor=blue,
    citecolor=blue]{hyperref}
\RequirePackage{natbib}
\usepackage[utf8]{inputenc}
\usepackage[english]{babel}

\usepackage{graphicx}
\usepackage{color}
\usepackage{rotating}
\usepackage{authblk}
\usepackage{appendix}

\allowdisplaybreaks

\def\simind{\stackrel{\mbox{\scriptsize{ind}}}{\sim}}
\def\simiid{\stackrel{\mbox{\scriptsize{iid}}}{\sim}}

\allowdisplaybreaks

\newcommand{\E}{\mathds{E}}

\renewcommand{\P}{\mathds{P}}

\newcommand{\PP}{\mathds{P}}

\newtheorem{thm}{Theorem}[section]
\newtheorem{lem}[thm]{Lemma}

\newtheorem{prp}[thm]{Proposition}

\providecommand{\keywords}[1]
{
  \small	
  \textbf{\textit{Keywords:}} #1
}

\begin{document}

\title{Merging of Bayes and quasi-Bayes empirical Bayes procedures for Poisson compound decisions}

% \author{
% Emanuele Dolera\\
% \texttt{emanuele.dolera@unipv.it}\\
% Department of Mathematics\\ 
% University of Pavia, Italy\\
% \and
% Stefano Favaro\\
% \texttt{stefano.favaro@unito.it}\\
% Department of Economics and Statistics\\ 
% University of Torino and Collegio Carlo Alberto, Italy\\
% \and
% Stefano Peluchetti\\
% \texttt{speluchetti@cogent.co.jp}\\
% Cogent Labs, Tokyo, Japan
% }

\author[1]{Stefano Favaro\thanks{stefano.favaro@unito.it}}
\author[2]{Sandra Fortini\thanks{sandra.fortini@unibocconi.it}}
\affil[1]{\small{Department of Economics and Statistics, University of Torino and Collegio Carlo Alberto, Italy}}
\affil[2]{\small{Department of Decision Sciences, Bocconi University, Italy}}

\maketitle

\begin{abstract}
The Poisson compound decision problem is a long-standing problem in statistics, in which empirical Bayes methods are used to estimate Poisson means under a mixture model. We study this problem from the viewpoint of $g$-modeling, comparing two nonparametric strategies for estimating the unknown mixing distribution: a Bayesian empirical Bayes strategy, based on the Dirichlet process posterior, and a quasi-Bayesian empirical Bayes strategy, based on Newton's algorithm. The latter is computationally attractive, but its relationship with the Bayesian strategy requires theoretical justification. Under a Poisson mixture model with a ``true'', or oracle, mixing distribution, we establish concentration rates for the marginal probability mass functions induced by the Bayesian and quasi-Bayesian estimates. These rates are then translated into rates of decay for the corresponding regrets, interpreted as excess Bayes risks, and used to prove a frequentist merging result between the Bayesian and quasi-Bayesian empirical Bayes strategies. We also extend the analysis to the multidimensional Poisson compound decision problem. Numerical experiments on synthetic data illustrate that the quasi-Bayesian
strategy achieves accuracy comparable to the Bayesian strategy, while requiring substantially fewer computational resources, especially in the multidimensional setting.
\end{abstract}

\keywords{Dirichlet process prior; empirical Bayes; frequentist merging; $g$-modeling; Newton's algorithm; regret}

%%%%%%%%%%%%%%%%%%%%%%%%%%%%%%%%
%%%%%%%%%%%%%%%%%%%%%%%%%%%%%%%%
%%%%%%%%%%%%%%%%%%%%%%%%%%%%%%%%
%%%%%%%%%%%%%%%%%%%%%%%%%%%%%%%%

\tableofcontents

\section{Introduction}

\subsection{Background and motivations}

Given $n\geq1$ observations modeled as independent Poisson random variables $Y_{1},\ldots,Y_{n}$ with means $\theta_{1},\ldots,\theta_{n}$, respectively, the Poisson compound decision problem concerns the estimation of $(\theta_{1},\ldots,\theta_{n})$ under the squared-error loss. Empirical Bayes provides a general approach to compound decision problems \citep{Rob(51),Rob(56),Zha(03)}. Specifically, assuming that the parameters $\theta_{i}$'s are i.i.d. as a prior $G$ on $\Theta\subseteq\mathbb{R}_{+}$, the best estimate of $\theta$ is the posterior mean or Bayes estimate, namely
\begin{equation}\label{eq:robbins}
\hat{\theta}_{G}(y)=E_{G}[\theta\mid Y_{i}=y]=\frac{\int_{\Theta}\theta\frac{\text{e}^{-\theta}\theta^{y}}{y!}G(d\theta)}{\int_{\Theta}\frac{\text{e}^{-\theta}\theta^{y}}{y!}G(d\theta)}=(y+1)\frac{p_{G}(y+1)}{p_{G}(y)}\qquad y\in\mathbb{N}_{0},
\end{equation}
where
\begin{displaymath}
p_{G}(y)=\int_\Theta \text{Poisson}(y\mid\theta)G(d\theta),
\end{displaymath}
where $\text{Poisson}(\cdot\mid\theta)$ denotes the Poisson probability mass function with mean $\theta\in\mathbb{R}_{+}$, i.e. the Poisson kernel. The Bayes estimate \eqref{eq:robbins} is referred to as Robbins' formula. Since $G$ is unknown in practice, the empirical Bayes approach proceeds by estimating either $p_G$, in the so-called \(f\)-modeling strategy, or $G$, in the so-called $g$-modeling strategy; see \citet{Efr(14),Efr(19)} and references therein.

A prominent nonparametric $f$-modeling strategy is Robbins' method \citep{Rob(56)}, which replaces $p_{G}$ in \eqref{eq:robbins} with the empirical distribution of the $Y_{i}$'s. Despite its conceptual appeal and computational simplicity, $f$-modeling can be numerically unstable and lacks robustness. In particular, it is sensitive to outlying observations, or more precisely to counts that occur rarely, which may yield exceptionally small or large estimates \citep{Bro(13),She(24)}. Examples of nonparametric $g$-modeling include estimates of $G$ in \eqref{eq:robbins} based on maximum likelihood and minimum distance methods \citep{Jan(24)}. Compared with $f$-modeling, $g$-modeling typically produces more accurate estimates and allows prior information to be incorporated more naturally. Its main drawback is computational cost, especially in  high-dimensional settings \citep{Jan(23),Teh(25)}.

A Bayesian $g$-modeling strategy relies on assigning a prior distribution to $G$, yielding what is commonly termed Bayes empirical Bayes \citep{Dee(81),Efr(19)}. A recent contribution in this direction is \citet{Can(26)}, where the existence of universal nonparametric priors with optimality guarantees for empirical Bayes estimation in the Poisson compound decision problem is established. In the related but distinct setting of the Gaussian compound decision problem, \citet{Ign(26)} propose a Bayesian nonparametric $g$-modeling strategy based on a Dirichlet process prior for $G$ \citep{Fer(73)}. The Dirichlet process prior is arguably one of the most widely used nonparametric priors, leading to the Dirichlet process mixture model \citep{Lo(84)}. Beyond the work of \citet{Can(26)}, however, Bayesian nonparametric $g$-modeling for the Poisson compound decision problem remains largely unexplored, especially in the $d$-dimensional setting \citep{Bro(85),Joh(86)}, where it does not appear to have been considered in the literature. Under the Dirichlet process mixture model, the mixing distribution $G$ is estimated from its posterior distribution given $Y_{1},\ldots,Y_{n}$, whose computation can rely on a wide range of numerical methods developed in the literature over the past three decades. These include MCMC methods based on auxiliary-variable, stick-breaking, slice-sampling and retrospective representations \citep{Nea(00),IshJam(01),Wal(07),PapRob(08)}, as well as numerous variational approximations \citep{BleJor(06)}. Although broadly applicable, these methods can be computationally demanding, particularly in the multidimensional setting.

\subsection{Preview of our contributions}

We consider a quasi-Bayesian $g$-modeling strategy for the Poisson compound decision problem, proposed by \citet{Fav(24)}, which estimates the unknown mixing distribution $G$ through the recursive procedure of \citet{Smi(78)}, referred to as Newton's algorithm \citep{New(98),Mar(08)}. This strategy is known to be computationally attractive relative to the Bayesian strategy based on a Dirichlet process mixture model: it replaces posterior inference for $G$ given $Y_{1},\ldots,Y_{n}$ by a sequence of straightforward recursive updates of $G$ \citep{For(20),Fav(24)}. We support this computational advantage with a rigorous frequentist validation for viewing the quasi-Bayesian strategy as a fast approximation to the Bayesian strategy. Under the assumption that $Y_1,\ldots,Y_n$ are i.i.d. from a Poisson mixture model with ``true''  mixing distribution \(G^\ast\), referred to as the oracle prior, we prove that the Bayesian and quasi-Bayesian strategies merge asymptotically by quantifying the rate at which the corresponding plug-in Bayes and quasi-Bayes estimates become equivalent as $n\rightarrow+\infty$. We further extend this frequentist merging result to the $d$-dimensional Poisson compound decision problem. This extension is especially relevant because posterior computation for the Bayesian strategy becomes substantially more demanding in higher dimensions. By contrast, the quasi-Bayesian strategy continues to rely on recursive updates of the mixing distribution $G$, making it particularly appealing in multidimensional settings.

We present numerical experiments on synthetic data that complement the theoretical analysis and illustrate the computational advantage of the quasi-Bayesian strategy over the Bayesian strategy. For each strategy, we compare the accuracy of the resulting empirical Bayes estimate with the computational effort required to obtain it, measured in terms of computational units and CPU time for processing $Y_{1},\ldots,Y_{n}$. The experiments show that the quasi-Bayesian strategy can attain an accuracy comparable to that of the Bayesian strategy, while requiring substantially fewer computational resources, with the advantage becoming more pronounced in multidimensional settings.

\subsection{Organization of the paper}

The paper is organized as follows. In Section~\ref{sec2}, we introduce the Bayesian and quasi-Bayesian $g$-modeling strategies and prove their frequentist merging. In Section~\ref{sec3}, we present numerical experiments comparing the two strategies. In Section~\ref{sec4}, we conclude with some remarks and directions for future work. Proofs and additional numerical illustrations are deferred to the Appendix \ref{app1}-\ref{app4}.

%%%%%%%%%%%%%%%%%%%%%%%%%%%%%%%%
%%%%%%%%%%%%%%%%%%%%%%%%%%%%%%%%
%%%%%%%%%%%%%%%%%%%%%%%%%%%%%%%%
%%%%%%%%%%%%%%%%%%%%%%%%%%%%%%%%

\section{Merging of Bayes and quasi-Bayes empirical Bayes procedures}\label{sec2}
\subsection{The $1$-dimensional setting}

We study the merging of two $g$-modeling empirical Bayes strategies for estimating the unknown mixing distribution $G$ in Robbins' formula \eqref{eq:robbins}: a Bayesian strategy and a quasi-Bayesian strategy. The notion of merging considered here is frequentist in nature, and therefore we evaluate both strategies under a ``true'' Poisson mixture model.  Specifically, we assume that $Y_{1:n}=(Y_{1},\ldots,Y_{n})$ are i.i.d. from a mixture of Poisson distributions with the ``true'' mixing distribution $G^{\ast}$ on $\Theta^{\ast}\subset\mathbb{R}_{+}$, or oracle prior. If $G^\ast$ were known, then Robbins' formula \eqref{eq:robbins} would yield the oracle Bayes estimate
\begin{equation}\label{eq:oracle}
\hat{\theta}^{\ast}(y):=\hat{\theta}_{G^{\ast}}(y)=(y+1)\frac{p_{G^{\ast}}(y+1)}{p_{G^{\ast}}(y)}\qquad  y\in\mathbb{N}_{0}.
\end{equation}
Since $G^{\ast}$ is unknown, the empirical Bayes approach proceeds by replacing the mixing distribution $G$ in Robbins' formula \eqref{eq:robbins} by an estimate, say $\hat G_n$, computed from $Y_{1:n}$. This gives the plug-in empirical Bayes estimate
\begin{equation}\label{eq:estimate}
\hat{\theta}_{n}(y):=\hat{\theta}_{\hat{G}_{n}}=(y+1)\frac{p_{\hat{G}_{n}}(y+1)}{p_{\hat{G}_{n}}(y)}\qquad  y\in\mathbb{N}_{0}.
\end{equation}
The regret incurring by using $\hat{\theta}_{n}$ in place of $\hat{\theta}^{\ast}$, also known as excess Bayes risk  \citep[Section 3]{Efr(19)}, is 
\begin{equation}\label{reg}
\text{Regret}(\hat{G}_n,G^*)=\sum_{y\in\mathbb{N}_{0}}(\hat\theta_{n}(y)-\hat\theta^{*}(y))^2p_{G^*}(y).
\end{equation}
A standard goal in the analysis of empirical Bayes estimates, consists in proving that the regret goes to zero at the minimax rate, as $n\rightarrow+\infty$ \citep{Pol(21)}; see also \citet{Efr(14),Efr(19)}.

In the following, we provide a regret analysis of the plug-in Bayes and quasi-Bayes estimates. Under the ``true'' Poisson mixture model, this analysis establishes the frequentist merging of the Bayesian and quasi-Bayesian strategies, quantifying the rate at which they become equivalent as $n\to+\infty$. Throughout our analysis, we assume that the parameter spaces $\Theta$ and $\Theta^{\ast}$ are compact subsets of $\mathbb{R}_{+}$.

\subsubsection{Bayes empirical Bayes}
We consider a nonparametric version of the Bayesian $g$-modeling strategy of \citet{Dee(81)}. The strategy estimates the mixing distribution $G$ in Robbins' formula \eqref{eq:robbins} by assigning a nonparametric prior distribution to $G$, namely a distribution on the space $\mathcal P(\Theta)$ of probability measures on $\Theta$. Specifically, we place on $G$ a Dirichlet process prior with strength parameter $c>0$ and base probability measure $H$ on $\Theta$ \citep{Fer(73)}; that is, $G\sim \mathrm{DP}(c,H)$. This leads to model $Y_{1:n}$ as follows:
\begin{equation}\label{eq:bayes_model}
\begin{array}{r@{\;}c@{\;}l@{\quad}c@{\quad}l}
Y_i& \mid &\theta_i& \simind &\mathrm{Poisson}(\cdot\mid \theta_i),\qquad i=1,\ldots,n,\\[0.2cm]
\theta_i& \mid &G& \simiid &G,\\[0.2cm]
&&G& \sim &\mathrm{DP}(c,H).
\end{array}
\end{equation}
Because of the conjugacy of the Dirichlet process prior \citep{Fer(73)}, the posterior distribution of $G$ given $Y_{1:n}$ is 
\begin{equation}\label{post_dp}
\Pi(dG\mid Y_{1:n})=\int_{\Theta^n}\mathrm{DP}\left(dG;\,c+n,\,\frac{cH+\sum_{i=1}^n\delta_{\theta_i}}{c+n}\right)\pi(d\theta_{1:n}\mid Y_{1:n}),
\end{equation}
where $\pi(\cdot\mid Y_{1:n})$ denotes the posterior distribution of the $\theta_{i}$'s under the model \eqref{eq:bayes_model}. See \citet{Lo(84)} for details.

Under the Bayesian model \eqref{eq:bayes_model}, and adopting a quadratic loss function for the evaluation functionals $G(A)$, with $A$ being a Borel set of $\Theta$, the optimal estimate of $G$ is given by the posterior mean
\begin{equation}\label{post_mean}
\hat G_n^{\text{\tiny{[B]}}}(\cdot)=\int_{\mathcal P(\Theta)}G(\cdot)\,\Pi(dG\mid Y_{1:n})=\frac{c}{c+n}H(\cdot)+\frac{1}{c+n}\sum_{i=1}^n\pi(\theta_i\in \cdot\mid Y_{1:n}).
\end{equation}
A Bayes estimate of $\theta_{i}$, $i=1,\ldots,n$, is then obtained by plugging \eqref{post_mean} into Robbin's formula \eqref{eq:robbins}, i.e., 
\begin{equation}\label{eq:bayes_eb_est}
\hat{\theta}^{\text{\tiny{[B]}}}_{n}(y):=\hat{\theta}_{\hat G_n^{\text{\tiny{[B]}}}}(y)=(y+1)\frac{p_{\hat G_n^{\text{\tiny{[B]}}}}(y+1)}{p_{\hat G_n^{\text{\tiny{[B]}}}}(y)}\qquad y\in\mathbb{N}_{0}
\end{equation}
Since $G\mapsto \hat{\theta}_G(y)$ is nonlinear, the estimate \eqref{eq:bayes_eb_est} does not, in general, coincide with the posterior mean of $\hat\theta_G(y)$.

We next provide a frequentist validation of the Bayes estimate $\hat{\theta}^{\text{\tiny{[B]}}}_{n}$. Since Robbins' formula depends on $G$ only through the probability mass function $p_G$, we study the posterior concentration of $p_{\hat G_n^{\text{\tiny{[B]}}}}$ around the ``true'' $p_{G^{\ast}}$, as $n\rightarrow+\infty$, namely the consistency of $p_{\hat G_n^{\text{\tiny{[B]}}}}$ under $G^{\ast}$. Although this posterior concentration result does not appear to be stated explicitly in the literature for Dirichlet process mixtures with Poisson kernel, it follows by adapting the arguments developed by \citet[Theorem~5.1]{Gho(01)} for Dirichlet process mixtures with Gaussian kernel; see also \citet{Ign(26)}. Assuming $\Theta=[{h_0},h]$ and $\Theta^{\ast}=[{h_0^\ast},h^{\ast}]$ for ${h_0,\;h_0^\ast,h}$ and $h^{\ast}$ such that  $0<h_0\leq h_0^*<h^\ast\leq h<+\infty$, the next proposition establishes a contraction rate of $p_{\hat G_n^{\text{\tiny{[B]}}}}$ toward $p_{G^\ast}$ under $G^{\ast}$.

\begin{prp}\label{pcrdp_1d}
Consider the Bayesian model \eqref{eq:bayes_model}, where $G\sim \operatorname{DP}(1,H)$ and $H$ is a probability measure supported on the set $[h_{0},h]$ with a continuous and strictly positive density. Assume that $Y_{1:n}\simiid p_{G^\ast}$, where $G^{\ast}$ is supported on $[h_0^\ast,h^\ast]$ with $0<h_0\leq h_0^\ast<h^\ast\leq h<+\infty$. If
\begin{equation}\label{rate_dp}
\varepsilon_n^{\text{\tiny{[B]}}}:=\frac{\log n}{\sqrt n},
\end{equation}
then there exist constants $C,c>0$, depending only on $(h_0,h,h_0^\ast,h^\ast,G^\ast,H)$, such that, for all sufficiently large $n$,
\begin{equation}\label{pcr_dp_1}
\P_{G^\ast}^n\!\left[\Pi\!\left(G\in\mathcal P([h_{0},h]):d_{\mathrm H}(p_G,p_{G^\ast})\ge C\varepsilon_n^{\text{\tiny{[B]}}}\,\middle|\, Y_{1:n}\right)\le\exp\{-c\log^2 n\}\right]\ge1-\frac{1}{n}.
\end{equation}
Moreover, if
\begin{displaymath}
\bar p_n(y):=\int_{\mathcal P([h_{0},h])}p_G(y)\,\Pi(dG\mid Y_{1:n}),\qquad y\in\mathbb N_{0},
\end{displaymath}
then there exists a constant $C'>0$, depending only on $(h_0,h,h_0^\ast,h^\ast,G^\ast,H)$, such that, for all sufficiently large $n$,
\begin{equation}\label{pcr_dp_2}
\P_{G^\ast}^n\!\left[d_{\mathrm H}(\bar p_n,p_{G^\ast})\ge C'\varepsilon_n^{\text{\tiny{[B]}}}\right]\le\frac{1}{n}.
\end{equation}
\end{prp}

See Appendix~\ref{app1} for the proof of Proposition~\ref{pcrdp_1d}. By linearity of $G\mapsto p_G$, the average of $p_{G}$ with respect to the posterior distribution \eqref{post_dp} is equivalent to evaluate $p_{G}$ at the posterior mean \eqref{post_mean}. Hence,
\begin{displaymath}
\bar p_n(y)=p_{\hat G_n^{\text{\tiny{[B]}}}}(y),\qquad y\in\mathbb N_0.
\end{displaymath}
Therefore, Proposition~\ref{pcrdp_1d} implies that $p_{\hat G_n^{\text{\tiny{[B]}}}}$ is consistent, in Hellinger distance, for $p_{G^\ast}$, with rate $\log n/\sqrt n$.

An alternative route to prove Proposition \ref{pcrdp_1d} could be based on the work of \citet{Can(26)}. In particular, \citet[Lemma 2.4]{Can(26)} yields the optimal posterior contraction rate $\varepsilon_n^{\text{\tiny{[B]}}}=\frac{\log n}{n\sqrt{\log\log n}}$ under a suitable thickness condition on the prior over the mixing distribution $G$; see \citet[Defintion 2]{Can(26)} for details. To use this approach in the Dirichlet process mixture model, one would need to verify that the Dirichlet process prior satisfies the corresponding thickness condition. We have instead followed the approach of \citet[Theorem~5.1]{Gho(01)}, since the same arguments extend naturally to the $d$-dimensional Poisson compound decision problem, whereas the approach of \citet{Can(26)} is developed for the $1$-dimensional setting.

\subsubsection{Quasi-Bayes empirical Bayes}

We consider the quasi-Bayesian $g$-modeling strategy of \citet{Fav(24)}. The strategy estimates the mixing distribution $G$ in Robbins' formula \eqref{eq:robbins} via Newton's algorithm \citep{New(98)}. Specifically, let $G_{0}$ be a probability measure on $\Theta$, and $\theta_{1}\sim G_{0}$. The $Y_{n}$'s are then modeled as follows:
\begin{equation}\label{eq:qbayes_model}
\begin{array}{r@{\;}c@{\;}l@{\quad}c@{\quad}l}
Y_n& \mid &\theta_n& \simind &\mathrm{Poisson}(\cdot\mid \theta_n),\qquad n\geq 1,\\[0.2cm]
\theta_{n+1}& \mid &Y_{1:n}& \sim &G_n,
\end{array}
\end{equation}
where $G_{n}$ is defined recursively through Newton's algorithm
\begin{equation}\label{eq:newton}
G_{n}(d\theta)=(1-\alpha_{n})G_{n-1}(d\theta)+\alpha_{n}\frac{\text{Poisson}(Y_{n}\mid\theta)G_{n-1}(d\theta)}{\int_\Theta \text{Poisson}(Y_{n}\mid\theta)G_{n-1}(d\theta)},
\end{equation}
with the $\alpha_{n}$'s in $(0,1)$ are such that  $\sum_{n\geq1}\alpha_{n}=+\infty$ and $\sum_{n\geq1}\alpha^{2}_{n}<+\infty$. According to \eqref{eq:newton}, after observing \(Y_{n+1}\), the model updates \(G_n\) by taking a weighted average of \(G_n\) and its posterior distribution based on \(Y_{n+1}\), with weight \(\alpha_{n+1}\). The sequence $(\alpha_{n})_{n\geq1}$ is referred to as the learning rate. A standard choice is $\alpha_{n}=(\alpha+n)^{-\gamma}$, with $\alpha>0$ and $\gamma\in (1/2,1]$; see \citet{For(20)}.

Under the quasi-Bayesian model \eqref{eq:qbayes_model}, an estimate of $G$ is given by $\hat{G}_{\gamma,n}^{\text{\tiny{[Q-B]}}}=G_{n}$. The quasi-Bayes estimate of $\theta_{i}$, $i=1,\ldots,n$, is then obtained by plugging $\hat{G}_{n}^{\text{\tiny{[Q-B]}}}$ into Robbin's formula \eqref{eq:robbins}, i.e., 
\begin{equation}\label{eq:qbayes_eb_est}
\hat{\theta}^{\text{\tiny{[Q-B]}}}_{\gamma,n}(y):=\hat{\theta}_{\hat{G}_{\gamma,n}^{\text{\tiny{[Q-B]}}}}(y)=(y+1)\frac{p_{\hat{G}_{\gamma,n}^{\text{\tiny{[Q-B]}}}}(y+1)}{p_{\hat{G}_{\gamma,n}^{\text{\tiny{[Q-B]}}}}(y)}\qquad y\in\mathbb{N}_{0}
\end{equation}
Since $G\mapsto \hat{\theta}_G(y)$ is nonlinear, the estimate \eqref{eq:qbayes_eb_est} does not, in general, coincide with the posterior mean of $\hat\theta_G(y)$.

A frequentist validation of the quasi-Bayes estimate $\hat{\theta}^{\text{\tiny{[Q-B]}}}_{\gamma,n}$, analogue to Proposition~\ref{pcrdp_1d}, follows from  \citet[{Theorem 4.8 and }Corollary 4.10]{MarTok(09)}. We state the result in the form needed here, namely with $\Theta=[{h_0},h]$ and $\Theta^{\ast}=[{h_0^\ast},h^{\ast}]$ for some ${h_0,\;h_0^\ast,h}$ and $h^{\ast}$ such that  $0<h_0\leq h_0^*<h^\ast\leq h<+\infty$.

\begin{prp}\label{pcrnew_1d}
Consider the quasi-Bayesian model \eqref{eq:qbayes_model}-\eqref{eq:newton}, where: i) $G_0$ is a probability measure supported on the set $[h_0,h]$, with a strictly positive density; ii)  for every $n\geq 1$, the learning rate is $\alpha_n\asymp n^{-\gamma}$ with $\gamma\in(2/3,1]$. Assume that $Y_{1:n}\simiid p_{G^\ast}$, where $G^\ast$ is absolutely continuous with respect to the Lebesgue measure and supported on $[h_0^\ast,h^\ast]$, with $0<h_0\leq h_0^*<h^\ast\leq h<+\infty$. If
\begin{equation}\label{rate_new}
\varepsilon_{\gamma,n}^{\text{\tiny{[Q-B]}}}:=\left\{
\begin{array}{ll}
\dfrac{1}{\sqrt{n^{1-\gamma}}} &\gamma \in (2/3,1)\\
\\
\dfrac{1}{\sqrt{\log n}}&\gamma=1
\end{array}
\right.
\end{equation}
then, for every $\delta>0$
\begin{equation}\label{pcr_newton}
\P_{G^\ast}^n\!\left[d_{\mathrm H}(p_{\hat{G}_{\gamma,n}^{\text{\tiny{[Q-B]}}}},p_{G^\ast})\ge \delta\, \varepsilon_{\gamma,n}^{\text{\tiny{[Q-B]}}}\right]\le\delta
\end{equation}
for all sufficiently large $n$.
\end{prp}

The proof is in Section~\ref{proof:pcrnew_1d}. The concentration rate in Proposition~\ref{pcrnew_1d} reflects the choice of the learning rate $\alpha_n\asymp n^{-\gamma}$. Larger values of $\gamma$ make $\alpha_n$ decay faster to zero, thereby assigning less weight to new observations in \eqref{eq:newton}. This comes at the
cost of a slower concentration rate: for $\gamma\in(2/3,1)$, $\varepsilon_{\gamma,n}^{\text{\tiny{[Q-B]}}}=n^{-(1-\gamma)/2}$, whose exponent decreases to zero as $\gamma$ approaches one, while at the endpoint $\gamma=1$ the rate becomes $1/\sqrt{\log n}$. Thus, compared with the Bayesian rate $\varepsilon_n^{\text{\tiny{[B]}}}=\log n/\sqrt n$ in
Proposition~\ref{pcrdp_1d}, Proposition~\ref{pcrnew_1d} yields a slower concentration rate. Indeed, $\varepsilon_n^{\text{\tiny{[B]}}}=o(\varepsilon_{\gamma,n}^{\text{\tiny{[Q-B]}}})
$ for every \(\gamma\in(2/3,1]\).

\subsubsection{Merging of Bayes and quasi-Bayes empirical Bayes}

We now turn Propositions~\ref{pcrdp_1d} and~\ref{pcrnew_1d} into regret bounds for the plug-in Bayes and quasi-Bayes empirical Bayes estimates. Since Robbins' formula depends on the mixing distribution $G$ only through the marginal probability mass function $p_{G}$, the Hellinger concentration rates obtained above can be used to control the regret incurred by replacing the oracle prior $G^\ast$ with an estimated mixing distribution. The next lemma provides this control for the Bayesian and quasi-Bayesian estimates.

\begin{lem}\label{lemma:Jana}
Let $Y_{1:n}\simiid p_{G^\ast}$, where $G^\ast$ is absolutely continuous with respect to the Lebesgue measure and supported on $[h_0^\ast,h^\ast]$, with $0<h_0^\ast<h^\ast<+\infty$.  
\begin{itemize}
    \item[a)]  Consider the Bayesian model \eqref{eq:bayes_model}, where $G\sim \operatorname{DP}(1,H)$ and $H$ is a probability measure on the set $[h_{0},h]$, for $0<h_{0}\leq h_{0}^{\ast}<h^{\ast}\leq h<+\infty$, with a continuous and strictly positive density. Then 
    %there exist constants $C_1$ and $C_2$, depending only on $h$ and $h^\ast$, such that for sufficiently large $n$
\begin{align*}
\mathrm{Regret}(\hat G_{n}^{\text{\tiny{[B]}}},G^*)=O_{\mathbb P^\ast}\left(\frac{(\log n)^3}{n\log\log n}\right)
%  &\leq C_1\frac{\log n}{\log\log n} d_H^2(p_{\hat G_n^{\tiny{[B]}}},p_{G^*})+C_2\frac{2}{n^5}.
\end{align*}
\item[b)] Consider the quasi-Bayesian model \eqref{eq:qbayes_model}-\eqref{eq:newton}, where: i) $G_0$ is a probability measure on the set $[h_0,h]$, for $0<h_0\leq h_0^\ast<h^\ast\leq h<+\infty$, with a strictly positive density; ii)  for every $n\geq 1$, the learning rate is $\alpha_n\asymp n^{-\gamma}$ with $\gamma\in(2/3,1]$. Then,
\begin{align}\label{eq:regret_rate_qb_1d}
\mathrm{Regret}(\hat G_{\gamma,n}^{\text{\tiny{[Q-B]}}},G^*)=
\left\{
\begin{array}{ll}
o_{\,\mathbb{P}^{\ast}}\left(\dfrac{\log n}{n^{1-\gamma}\log\log n}\right)&\gamma <1\\
\\
o_{\,\mathbb{P}^{\ast}}\left(\dfrac{1}{\log\log n}\right)&\gamma =1.
\end{array}
\right.
%&\leq C_1\frac{\log n}{\log\log n} d_H^2(p_{\hat G_n^{\tiny{[B]}}},p_{G^*})+C_2\frac{2}{n^5}.
\end{align}
\end{itemize}
\end{lem}
\begin{proof} By \citet[Lemma 4]{Jan(24)}, for any $\hat G$ satisfying   $\hat G[0,h]=1$ with $0<h^\ast<h<+\infty$ and for any $K\geq 1$,
\begin{displaymath}
   \operatorname{Regret}(\widehat G,G^\ast)\leq \left\{12(h^2+{h^\ast}^2)+48(h+h^\ast)K\right\}d_H^2(p_{G^\ast},p_{\hat G})+2(h+h^\ast)^2\sum_{y>K} p_{G^*}(y).
\end{displaymath}
Furthermore, by \citet[ Lemma 11]{Jan(24)}, if 
$$
K=\min\left\{\left\lceil \frac{5(he^2+2)\log n}{\log\log n}\right\rceil,he^2+5\log n\right\},
$$
then, for any $n\geq 3$,  $\sum_{y>K} p_{G^\ast}(y)\leq \frac{2}{n^5}$. Hence, for any estimate $\hat G_n$ of the mixing distribution $G$,  there exist constants $C_1$ and $C_2$, depending only on $h$ and $h^\ast$, such that for sufficiently large $n$ it holds
\begin{align*}
  \text{Regret}(\hat G_n,G^\ast)&\leq C_1\frac{\log n}{\log\log n} d_H^2(p_{\hat G_n},p_{G^*})+C_2\frac{2}{n^5}.
\end{align*}
Claim (a) follows from the above inequality applied to the estimate
$\hat G_n^{\text{\tiny{[B]}}}$, together with Proposition
\ref{pcrdp_1d}. Claim (b) follows analogously, by applying the same
inequality to the estimate $\hat G_{\gamma,n}^{\text{\tiny{[Q-B]}}}$ and using
Proposition \ref{pcrnew_1d}.
\end{proof}

The next theorem establishes the frequentist merging of the Bayesian and quasi-Bayesian empirical Bayes strategies. In particular, under the ``true'' mixing distribution $G^{\ast}$ on $\Theta^{\ast}\subseteq\mathbb R_+$, we define the regret incurred by using the quasi-Bayes estimate $\hat\theta_{\gamma,n}^{\text{\tiny{[Q-B]}}}$ in place of the Bayes estimate $\hat\theta_{n}^{\text{\tiny{[B]}}}$ as
\begin{equation}\label{merging_regret_1d}
\text{Regret}(\hat G_{\gamma,n}^{\text{\tiny{[Q-B]}}},\hat G_{n}^{\text{\tiny{[B]}}};G^{\ast})=\sum_{y\in\mathbb N_0} (\hat\theta_{n}^{\text{\tiny{[B]}}}(y)-\hat\theta_{\gamma,n}^{\text{\tiny{[Q-B]}}}(y))^2p_{G^*}(y),
\end{equation}
and show that, as $n\rightarrow+\infty$, it vanishes at the same rate as that of $\text{Regret}(\hat G_{\gamma,n}^{\text{\tiny{[Q-B]}}},G^{\ast})$, displayed in \eqref{eq:regret_rate_qb_1d}.

\begin{thm}\label{merge_1d}
Let $Y_{1:n}\simiid p_{G^\ast}$, where $G^\ast$ is absolutely continuous with respect to the Lebesgue measure and supported on the set $[h_0^\ast,h^\ast]$, with $0<h_0^\ast<h^\ast<+\infty$. Consider the Bayesian model \eqref{eq:bayes_model}, where $G\sim \operatorname{DP}(1,H)$ and $H$ is a probability measure on the set $[h_{0},h]$, for $0<h_{0}\leq h_{0}^{\ast}<h^{\ast}\leq h<+\infty$, with a continuous and strictly positive density. Further, consider the quasi-Bayesian model \eqref{eq:qbayes_model}-\eqref{eq:newton}, where: i) $G_0$ is a probability measure on the set $[h_0,h]$, for $0<h_0\leq h_0^\ast<h^\ast\leq h<+\infty$, with a strictly positive density; ii)  for every $n\geq 1$, the learning rate is $\alpha_n\asymp n^{-\gamma}$ with $\gamma\in(2/3,1]$. %Let $G\sim \Pi= \operatorname{DP}(1,\alpha)$, where $\alpha$ is a probability measure on $[0,h]$ with  a continuous density bounded away from zero and infinity, and let $\hat G_n^{\text{\tiny{[B]}}}(\cdot)=\int_{\mathcal P(\Theta)}G(\cdot)\,\Pi(dG\mid Y_{1:n})$.
%Let $G_0$ be a probability measure supported on  $[h_0,h]$, and with strictly positive density, and for every $n\geq 1$, let $\hat{G}_{\gamma,n}^{\text{\tiny{[Q-B]}}}=G_n$, with $(G_n)$ defined recursively as \eqref{eq:newton} with $\alpha_n\asymp n^{-\gamma}$ for some $\gamma\in (2/3,1]$. 
\begin{itemize}
    \item[a)] The estimates  $p_{n}^{\text{\tiny{[B]}}}$ and $p_{n}^{\text{\tiny{[Q-B]}}}$ of $p_{G^\ast}$ merge in Hellinger distance, as $n\rightarrow\infty$, and
    \begin{displaymath}
d_H(p_{n}^{\text{\tiny{[B]}}},p_{\gamma,n}^{\text{\tiny{[Q-B]}}})=\left\{
\begin{array}{ll}
o_{\,\mathbb{P}^{\ast}}\left(\dfrac {1}{\sqrt{n^{1-\gamma}}}\right)&\gamma<1\\
\\
o_{\,\mathbb{P}^{\ast}}\left(\dfrac{1}{\sqrt{\log n}}\right)&\gamma=1,
\end{array}\right.
\end{displaymath}
\item[b)] The plug-in estimates $\hat\theta_{n}^{\text{\tiny{[B]}}}$ and $\hat\theta_{\gamma,n}^{\text{\tiny{[Q-B]}}}$ merge in $L^2(p_{G^\ast})$, as $n\rightarrow+\infty$, and
\begin{displaymath}
\mathrm{Regret}(\hat G_{\gamma,n}^{\text{\tiny{[Q-B]}}},\hat G_{n}^{\text{\tiny{[B]}}};G^{\ast})=
\left\{
\begin{array}{ll}
o_{\,\mathbb{P}^{\ast}}\left(\dfrac{\log n}{n^{1-\gamma}\log\log n}\right)&\gamma <1\\
\\
o_{\,\mathbb{P}^{\ast}}\left(\dfrac{1}{\log\log n}\right)&\gamma =1.
\end{array}
\right.
\end{displaymath}
\end{itemize}
\end{thm}
\begin{proof} 
With regards to a), by Proposition \ref{pcrdp_1d}, for every $\gamma\in (2/3,1]$, $d_H(p_n^{\text{\tiny{[B]}}},p_{G^\ast})=o_{\PP^\ast}(\epsilon_{\gamma,n}^{\text{\tiny{[Q-B]}}})$, with $\varepsilon_{\gamma,n}^{\text{\tiny{[Q-B]}}}$ as in \eqref{rate_new}. The claim follows by Proposition \ref{pcrnew_1d} and the triangular inequality. With regards to b),
%By Lemma 11 in \cite{Jan(24)}, if $$K=\min\left\{\left\lceil \frac{5(he^2+2)\log n}{\log\log n}\right\rceil,he^2+5\log n\right\},$$then, for any $n\geq 3$,  $\sum_{y>K} p_{G^\ast}(y)\leq \frac{2}{n^5}$.Applying Lemma \ref{lemma:Jana} to $\hat G_n^{\text{\tiny{[B]}}}$ and to $\hat G_n^{\text{\tiny{[Q-B]}}}$, we can ensure the existence of a constants $C_1$ and $C_2$, depending only on $h$ and $h^\ast$, such that for sufficiently large $n$\begin{align*}    \sum_{y\in\mathbb{N}_{0}}(\hat\theta_{n}^{\text{\tiny{[B]}}}(y)-\hat\theta^{*}(y))^2p_{G^*}(y)&\leq C_1\frac{\log n}{\log\log n} d_H^2(p_{\hat G_n^{\tiny{[B]}}},p_{G^*})+C_2\frac{2}{n^5}\end{align*} and \begin{align*}     \sum_{y\in\mathbb{N}_{0}}(\hat\theta_{n}^{\text{\tiny{[Q-B]}}}(y)-\hat\theta^{*}(y))^2p_{G^*}(y)&\leq C_1\frac{\log n}{\log\log n} d_H^2(p_{\hat G_n^{\tiny{[B]}}},p_{G^*})+C_2\frac{2}{n^5}. \end{align*} Thus,
 \begin{align*}
     \text{Regret}(\hat G_{\gamma,n}^{\text{\tiny{[Q-B]}}},\hat G_{n}^{\text{\tiny{[B]}}};G^{\ast})
      &\leq 2\text{Regret}(\hat G_{n}^{\text{\tiny{[B]}}},G^\ast)+2\text{Regret}(\hat G_{n}^{\text{\tiny{[Q-B]}}},G^\ast)\\
      %&\quad\quad  \leq 2 C_1\frac{\log n}{\log\log n}      \left\{ d_H^2(p_{\hat G_n^{\tiny{[B]}}},p_{G^*})+d_H^2(p_{\hat G_{\gamma,n}^{\tiny{[Q-B]}}},p_{G^*})\right\}+2C_2\frac{2}{n^5}\\
    & =
      \left\{
\begin{array}{ll}
o_{\,\mathbb{P}^{\ast}}(\frac{\log n}{n^{1-\gamma}\log\log n})&\gamma <1\\
\\
o_{\,\mathbb{P}^{\ast}}(\frac{1}{\log\log n})&\gamma =1,
\end{array}
\right.
\end{align*}
where the last equality comes from Lemma \ref{lemma:Jana}.
\end{proof}

\subsection{The $d$-dimensional setting, $d>1$}

Let $d\geq1$ be fixed. Given $n\geq1$ observations modeled as independent $d$-dimensional Poisson random vectors $\boldsymbol Y_1,\ldots,\boldsymbol Y_n$, with corresponding mean vectors $\boldsymbol\theta_1,\ldots,\boldsymbol\theta_n\in\mathbb R_+^d$, the $d$-dimensional Poisson compound decision problem concerns the estimation of $(\boldsymbol\theta_1,\ldots,\boldsymbol\theta_n)$ under squared-error loss. Assuming that the $\boldsymbol\theta_i$'s are i.i.d. from a prior $G$ on $\boldsymbol{\Theta}\subseteq\mathbb R_+^d$, the Bayes estimate of $\boldsymbol\theta$ is the posterior mean
\begin{equation}\label{eq:robbins_d_vector}
\hat{\boldsymbol\theta}_{G}(\boldsymbol y)=\E_{G}[\boldsymbol\theta\mid \boldsymbol Y_{i}=\boldsymbol y]=\frac{\int_{\boldsymbol{\Theta}}\boldsymbol\theta\, \prod_{\ell=1}^{d}\text{Poisson}(y_{\ell}\mid\theta_{\ell})\,G(d\boldsymbol\theta)}{\int_{\boldsymbol{\Theta}}\prod_{\ell=1}^{d}\text{Poisson}(y_{\ell}\mid\theta_{\ell})\,G(d\boldsymbol\theta)}\qquad\boldsymbol y\in\mathbb N_0^d.
\end{equation}
Equivalently, if $\boldsymbol e_\ell$ denotes the $\ell$-th coordinate vector in $\mathbb R^d$, then the $\ell$-th coordinate of the Bayes estimate \eqref{eq:robbins_d_vector} is
\begin{displaymath}
\hat\theta_{G,\ell}(\boldsymbol y)=(y_\ell+1)\frac{p_G(\boldsymbol y+\boldsymbol e_\ell)}{p_G(\boldsymbol y)},\qquad \ell=1,\ldots,d,
\end{displaymath}
where
\begin{displaymath}
p_{G}(\boldsymbol y)=\int_{\boldsymbol{\Theta}}\prod_{\ell=1}^{d}\text{Poisson}(y_{\ell}\mid\theta_{\ell})\,G(d\boldsymbol\theta).
\end{displaymath}
The Bayes estimate \eqref{eq:robbins_d_vector} is referred to as multidimensional Robbins' formula; see, e.g., \citet{Bro(85)}.

As in the $1$-dimensional setting, we study the frequentist merging of two $g$-modeling empirical Bayes strategies for estimating the unknown mixing distribution $G$ in the multidimensional Robbins' formula \eqref{eq:robbins_d_vector}: a Bayesian strategy and a quasi-Bayesian strategy. With regards to the ``true'' $d$-dimensional Poisson mixture model, we assume that $\boldsymbol Y_{1:n}=(\boldsymbol Y_1,\ldots,\boldsymbol Y_n)$ are i.i.d. from a $d$-dimensional Poisson mixture distribution with ``true'' mixing distribution $G^\ast$ on $\boldsymbol\Theta^\ast\subset\mathbb R_+^d$, or oracle prior. If $G^\ast$ were known, then the multidimensional Robbins's formula \eqref{eq:robbins_d_vector} would yield the oracle Bayes estimate
\begin{equation}\label{eq:oracle_d}
\hat{\boldsymbol\theta}^{\ast}(\boldsymbol y):=\hat{\boldsymbol\theta}_{G^\ast}(\boldsymbol y),\qquad\boldsymbol y\in\mathbb N_0^d,
\end{equation}
whose $\ell$-th coordinate is
\begin{displaymath}
\hat\theta^{\ast}_{\ell}(\boldsymbol y)=(y_\ell+1)\frac{p_{G^\ast}(\boldsymbol y+\boldsymbol e_\ell)}{p_{G^\ast}(\boldsymbol y)},\qquad \ell=1,\ldots,d.
\end{displaymath}
Since $G^\ast$ is unknown, the empirical Bayes approach proceeds by replacing the mixing distribution $G$ in \eqref{eq:robbins_d_vector} by an estimate, say $\hat G_n$, computed from $\boldsymbol Y_{1:n}$. This gives the plug-in empirical Bayes estimate
\begin{equation}\label{eq:estimate_d}
\hat{\boldsymbol\theta}_{n}(\boldsymbol y):=\hat{\boldsymbol\theta}_{\hat G_n}(\boldsymbol y),\qquad\boldsymbol y\in\mathbb N_0^d,
\end{equation}
whose $\ell$-th coordinate is
\begin{displaymath}
\hat\theta_{n,\ell}(\boldsymbol y)=(y_\ell+1)\frac{p_{\hat G_n}(\boldsymbol y+\boldsymbol e_\ell)}{p_{\hat G_n}(\boldsymbol y)},\qquad \ell=1,\ldots,d.
\end{displaymath}
In particular, in analogy with the $1$-dimensional setting, the regret incurred by using \(\hat{\boldsymbol\theta}_{n}\) in place of \(\hat{\boldsymbol\theta}^{\ast}\) is
\begin{equation}\label{reg_d}
\operatorname{Regret}(\hat G_n,G^\ast)=\sum_{\boldsymbol y\in\mathbb N_0^d}\left\|\hat{\boldsymbol\theta}_{n}(\boldsymbol y)-\hat{\boldsymbol\theta}^{\ast}(\boldsymbol y)
\right\|_2^2p_{G^\ast}(\boldsymbol y).
\end{equation}
This is precisely the $d$-dimensional counterpart of the regret \eqref{reg} in the $1$-dimensional; see \citet{Efr(14),Efr(19)}.

In the following, we provide a regret analysis of the plug-in Bayes and quasi-Bayes estimates. Under the ``true'' Poisson mixture model, this analysis establishes the frequentist merging of the Bayesian and quasi-Bayesian strategies, quantifying the rate at which they become equivalent as $n\to+\infty$. Throughout our analysis, we assume that the parameter spaces $\boldsymbol{\Theta}$ and $\boldsymbol{\Theta}^{\ast}$ are compact subsets of $\mathbb R_+^d$.

\subsubsection{Bayes empirical Bayes}

The Bayes empirical Bayes strategy estimates the mixing distribution $G$ in the multidimensional Robbins' formula \eqref{eq:robbins_d_vector} by assigning a nonparametric prior distribution to $G$, namely a distribution on the space $\mathcal P(\boldsymbol\Theta)$ of probability measures on $\boldsymbol\Theta\subseteq\mathbb R_+^d$. We place on $G$ a Dirichlet process prior with strength parameter $c>0$ and base probability measure $H$ on $\boldsymbol\Theta$. This leads to model $\boldsymbol Y_{1:n}$ as follows:
\begin{equation}\label{eq:bayes_model_d}
\begin{array}{r@{\;}c@{\;}l@{\quad}c@{\quad}l}
Y_{i,\ell}& \mid &\boldsymbol\theta_i& \simind &\mathrm{Poisson}(\cdot\mid \theta_{i,\ell}),\qquad \ell=1,\ldots,d,\quad i=1,\ldots,n,\\[0.2cm]
\boldsymbol\theta_i& \mid &G& \simiid &G,\\[0.2cm]
&&G& \sim &\mathrm{DP}(c,H).
\end{array}
\end{equation}
Because of the conjugacy of the Dirichlet process prior \citep{Fer(73)}, the posterior distribution of $G$ given $\boldsymbol Y_{1:n}$ is
\begin{equation}\label{post_dp_d}\Pi(dG\mid \boldsymbol Y_{1:n})=\int_{\boldsymbol\Theta^n}\mathrm{DP}\left(dG;\,c+n,\,\frac{cH+\sum_{i=1}^n\delta_{\boldsymbol\theta_i}}{c+n}\right)\pi(d\boldsymbol\theta_{1:n}\mid \boldsymbol Y_{1:n}),
\end{equation}
where \(\pi(\cdot\mid \boldsymbol Y_{1:n})\) denotes the posterior distribution of the \(\boldsymbol\theta_i\)'s under the model \eqref{eq:bayes_model_d}. See \citet{Lo(84)} for details.

Under the Bayesian model \eqref{eq:bayes_model_d}, and adopting a quadratic loss function for the evaluation functionals \(G(A)\), with \(A\) being a Borel set of \(\boldsymbol\Theta\), the optimal estimate of \(G\) is given by the posterior mean
\begin{equation}\label{post_mean_d}
\hat G_n^{\text{\tiny{[B]}}}(\cdot)=\int_{\mathcal P(\boldsymbol\Theta)}G(\cdot)\,\Pi(dG\mid \boldsymbol Y_{1:n})=\frac{c}{c+n}H(\cdot)+\frac{1}{c+n}\sum_{i=1}^n
\pi(\boldsymbol\theta_i\in \cdot\mid \boldsymbol Y_{1:n}).
\end{equation}
A Bayes estimate of \(\boldsymbol\theta_i\), \(i=1,\ldots,n\), is then obtained by plugging \eqref{post_mean_d} into the multidimensional Robbins' formula \eqref{eq:robbins_d_vector}, i.e.,
\begin{equation}\label{eq:bayes_eb_est_d}
\hat{\boldsymbol\theta}^{\text{\tiny{[B]}}}_{n}(\boldsymbol y):=\hat{\boldsymbol\theta}_{\hat G_n^{\text{\tiny{[B]}}}}(\boldsymbol y),\qquad\boldsymbol y\in\mathbb N_0^d.
\end{equation}
Equivalently, the $\ell$-th coordinate of \eqref{eq:bayes_eb_est_d} is
\begin{displaymath}
\hat\theta^{\text{\tiny{[B]}}}_{n,\ell}(\boldsymbol y)=(y_\ell+1)\frac{p_{\hat G_n^{\text{\tiny{[B]}}}}(\boldsymbol y+\boldsymbol e_\ell)}{p_{\hat G_n^{\text{\tiny{[B]}}}}(\boldsymbol y)},\qquad \ell=1,\ldots,d.
\end{displaymath}
Since \(G\mapsto \hat{\boldsymbol\theta}_G(\boldsymbol y)\) is nonlinear, the estimate \eqref{eq:bayes_eb_est_d} does not, in general, coincide with the posterior mean of \(\hat{\boldsymbol\theta}_G(\boldsymbol y)\).

We next provide a frequentist validation of the Bayes estimate $\hat{\boldsymbol\theta}^{\text{\tiny{[B]}}}_{n}$, in analogy with Proposition \ref{pcrdp_1d}. Since the multidimensional Robbins' formula \eqref{eq:robbins_d_vector} depends on the mixing distribution $G$ only through the probability mass function $p_G$, we study the posterior concentration of $p_{\hat G_n^{\text{\tiny{[B]}}}}$ around the ``true'' $p_{G^\ast}$, as $n\rightarrow+\infty$. Assuming $\boldsymbol{\Theta}=[h_0,h]^d$ and $\boldsymbol{\Theta}^\ast=[h_0^\ast,h^\ast]^d$ for $h_0,h_0^\ast,h$ and $h^\ast$ such that $0<h_0\leq h_0^\ast<h^\ast\leq h<+\infty$, the next proposition establishes a contraction rate of $p_{\hat G_n^{\text{\tiny{[B]}}}}$ toward $p_{G^\ast}$ under $G^{\ast}$.

\begin{prp}\label{pcrdp_d}
Consider the Bayesian model \eqref{eq:bayes_model_d}, where $G\sim\operatorname{DP}(1,H)$ and $H$ is a probability measure supported on $[h_{0},h]^d$ with a continuous and strictly positive density. Assume that $\boldsymbol{Y}_{1:n}\simiid p_{G^\ast}$, where $G^{\ast}$ is supported on $[h_0^\ast,h^\ast]^d$ with $0<h_0\leq h_0^\ast<h^\ast\leq h<+\infty$. If \begin{equation}\label{rate_dp_d} \varepsilon_{n,d}^{\text{\tiny{[B]}}}:= \frac{(\log n)^{(d+1)/2}}{\sqrt n}, \end{equation} then there exist constants $C,c>0$, depending only on $(d,h_0,h,h_0^\ast,h^\ast,G^\ast,H)$, such that, for all sufficiently large $n$,
\begin{equation}\label{pcr_dp_d_1}
\P_{G^\ast}^n\!\left[\Pi\!\left(G\in\mathcal P([h_{0},h]^d):d_{\mathrm H}(p_G,p_{G^\ast})\ge C\varepsilon_{n,d}^{\text{\tiny{[B]}}}\,\middle|\,\boldsymbol Y_{1:n}\right)\le\exp\!\left\{-c(\log n)^{d+1}\right\}\right]\ge 1-\frac{1}{n}.
\end{equation}
Moreover, if
\begin{displaymath}
\bar p_n(\boldsymbol y):=\int_{\mathcal P([h_{0},h]^d)}p_G(\boldsymbol y)\,\Pi(dG\mid \boldsymbol Y_{1:n}),\qquad\boldsymbol y\in\mathbb N_0^d,
\end{displaymath}
then there exists a constant $C'>0$, depending only on $(d,h_0,h,h_0^\ast,h^\ast,G^\ast,H)$, such that, for all sufficiently large $n$,
\begin{equation}\label{pcr_dp_d_2}
\P_{G^\ast}^n\!\left[d_{\mathrm H}(\bar p_n,p_{G^\ast})\ge C'\varepsilon_{n,d}^{\text{\tiny{[B]}}}\right]\le\frac{1}{n}.
\end{equation}
\end{prp}

%\begin{center}
%\vdots
%\end{center}
%
%\begin{remark}
%In dimension $d=1$, the condition $G^\ast([3b,h^\ast])>4\alpha_0$ follows automatically from $G^\ast\neq\delta_0$ by choosing $b>0$ small enough. In dimension $d>1$, this is no longer automatic: $G^\ast$ may concentrate on a lower-dimensional face of $[0,h^\ast]^d$---for example, on the hyperplane $\{\theta_1=0\}\cap[0,h^\ast]^d$---while still being nondegenerate in the remaining coordinates. On such a face, the lower bound
%\begin{displaymath}
%p_G(\boldsymbol y)\ge \alpha e^{-dh}\frac{b^{|\boldsymbol y|}}{\boldsymbol y!}
%\end{displaymath}
%used in Lemma~\ref{lem:Poisson_WS_d} would fail, because no mass of $G$ is guaranteed to lie in the interior $[b,h]^d$. The explicit assumption
%\begin{displaymath}
%G^\ast([3b,h^\ast]^d)>4\alpha_0
%\end{displaymath}
%rules out this boundary concentration and ensures that the ratio bound
%\begin{displaymath}
%\frac{p_{G^\ast}(\boldsymbol y)}{p_G(\boldsymbol y)}\le C_0 r^{|\boldsymbol y|}
%\end{displaymath}
%holds uniformly. Extending the result to mixing distributions $G^\ast$ supported on lower-dimensional faces would require a separate analysis tailored to the geometry of the
%corresponding face.
%\end{remark}
%
%\begin{center}
%\vdots
%\end{center}

\subsubsection{Quasi-Bayes empirical Bayes}

The quasi-Bayesian $g$-modeling strategy of \citet{Fav(24)} estimates the mixing distribution $G$ in the multidimensional Robbins' formula \eqref{eq:robbins_d_vector} via Newton's algorithm \citep{New(98)}. Specifically, let $G_0$ be a probability measure on $\boldsymbol\Theta$, and let $\boldsymbol\theta_1\sim G_0$. The $\boldsymbol{Y}_{n}$'s are then modeled as follows:
\begin{equation}\label{eq:qbayes_model_d}
\begin{array}{r@{\;}c@{\;}l@{\quad}c@{\quad}l}
Y_{n,\ell}& \mid &\boldsymbol\theta_n& \simind &\mathrm{Poisson}(\cdot\mid \theta_{n,\ell}),\qquad \ell=1,\ldots,d,\quad n\geq 1,\\[0.2cm]
\boldsymbol\theta_{n+1}& \mid &\boldsymbol Y_{1:n}& \sim &G_n,
\end{array}
\end{equation}
where $G_n$ is defined recursively through Newton's algorithm
\begin{equation}\label{eq:newton_d}
G_n(d\boldsymbol\theta)=(1-\alpha_n)G_{n-1}(d\boldsymbol\theta)+\alpha_n\frac{\prod_{\ell=1}^{d}\mathrm{Poisson}(Y_{n,\ell}\mid\theta_{\ell})G_{n-1}(d\boldsymbol\theta)
}{\int_{\boldsymbol\Theta}\prod_{\ell=1}^{d}\mathrm{Poisson}(Y_{n,\ell}\mid\theta_{\ell})G_{n-1}(d\boldsymbol\theta)},
\end{equation}
with the $\alpha_{n}$'s in $(0,1)$ such that  $\sum_{n\geq1}\alpha_{n}=+\infty$ and $\sum_{n\geq1}\alpha^{2}_{n}<+\infty$. According to \eqref{eq:newton_d}, after observing $\boldsymbol Y_{n+1}$, the model updates $G_n$ by taking a weighted average of $G_n$ and its posterior distribution based on $\boldsymbol Y_{n+1}$, with weight $\alpha_{n+1}$. A standard choice is \(\alpha_n=(\alpha+n)^{-\gamma}\), with \(\alpha>0\) and \(\gamma\in(1/2,1]\).

Under the quasi-Bayesian model \eqref{eq:qbayes_model_d}, an estimate of $G$ is given by $\hat G_{\gamma,n}^{\text{\tiny{[Q-B]}}}=G_n$. The quasi-Bayes estimate of $\boldsymbol\theta_i$, $i=1,\ldots,n$, is then obtained by plugging \(\hat G_{\gamma,n}^{\text{\tiny{[Q-B]}}}\) into the multidimensional Robbins formula \eqref{eq:robbins_d_vector}, i.e.,
\begin{equation}\label{eq:qbayes_eb_est_d}
\hat{\boldsymbol\theta}^{\text{\tiny{[Q-B]}}}_{\gamma,n}(\boldsymbol y):=\hat{\boldsymbol\theta}_{\hat G_{\gamma,n}^{\text{\tiny{[Q-B]}}}}(\boldsymbol y),\qquad\boldsymbol y\in\mathbb N_0^d.
\end{equation}
Equivalently, the $\ell$-th coordinate of \eqref{eq:qbayes_eb_est_d} is
\begin{displaymath}
\hat\theta^{\text{\tiny{[Q-B]}}}_{\gamma,n,\ell}(\boldsymbol y)=(y_\ell+1)\frac{p_{\hat G_{\gamma,n}^{\text{\tiny{[Q-B]}}}}(\boldsymbol y+\boldsymbol e_\ell)}{p_{\hat G_{\gamma,n}^{\text{\tiny{[Q-B]}}}}(\boldsymbol y)},\qquad \ell=1,\ldots,d.
\end{displaymath}
Since \(G\mapsto\hat{\boldsymbol\theta}_G(\boldsymbol y)\) is nonlinear, the estimate \eqref{eq:qbayes_eb_est_d} does not, in general, coincide with the posterior mean of \(\hat{\boldsymbol\theta}_G(\boldsymbol y)\).

A frequentist validation of the quasi-Bayes estimate \(\hat{\boldsymbol\theta}^{\text{\tiny{[Q-B]}}}_{n}\), analogue to Proposition~\ref{pcrdp_d}, follows by a straightforward adaptation to the $d$-dimensional Poisson kernel of the arguments developed in  \citet[{Theorem 4.8 and }Corollary 4.10]{MarTok(09)}.  We state the result in the form needed here, with $\Theta=[{h_0},h]^d$ and $\Theta^{\ast}=[{h_0^\ast},h^{\ast}]^d$ for some ${h_0,\;h_0^\ast,h}$ and $h^{\ast}$ such that  $0<h_0\leq h_0^*<h^\ast\leq h<+\infty$.

\begin{prp}\label{pcrnew_d}
Consider the quasi-Bayesian model \eqref{eq:qbayes_model_d}-\eqref{eq:newton_d}, where: i) $G_0$ be a probability measure supported on the set $[h_0,h]^d$, with a strictly positive density; ii) for every $n\geq 1$, the learning rate is $\alpha_n\asymp n^{-\gamma}$ with $\gamma\in(2/3,1]$. Assume that $\boldsymbol{Y}_{1:n}\simiid p_{G^\ast}$, where $G^{\ast}$  is absolutely continuous with respect to the Lebesgue measure and supported on $[h_0^\ast,h^\ast]^d$, with $0<h_0\leq h_0^\ast<h^\ast\leq h<+\infty$. If
\begin{displaymath}
\varepsilon_{\gamma,n}^{\text{\tiny{[Q-B]}}}:=\left\{
\begin{array}{ll}
\dfrac{1}{\sqrt{n^{1-\gamma}}} &\gamma \in (2/3,1)\\
\\
\dfrac{1}{\sqrt{\log n}}&\gamma=1
\end{array}
\right.
\end{displaymath}
then, for every $\delta>0$
\begin{equation}\label{pcr_newton}
\P_{G^\ast}^n\!\left[d_{\mathrm H}(p_{\hat{G}_{\gamma,n}^{\text{\tiny{[Q-B]}}}},p_{G^\ast})\ge \delta\, \varepsilon_{\gamma,n}^{\text{\tiny{[Q-B]}}}\right]\le\delta
\end{equation}
for all sufficiently large $n$.
\end{prp}

The proof is in Section \ref{proof:pcrnew_d}. The interpretation of the rate in Proposition~\ref{pcrnew_d} is the same as in the $1$-dimensional case discussed after Proposition~\ref{pcrnew_1d}. The rate is driven by the decay of the learning rate $\alpha_n\asymp n^{-\gamma}$, and larger values of $\gamma$ lead to slower concentration. For fixed dimension $d$, this rate is slower than the Bayesian rate in Proposition~\ref{pcrdp_d}, since $\varepsilon_{n,d}^{\text{\tiny{[B]}}} = o(\varepsilon_{\gamma,n}^{\text{\tiny{[Q-B]}}})$ for every $\gamma\in(2/3,1]$.

\subsubsection{Merging of Bayes and quasi-Bayes empirical Bayes}\label{sec:merging_d}

We now turn Propositions~\ref{pcrdp_d} and~\ref{pcrnew_d} into regret bounds for the plug-in Bayes and quasi-Bayes empirical Bayes estimates. Since Robbins' formula depends on the mixing distribution $G$ only through the marginal probability mass function $p_{G}$, the Hellinger concentration rates obtained above can be used to control the regret incurred by replacing the oracle prior $G^\ast$ with an estimated mixing distribution. The next lemma provides this control for the Bayesian and quasi-Bayesian estimates.

\begin{lem}\label{lemma:Janad}
Let $\boldsymbol Y_{1:n}\simiid p_{G^\ast}$, where $G^\ast$ is absolutely continuous with respect to the Lebesgue measure and supported on $[h_0^\ast,h^\ast]^d$, with $0<h_0^\ast<h^\ast<+\infty$.  
\begin{itemize}
    \item[a)]  Consider the Bayesian model \eqref{eq:bayes_model_d}, where  $G\sim\operatorname{DP}(1,H)$ and $H$ is a probability measure on the set $[h_{0},h]^d$, for $0<h_{0}\leq h_{0}^{\ast}<h^{\ast}\leq h<+\infty$, with a continuous and strictly positive density. Then 
    %there exist constants $C_1$ and $C_2$, depending only on $h$ and $h^\ast$, such that for sufficiently large $n$
\begin{align*}
\mathrm{Regret}(\hat G_{n}^{\text{\tiny{[B]}}},G^*)=O_{\mathbb P^\ast}\left(\frac{(\log n)^{d+2}}{n\log\log n}\right)
%  &\leq C_1\frac{\log n}{\log\log n} d_H^2(p_{\hat G_n^{\tiny{[B]}}},p_{G^*})+C_2\frac{2}{n^5}.
\end{align*}
\item[b)] Consider the quasi-Bayesian model \eqref{eq:qbayes_model_d}-\eqref{eq:newton_d}, where: i)  $G_0$ is a probability measure on the set $[h_0,h]^d$, for  $0<h_0\leq h_0^\ast< h^\ast\leq h<+\infty$, with a strictly positive density; ii) for every $n\geq1$, the learning rate is $\alpha_n\asymp n^{-\gamma}$ for some $\gamma\in (2/3,1]$. Then,
\begin{align}\label{eq:regret_rate_qb_d}
\mathrm{Regret}(\hat G_{\gamma,n}^{\text{\tiny{[Q-B]}}},G^*)  =
\left\{
\begin{array}{ll}
o_{\,\mathbb{P}^{\ast}}\left(\dfrac{\log n}{n^{1-\gamma}\log\log n}\right)&\gamma <1\\
\\
o_{\,\mathbb{P}^{\ast}}\left(\dfrac{1}{\log\log n}\right)&\gamma =1.
\end{array}
\right.
%&\leq C_1\frac{\log n}{\log\log n} d_H^2(p_{\hat G_n^{\tiny{[B]}}},p_{G^*})+C_2\frac{2}{n^5}.
\end{align}
\end{itemize}
\end{lem}
\begin{proof} By the multivariate extension of \citet[Lemma 4]{Jan(24)} in Lemma \ref{lem:regret-density-d}, for any $\hat G$ satisfying   $\hat G[0,h]^d=1$ with $0<h^\ast<h<+\infty$ and for any $K\geq 1$,
\[
\operatorname{Regret}(\widehat G,G^\ast)
\leq
d\Big\{
6(h^2+(h^\ast)^2)+24(h+h^\ast)K
\Big\}
d_H^2(p_{\widehat G},p_{G^\ast})
+
(h+h^\ast)^2
\sum_{\ell=1}^d
\sum_{\boldsymbol y\in\mathbb N_0^d:\,y_\ell\geq K}
p_{G^\ast}(\boldsymbol y).
\]
Furthermore, by \citet[Lemma 11]{Jan(24)}, if 
$$
K=\min\left\{\left\lceil \frac{5(he^2+2)\log n}{\log\log n}\right\rceil,he^2+5\log n\right\},
$$
then, for any $n\geq 3$,  $\sum_{\boldsymbol y\in\mathbb N_0^d:\,y_\ell\geq K}
p_{G^\ast}(\boldsymbol y)\leq \frac{2}{n^5}$. Hence, for any estimate $\hat G_n$ of the mixing distribution $G$,  there exist constants $C_1$ and $C_2$, depending only on$d$,  $h$ and $h^\ast$, such that for sufficiently large $n$ it holds
\begin{align*}
  \text{Regret}(\hat G_n,G^\ast)&\leq C_1\frac{\log n}{\log\log n} d_H^2(p_{\hat G_n},p_{G^*})+C_2\frac{2}{n^5}.
\end{align*}
Claim (a) follows from the above inequality applied to the estimate
$\hat G_n^{\text{\tiny{[B]}}}$, together with Proposition
\ref{pcrdp_d}. Claim (b) follows analogously, by applying the same
inequality to the estimate $\hat G_{\gamma,n}^{\text{\tiny{[Q-B]}}}$ and using
Proposition \ref{pcrnew_d}.
\end{proof}

The next theorem establishes the frequentist merging of the Bayesian and quasi-Bayesian empirical Bayes strategies. In particular, under the ``true'' mixing distribution $G^{\ast}$ on $\boldsymbol{\Theta}^{\ast}\subseteq\mathbb R_+^d$, we define the regret incurred by using the quasi-Bayes estimate $\hat{\boldsymbol\theta}^{\text{\tiny{[Q-B]}}}_{\gamma,n}$ in place of  the Bayes estimate $\hat{\boldsymbol\theta}^{\text{\tiny{[B]}}}_{n}$ as
\begin{equation}\label{merging_regret_d}
\text{Regret}(\hat G_{\gamma,n}^{\text{\tiny{[Q-B]}}},\hat G_{n}^{\text{\tiny{[B]}}};G^{\ast})=\sum_{\boldsymbol y\in\mathbb N_0^d} ||\hat{\boldsymbol \theta}_{n}^{\text{\tiny{[B]}}}(\boldsymbol y)-\hat{\boldsymbol \theta}_{\gamma,n}^{\text{\tiny{[Q-B]}}}(\boldsymbol y)||^2p_{G^*}(\boldsymbol y),
\end{equation}
and show that, as $n\rightarrow+\infty$, it vanishes at the same rate as that of $\text{Regret}(\hat G_{\gamma,n}^{\text{\tiny{[Q-B]}}},G^{\ast})$, displayed in \eqref{eq:regret_rate_qb_d}.

\begin{thm}\label{merge_d}
Let $\boldsymbol Y_{1:n}\simiid p_{G^\ast}$, where $G^\ast$ is absolutely continuous with respect to the Lebesgue measure and on $[h_0^\ast,h^\ast]^d$, with  $0<h_0^\ast<h^\ast<+\infty$. Consider the Bayesian model \eqref{eq:bayes_model_d}, where  $G\sim\operatorname{DP}(1,H)$ and $H$ is a probability measure on the set $[h_{0},h]^d$, for $0<h_{0}\leq h_{0}^{\ast}<h^{\ast}\leq h<+\infty$, with a continuous and strictly positive density. Further, consider the quasi-Bayesian model \eqref{eq:qbayes_model_d}-\eqref{eq:newton_d}, where: i)  $G_0$ is a probability measure supported on the set $[h_0,h]^d$, for  $0<h_0\leq h_0^\ast< h^\ast\leq h<+\infty$, with a strictly positive density; ii) for every $n\geq1$, the learning rate is $\alpha_n\asymp n^{-\gamma}$ for some $\gamma\in (2/3,1]$. 
%Let $G\sim \Pi= \operatorname{DP}(1,\alpha)$, where $\alpha$ is a probability measure on $[0,h]^d$ with a continuous density bounded away from zero and infinity and $h^\ast<h<+\infty$, and  let $\hat G_n^{\text{\tiny{[B]}}}(\cdot)=\int_{\mathcal P(\Theta)}G(\cdot)\,\Pi(dG\mid Y_{1:n})$. Let $G_0$ be a probability measure on $[h_0,h]^d$ with strictly positive density and $0<h_0\leq h_0^\ast$, and for every $n\geq 1$, let $\hat{G}_{\gamma,n}^{\text{\tiny{[Q-B]}}}=G_n$, with $(G_n)$ defined recursively as \eqref{eq:newton} with $\alpha_n\asymp n^{-\gamma}$ for some $\gamma\in (2/3,1]$. 
\begin{itemize}
    \item[a)] The estimates  $p_{n}^{\text{\tiny{[B]}}}$ and $p_{n}^{\text{\tiny{[Q-B]}}}$ of $p_{G^\ast}$ merge in Hellinger distance, as $n\rightarrow\infty$, and
    \begin{displaymath}
d_H(p_{n}^{\text{\tiny{[B]}}},p_{\gamma,n}^{\text{\tiny{[Q-B]}}})=\left\{
\begin{array}{ll}
o_{\,\mathbb{P}^{\ast}}\left(\dfrac {1}{\sqrt{n^{1-\gamma}}}\right)&\gamma<1\\
\\
o_{\,\mathbb{P}^{\ast}}\left(\dfrac{1}{\sqrt{\log n}}\right)&\gamma=1,
\end{array}\right.
\end{displaymath}
\item[b)] The plug-in estimates $\hat{\boldsymbol \theta}_{n}^{\text{\tiny{[B]}}}$ and $\hat{\boldsymbol\theta}_{\gamma,n}^{\text{\tiny{[Q-B]}}}$ merge in $L^2(p_{G^*})$, as $n\rightarrow\infty$ in $L^2(p_{G^*})$, and
\begin{displaymath}
\mathrm{Regret}(\hat G_{\gamma,n}^{\text{\tiny{[Q-B]}}},\hat G_{n}^{\text{\tiny{[B]}}};G^{\ast})
=
\left\{
\begin{array}{ll}
o_{\,\mathbb{P}^{\ast}}\left(\dfrac{\log n}{n^{1-\gamma}\log\log n}\right)&\gamma <1\\
\\
o_{\,\mathbb{P}^{\ast}}\left(\dfrac{1}{\log\log n}\right)&\gamma =1.
\end{array}
\right.
\end{displaymath}
\end{itemize}
\end{thm}
\begin{proof} 
With regards to a), by Proposition \ref{pcrdp_d}, for every $\gamma\in (2/3,1]$, $d_H(p_n^{\text{\tiny{[B]}}},p_{G^\ast})=o_{\PP^\ast}(\varepsilon_{\gamma,n}^{\text{\tiny{[Q-B]}}})$  with $\varepsilon_{\gamma,n}^{\text{\tiny{[Q-B]}}}$ as in \eqref{rate_new}. The claim follows by Proposition \ref{pcrnew_d} and the triangular inequality. With regards to b),
%By Lemma 11 in \cite{Jan(24)}, if $$K=\min\left\{\left\lceil \frac{5(he^2+2)\log n}{\log\log n}\right\rceil,he^2+5\log n\right\},$$then, for any $n\geq 3$,  $\sum_{y>K} p_{G^\ast}(y)\leq \frac{2}{n^5}$.Applying Lemma \ref{lemma:Jana} to $\hat G_n^{\text{\tiny{[B]}}}$ and to $\hat G_n^{\text{\tiny{[Q-B]}}}$, we can ensure the existence of a constants $C_1$ and $C_2$, depending only on $h$ and $h^\ast$, such that for sufficiently large $n$\begin{align*}    \sum_{y\in\mathbb{N}_{0}}(\hat\theta_{n}^{\text{\tiny{[B]}}}(y)-\hat\theta^{*}(y))^2p_{G^*}(y)&\leq C_1\frac{\log n}{\log\log n} d_H^2(p_{\hat G_n^{\tiny{[B]}}},p_{G^*})+C_2\frac{2}{n^5}\end{align*} and \begin{align*}     \sum_{y\in\mathbb{N}_{0}}(\hat\theta_{n}^{\text{\tiny{[Q-B]}}}(y)-\hat\theta^{*}(y))^2p_{G^*}(y)&\leq C_1\frac{\log n}{\log\log n} d_H^2(p_{\hat G_n^{\tiny{[B]}}},p_{G^*})+C_2\frac{2}{n^5}. \end{align*} Thus,
 \begin{align*}
     \text{Regret}(\hat G_{\gamma,n}^{\text{\tiny{[Q-B]}}},\hat G_{n}^{\text{\tiny{[B]}}};G^{\ast})
      &\leq 2\text{Regret}(\hat G_{n}^{\text{\tiny{[B]}}},G^\ast)+2\;\text{Regret}(\hat G_{n}^{\text{\tiny{[Q-B]}}},G^\ast)\\
      %&\quad\quad  \leq 2 C_1\frac{\log n}{\log\log n}      \left\{ d_H^2(p_{\hat G_n^{\tiny{[B]}}},p_{G^*})+d_H^2(p_{\hat G_{\gamma,n}^{\tiny{[Q-B]}}},p_{G^*})\right\}+2C_2\frac{2}{n^5}\\
  &=
      \left\{
\begin{array}{ll}
o_{\,\mathbb{P}^{\ast}}\left(\dfrac{\log n}{n^{1-\gamma}\log\log n}\right)&\gamma <1\\
\\
o_{\,\mathbb{P}^{\ast}}\left(\dfrac{1}{\log\log n}\right)&\gamma =1,
\end{array}
\right.
\end{align*}
where the last equality comes from Lemma \ref{lemma:Janad}.
\end{proof}

%%%%%%%%%%%%%%%%%%%%%%%%%%%%%%%%
%%%%%%%%%%%%%%%%%%%%%%%%%%%%%%%%
%%%%%%%%%%%%%%%%%%%%%%%%%%%%%%%%
%%%%%%%%%%%%%%%%%%%%%%%%%%%%%%%%

\section{Synthetic-data illustrations}\label{sec3}

\subsection{The $1$-dimensional setting}

For $n\in\{50,\,100,\,200,\,400,\,1,000,\,2,000,\,4,000,\,8,000\}$, generate i.i.d. data $Y_{1:n}=(Y_{1},\ldots,Y_{n})$ from a Poisson mixture model with Weibull prior $G$ of scale parameter $5$ and shape parameter $3$. We compare quasi-Bayes estimate $\hat{\theta}^{\text{\tiny{[Q-B]}}}_{n}$ and the Bayes estimate $\hat{\theta}^{\text{\tiny{[B]}}}_{n}$ with the corresponding oracle Bayes estimate $\hat{\theta}^{\ast}$. The oracle $\hat{\theta}^{\ast}$ is obtained from \eqref{eq:oracle} with $G^{\ast}$ being the Weibull prior distribution that generates the $\theta_{i}$'s, and evaluating the marginal likelihood $p_{G^\ast}$ numerically through the trapezoidal rule.

As a measure of accuracy of the plug-in empirical Bayes estimates, we consider the empirical mean squared error (E-mse). For $n\in\mathbb{N}$, let $(\theta_1,\ldots,\theta_n)$ be the values generated from the Weibull prior distribution, and let $\hat{\theta}_{n}$ be the plug-in empirical Bayes estimate, as defined in \eqref{eq:estimate}. The E-mse is defined as
\begin{displaymath}
\mathrm{E\text{-}mse}(\hat{G}_{n})=\frac{1}{n}\sum_{i=1}^n \bigl(\hat\theta_{n}(y_i)-\theta_i\bigr)^2.
\end{displaymath}
For the oracle Bayes estimate $\hat{\theta}^{\ast}$, the E-mse is referred to as the empirical minimum mean squared error (E-mmse), i.e.
\begin{displaymath}
\mathrm{E\text{-}mmse}=\frac{1}{n}\sum_{i=1}^n \bigl(\hat\theta^{\ast}(y_i)-\theta_i\bigr)^2.
\end{displaymath}
Finally, we define the empirical regret (E-regret) as $\mathrm{E\text{-}regret}(\hat{G}_{n})=\mathrm{E\text{-}mse}(\hat{G}_{n})-\mathrm{E\text{-}mmse}$, namely the excess empirical squared error of $\hat{\theta}_{n}$ with respect to the benchmark $\hat{\theta}^{\ast}$; see \citet{Efr(14),Efr(19)} for details.

With regards to the quasi-Bayes estimate $\hat{\theta}^{\text{\tiny{[Q-B]}}}_{n}$,  Newton's algorithm \eqref{eq:newton} requires the numerical evaluation of an integral, which we approximate by the trapezoidal rule. To perform this evaluation, the density function of the mixing distribution $G_{n}$ is represented through its values on a fixed uniform grid of $d\in\{5,000,\,1,000,\,500,\,100,\,50,\,10\}$ quadrature points over $\Theta=(0,U_{\Theta})$, where $U_\Theta=\max\{\max\{Y_{1:n}\},\lceil Q_{n,0.99}+4\sqrt{\max\{Q_{n,0.99},1\}}\rceil\}$, with $Q_{n,0.99}=\text{Quantile}(Y_{1:n};0.99)$. This representation is used only for numerical evaluation and imposes no modeling restriction on $\Theta$. Further, we set the initial guess $G_{0}$ to be Uniform over $\Theta$, and take the learning rate to be $\alpha_{n}=(1+n)^{-0.99}$. Table \ref{weib_tab_sens} reports the $\mathrm{E\text{-}mse}(\hat{G}_{\gamma,n}^{\text{\tiny{[Q-B]}}})$ and $\mathrm{E\text{-}regret}(\hat{G}_{\gamma,n}^{\text{\tiny{[Q-B]}}})$ as the sample size $n$ and the grid resolution $d$ vary.

\begin{table}[ht]
\centering
\caption{\footnotesize{Weibull prior: $\mathrm{E\text{-}mse}(\hat{G}_{\gamma,n}^{\text{\tiny{[Q-B]}}})$ and $\mathrm{E\text{-}regret}(G_{\gamma,n}^{\text{\tiny{[Q-B]}}})$ as $n$ and $d$ vary.}}
{
\setlength{\tabcolsep}{0pt}
\begin{tabular}{@{}l@{\hspace{0.5cm}}*{6}{>{\centering\arraybackslash}p{2.1cm}}@{}}
\hline
\hline
 & $d=5{,}000$ & $d=1{,}000$ & $d=500$ & $d=100$ & $d=50$ & $d=10$ \\[0.1cm]
\hline
\multicolumn{7}{@{}l}{\underline{$n=50$}} \\[0.05cm]
$\mathrm{E\text{-}mse}(\hat{G}_{\gamma,n}^{\text{\tiny{[Q-B]}}})$    & 2.485 & 2.485 & 2.485 & 2.484 & 2.483 & 2.472 \\
$\mathrm{E\text{-}regret}(G_{\gamma,n}^{\text{\tiny{[Q-B]}}})$ & 0.100 & 0.100 & 0.100 & 0.100 & 0.099 & 0.088 \\[0.4cm]

\multicolumn{7}{@{}l}{\underline{$n=100$}} \\[0.05cm]
$\mathrm{E\text{-}mse}(\hat{G}_{\gamma,n}^{\text{\tiny{[Q-B]}}})$  & 1.989 & 1.989 & 1.989 & 1.988 & 1.986 & 1.971 \\
$\mathrm{E\text{-}regret}(G_{\gamma,n}^{\text{\tiny{[Q-B]}}})$ & 0.193 & 0.193 & 0.193 & 0.192 & 0.190 & 0.175 \\[0.4cm]

\multicolumn{7}{@{}l}{\underline{$n=200$}} \\[0.05cm]
$\mathrm{E\text{-}mse}(\hat{G}_{\gamma,n}^{\text{\tiny{[Q-B]}}})$   & 2.024 & 2.024 & 2.024 & 2.023 & 2.023 & 2.063 \\
$\mathrm{E\text{-}regret}(G_{\gamma,n}^{\text{\tiny{[Q-B]}}})$ & -0.043 & -0.043 & -0.043 & -0.043 & -0.043 & -0.003 \\[0.4cm]

\multicolumn{7}{@{}l}{\underline{$n=400$}} \\[0.05cm]
$\mathrm{E\text{-}mse}(\hat{G}_{\gamma,n}^{\text{\tiny{[Q-B]}}})$    & 1.947 & 1.947 & 1.947 & 1.947 & 1.946 & 1.963 \\
$\mathrm{E\text{-}regret}(G_{\gamma,n}^{\text{\tiny{[Q-B]}}})$ & 0.362 & 0.362 & 0.362 & 0.361 & 0.361 & 0.378 \\[0.1cm]
\hline
\hline
\end{tabular}
}
\label{weib_tab_sens}
\end{table}

Table \ref{weib_tab_sens} provides a sensitivity analysis of Newton's algorithm with respect to the number of quadrature points $d\in\{5,000,\,1,000,\,500,\,100,\,50,\,10\}$; in particular, it shows that the empirical performance of Newton's algorithm is robust to the choice of $d$. For the next evaluations of $\hat{\theta}^{\text{\tiny{[Q-B]}}}_{n}$ we set $d=1,000$.

With regards to the Bayes estimate $\hat{\theta}^{\text{\tiny{[B]}}}_{n}$, we adopt Algorithm 8 of \citet{Nea(00)} to evaluate the posterior distribution \eqref{post_dp} of the Dirichlet process mixture model. We set the strength parameter $c=1$, and the base probability measure $H$ to be a Gamma of shape $3$ and scale $(3n)^{-1}\sum_{1\leq i\leq n}Y_{i}$, so that the prior mean of $H$ matches the empirical mean of the data $Y_{1:n}$. We consider Algorithm 8 with $m=5$ auxiliary components; we refer to  \citet[Section 6]{Nea(00)} for details on the specification of $m$.  At each iteration of the MCMC scheme, the cluster or mixture-component allocation of each observation is updated by comparing the likelihood function under the currently occupied clusters with the likelihood function under the auxiliary clusters  drawn from $H$; conditionally on the allocations, the occupied atoms are updated from their Gamma full conditional distributions. We run the MCMC for  $T=5,000$ iterations, discard the first $1,000$ iterations as burn-in, and retain one draw every $5$ iterations. The evaluation of $\hat{\theta}^{\text{\tiny{[B]}}}_{n}$ is done by averaging over the retained posterior draws.

Results are reported in Figure \ref{weib_compare1}-\ref{weib_compare2}-\ref{weib_cpu}. Figure \ref{weib_compare1}-\ref{weib_compare2} display the quasi-Bayes, Bayes and oracle Bayes estimates. Figure \ref{weib_cpu} compares the quasi-Bayes and Bayes estimates in terms of empirical performance and computational cost. Their empirical performances are measured by the corresponding E-regrets, i.e. $\mathrm{E\text{-}regret}(\hat{G}_{\gamma,n}^{\text{\tiny{[Q-B]}}})$ and $\mathrm{E\text{-}regret}(\hat{G}_{n}^{\text{\tiny{[B]}}})$, whereas computational cost is measured by the number of computational units and by CPU time. For the quasi-Bayes estimate, one computational unit is defined as one likelihood evaluation at one quadrature point of the grid used to represent the density of the mixing distribution $G_{n}$; thus, if the grid contains $d$ quadrature points, then the total number of computational units is $nd$. For the Bayes estimate, one computational unit is defined as one likelihood evaluation at one candidate atom of the mixing distribution during the MCMC update of a cluster allocation. If $k_{-i}(t)$ denotes the number of occupied clusters after removing observation $i$ at MCMC iteration $t$, then the total number of computational units is $\sum_{1\leq t\leq T}\sum_{1\leq i\leq n}(k_{-i}(t)+m)$. The CPU time refers to the wall-clock time, in seconds, required to estimate $G$.

Figure \ref{regret_comparison_1d} reports the empirical regret incurred by using the quasi-Bayes estimate $\hat{\theta}^{\text{\tiny{[Q-B]}}}_{\gamma,n}$ in place of the Bayes estimate $\hat{\theta}^{\text{\tiny{[B]}}}_{n}$, i.e.,
\begin{displaymath}
\text{E-regret}(\hat{G}_{\gamma,n}^{\text{\tiny{[Q-B]}}},\hat{G}_{n}^{\text{\tiny{[B]}}};G^{\ast})=\frac{1}{n}\sum_{i=1}^{n}(\hat{\theta}^{\text{\tiny{[Q-B]}}}_{\gamma,n}(y_{i})-\hat{\theta}^{\text{\tiny{[B]}}}_{\gamma,n}(y_{i}))^{2};
\end{displaymath}
 this is the empirical counterpart of \eqref{merging_regret_1d}. Figure \ref{regret_comparison_1d} provides an empirical evidence that $\hat{\theta}^{\text{\tiny{[Q-B]}}}_{\gamma,n}$ and $\hat{\theta}^{\text{\tiny{[B]}}}_{n}$ merge as $n\rightarrow+\infty$.

\begin{figure}
\centering
\includegraphics[width=.75\textwidth,height=.55\textheight,keepaspectratio]{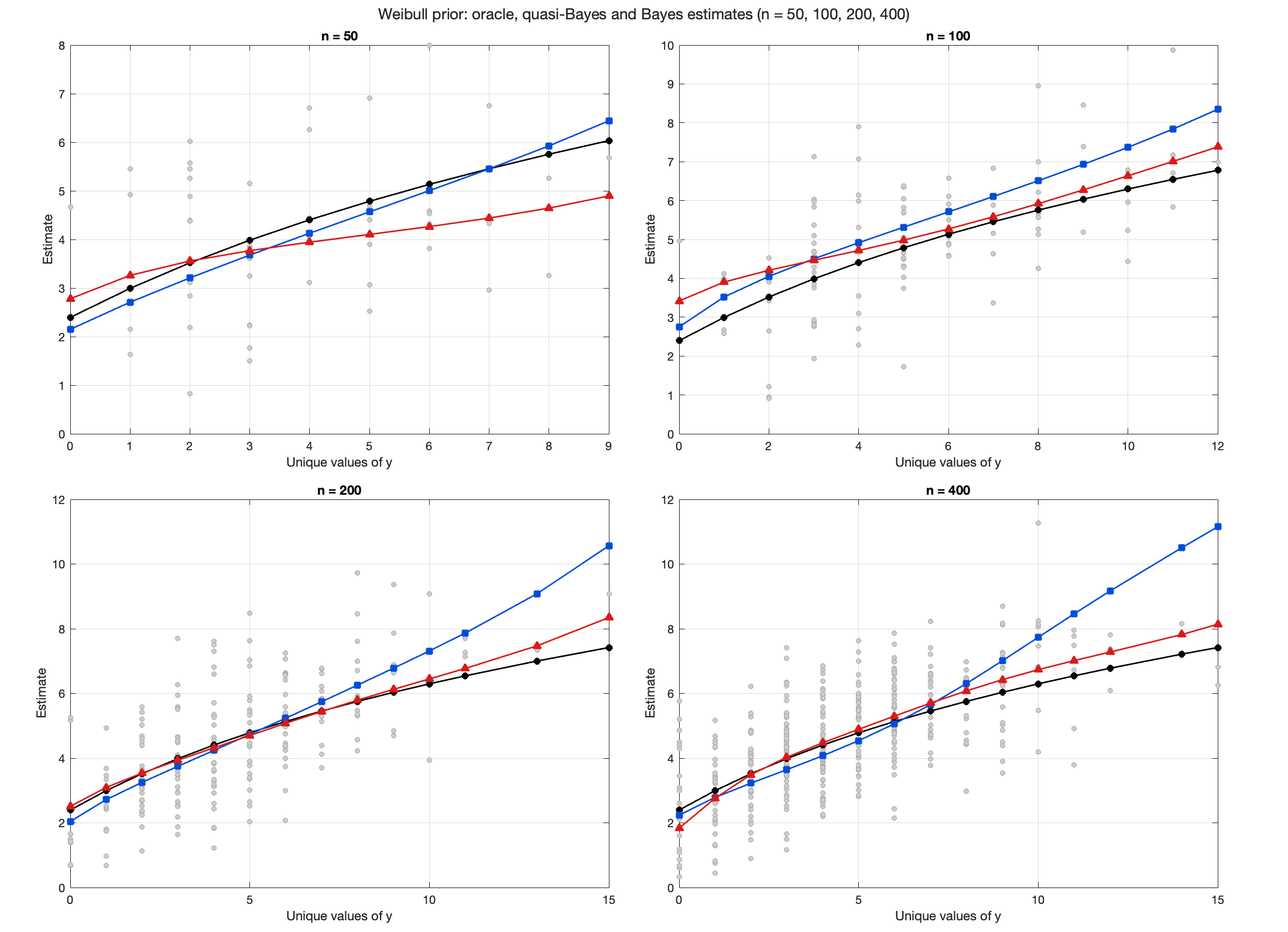}
\caption{\footnotesize{Weibull prior, $n\in\{50,\,100,\,200,\,400\}$: data points plotted against the ``true'' parameters (grey), together with the corresponding oracle Bayes (black), Bayes (red), and quasi-Bayes (blue) estimates.}}
\label{weib_compare1}
\end{figure}

\begin{figure}
\centering
\includegraphics[width=.75\textwidth,height=.55\textheight,keepaspectratio]{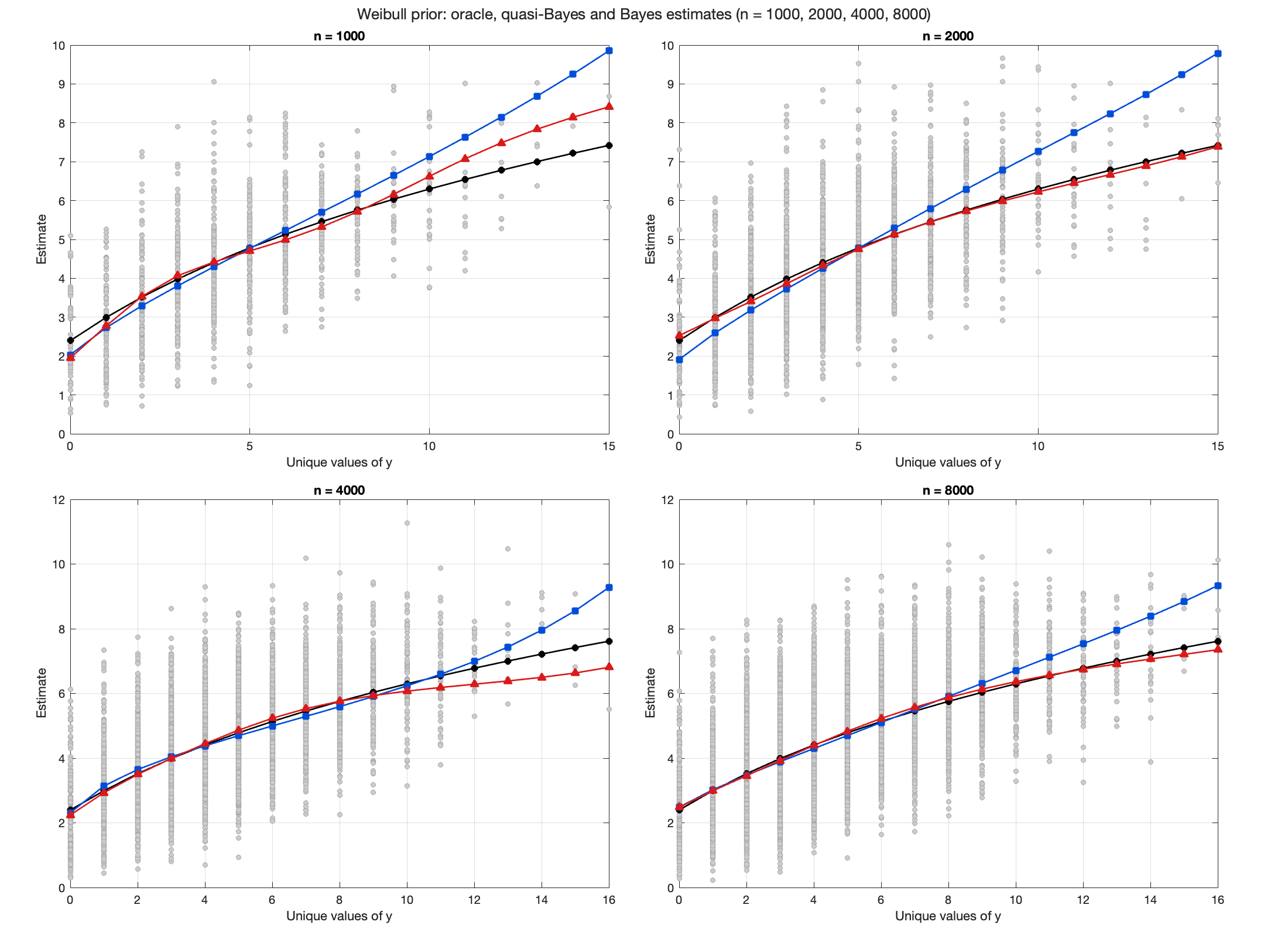}
\caption{\footnotesize{Weibull prior, $n\in\{1,000,\,2,000,\,4,000,\,8,000\}$: data points plotted against the ``true'' parameters (grey), together with the corresponding oracle Bayes (black), Bayes (red), and quasi-Bayes (blue) estimates.}}
\label{weib_compare2}
\end{figure}

\begin{figure}
\centering
\includegraphics[width=.95\textwidth,height=.85\textheight,keepaspectratio]{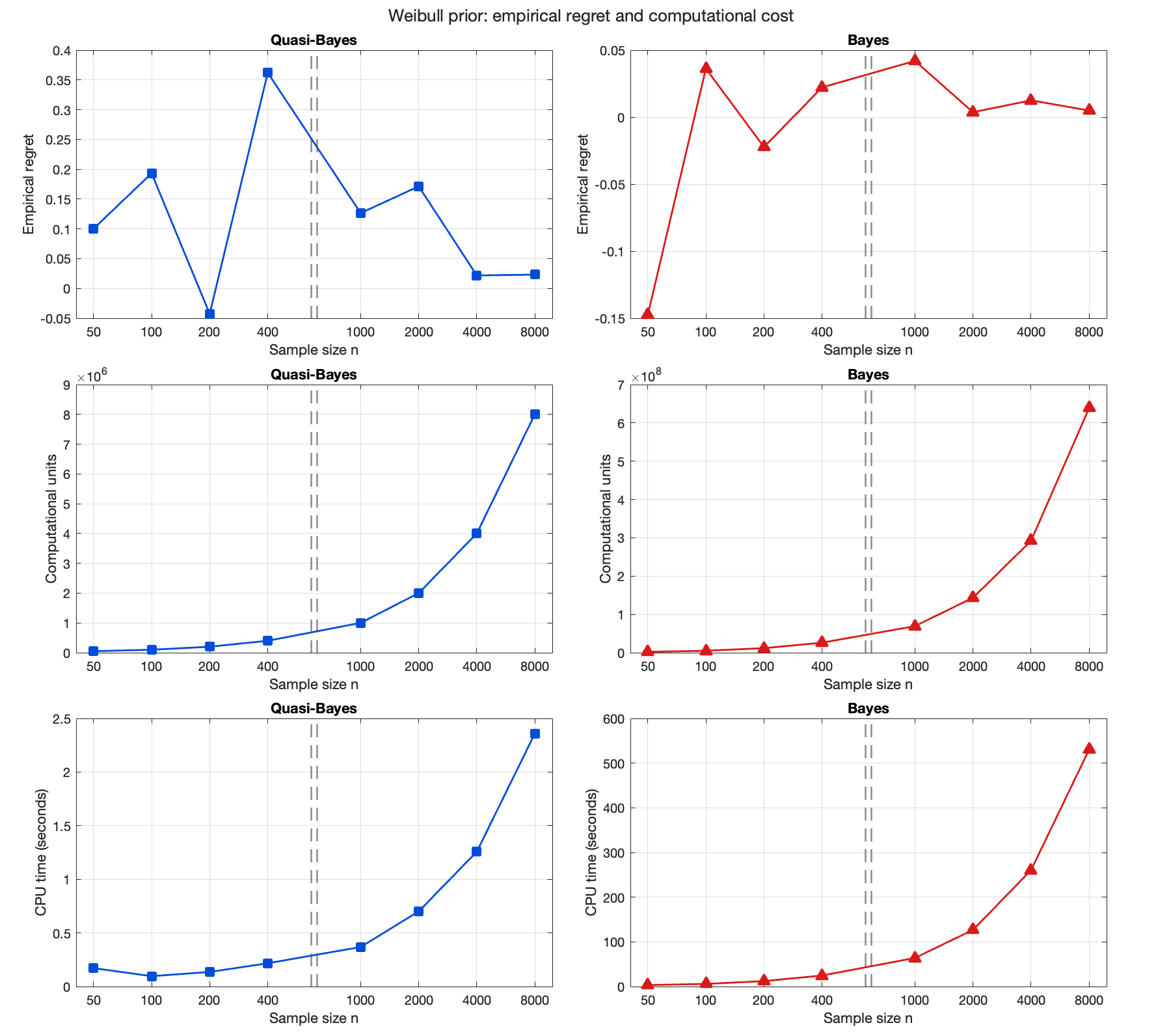}
\caption{\footnotesize{Weibull prior: quasi-Bayes (blue) and Bayes (red) estimates compared by E-regret (top panels), computational units (middle panels), and CPU time (bottom panels).}}
\label{weib_cpu}
\end{figure}

\begin{figure}
\centering
\includegraphics[width=.95\textwidth,height=.85\textheight,keepaspectratio]{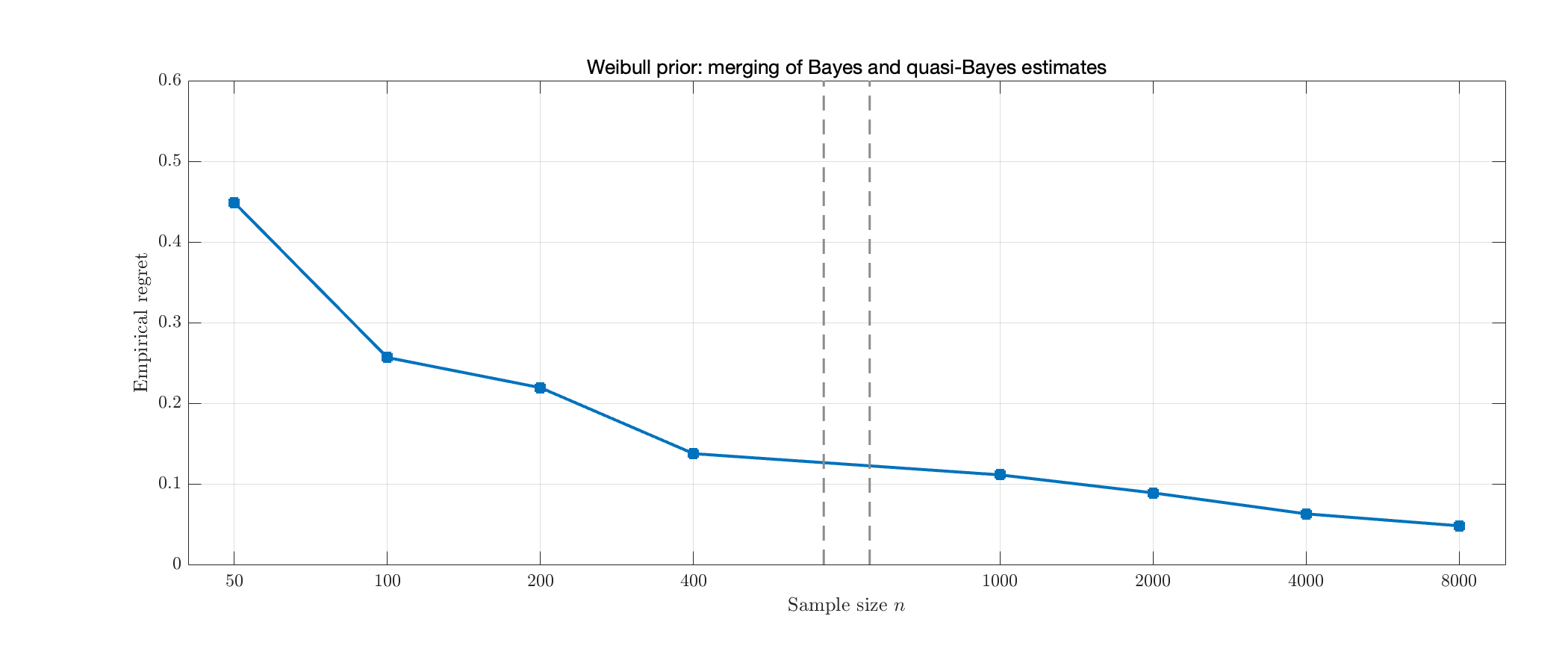}
\caption{\footnotesize{Weibull prior: E-regret incurred by using the quasi-Bayes estimate in place of the Bayes estimate.}}
\label{regret_comparison_1d}
\end{figure}

\subsection{The $d$-dimensional setting, $d=2$}

For $n\in\{50,\,100,\,200,\,400,\,1,000,\,2,000,\,4,000,\,8,000\}$, we generate i.i.d. data $\boldsymbol Y_{1:n}=(\boldsymbol Y_1,\ldots,\boldsymbol Y_n)$, with $\boldsymbol Y_i=(Y_{i,1},Y_{i,2})\in\mathbb N_0^2$ from a $2$-dimensional Poisson mixture model with a product Weibull prior $G=G_1\otimes G_2$, where $G_{\ell}$ is the Weibull distribution of scale parameter $5$ and shape parameter $3$. We compare quasi-Bayes estimate $\hat{\boldsymbol{\theta}}^{\text{\tiny{[Q-B]}}}_{\gamma,n}$ and the Bayes estimate $\hat{\boldsymbol{\theta}}^{\text{\tiny{[B]}}}_{n}$ with the oracle Bayes estimate $\hat{\boldsymbol{\theta}}^{\ast}$. The oracle $\hat{\boldsymbol{\theta}}^{\ast}$ is obtained from \eqref{eq:oracle_d} with $G^{\ast}=G_1^{\ast}\otimes G_2^{\ast}$ being the product Weibull prior distribution that generates the $\boldsymbol{\theta}_{i}$'s, and evaluating the marginal likelihood $p_{G^\ast}$ numerically through the trapezoidal rule.

As a measure of accuracy of the plug-in empirical Bayes estimates, we consider the empirical mean squared error (E-mse). For $n\in\mathbb{N}$, let $(\boldsymbol{\theta}_1,\ldots,\boldsymbol{\theta}_n)$ be the values generated from the Weibull prior distribution, and let $\hat{\boldsymbol{\theta}}_{n}$ be the plug-in empirical Bayes estimate, as defined in \eqref{eq:estimate}. The E-mse is defined as
\begin{displaymath}
\mathrm{E\text{-}mse}(\hat{G}_{n})=\frac{1}{2n}\sum_{i=1}^{n}\sum_{\ell=1}^{2}\left\{\hat\theta_{n,\ell}(\boldsymbol y_i)-\theta_{i,\ell}\right\}^{2}.
\end{displaymath}
For the oracle Bayes estimate $\hat{\boldsymbol{\theta}}^{\ast}$, the E-mse is referred to as the empirical minimum mean squared error (E-mmse), i.e.
\begin{displaymath}
\mathrm{E\text{-}mmse}=\frac{1}{2n}\sum_{i=1}^{n}\sum_{\ell=1}^{2}\left\{\hat{\theta}^{\ast}_{\ell}(\boldsymbol y_i)-\theta_{i,\ell}\right\}^{2}.
\end{displaymath}
Finally, we define the empirical regret (E-regret) as $\mathrm{E\text{-}regret}(\hat{G}_{n})=\mathrm{E\text{-}mse}(\hat{G}_{n})-\mathrm{E\text{-}mmse}$, namely the excess empirical squared error of $\hat{\boldsymbol{\theta}}_{n}$ with respect to the benchmark $\hat{\boldsymbol{\theta}}^{\ast}$; see \citet{Efr(14),Efr(19)} for details.

With regards to the quasi-Bayes estimate $\hat{\boldsymbol\theta}^{\text{\tiny{[Q-B]}}}_{n}$, the $2$-dimensional version of Newton's algorithm \eqref{eq:newton} requires the numerical evaluation of an integral with respect to the mixing distribution, which we approximate by the tensor-product trapezoidal rule. To perform this evaluation, the density function of the mixing distribution $G_n$ is represented through its values on a fixed tensor-product grid over $\boldsymbol{\Theta}=(0,U_{\Theta_{1}})\times(0,U_{\Theta_{2}})$ with $201$ quadrature points per coordinate, yielding $201^2$ grid points in total. For each coordinate $\ell=1,2$, we set $U_{\Theta_{\ell}}=\max\{\max\{Y_{1:n,\ell}\},\lceil Q_{n,0.99}((\ell))+4\sqrt{\max\{Q_{n,0.99}(\ell),1\}}\rceil\}$,  where $Q_{n,0.99}(\ell)=\operatorname{Quantile}(Y_{1:n,\ell};0.99)$. As for the $1$-dimensional setting, this representation is used only for numerical evaluation and imposes no modeling restriction on $\boldsymbol{\Theta}$. Further, we set the initial guess $G_0$ to be Uniform over $\boldsymbol{\Theta}$, and take the learning rate to be $\alpha_n=(1+n)^{-1}$. 

With regards to the Bayes estimate  $\hat{\boldsymbol{\theta}}^{\text{\tiny{[B]}}}_{n}$, we adopt the $2$-dimensional version of Algorithm 8 of \citet{Nea(00)} to evaluate the  posterior distribution \eqref{post_dp_d} of the $2$-dimensional  Dirichlet process mixture model. We set the strength parameter $c=1$, and take the base  probability measure $H$ to be a product Gamma distribution $\boldsymbol{H}=H_1\otimes H_2$, with independent marginals. For each coordinate $\ell=1,2$, $H_\ell$ is a Gamma distribution of shape $2$ and scale $(2n)^{-1}\sum_{1\leq i\leq n}Y_{i,\ell}$ so that the prior mean of $H_\ell$ matches the empirical mean of $Y_{1:n,\ell}$. We consider Algorithm 8 with $m=5$ auxiliary components.  At each iteration of the MCMC scheme, the cluster allocation of each observation $\boldsymbol Y_i$ is updated by comparing the multivariate likelihood under the currently occupied clusters with the multivariate likelihood under the auxiliary clusters drawn from $\boldsymbol{H}$. Conditionally on the allocations, the occupied atoms are updated coordinate-wise from their Gamma full conditional distributions. We run the MCMC for $T=5,000$ iterations, discard the first $1,000$ iterations as burn-in, and retain one draw every $5$ iterations. The evaluation of $\hat{\boldsymbol{\theta}}^{\text{\tiny{[B]}}}_{n}$ is done by averaging over the retained posterior draws.

Results are reported in Figure \ref{weib_compare1_d}-\ref{weib_compare2_d}-\ref{weib_cpu_d}. Figure \ref{weib_compare1_d}-\ref{weib_compare2_d} display the quasi-Bayes, Bayes and oracle Bayes estimates. Figure \ref{weib_cpu_d} compares the quasi-Bayes and Bayes estimates in terms of empirical performance, through E-regret, and computational cost, through the number of computational units and CPU time. Computational units and by CPU time are defined as in the $1$-dimensional setting. Figure \ref{regret_comparison_d} reports the empirical regret incurred by using the quasi-Bayes estimate  $\hat{\boldsymbol\theta}^{\text{\tiny{[Q-B]}}}_{\gamma,n}$ in place of the Bayes estimate $\hat{\boldsymbol\theta}^{\text{\tiny{[B]}}}_{n}$, i.e.
\begin{displaymath}
\text{E-regret}(\hat{G}_{\gamma,n}^{\text{\tiny{[Q-B]}}},\hat{G}_{n}^{\text{\tiny{[B]}}};G^{\ast})=\frac{1}{2n}\sum_{i=1}^{n}\sum_{\ell=1}^{2}(\hat{\theta}_{\gamma,n,l}^{\text{\tiny{[Q-B]}}}(\boldsymbol{y}_{i})-\hat{\theta}_{n,l}^{\text{\tiny{[B]}}}(\boldsymbol{y}_{i}))^{2}
\end{displaymath}
 this is the empirical counterpart of \eqref{merging_regret_d}. Figure \ref{regret_comparison_d} provides an empirical evidence that $\hat{\boldsymbol\theta}^{\text{\tiny{[Q-B]}}}_{\gamma,n}$ and $\hat{\boldsymbol\theta}^{\text{\tiny{[B]}}}_{n}$ merge as $n\rightarrow+\infty$.

\begin{figure}
\centering
\includegraphics[width=.75\textwidth,height=.55\textheight,keepaspectratio]{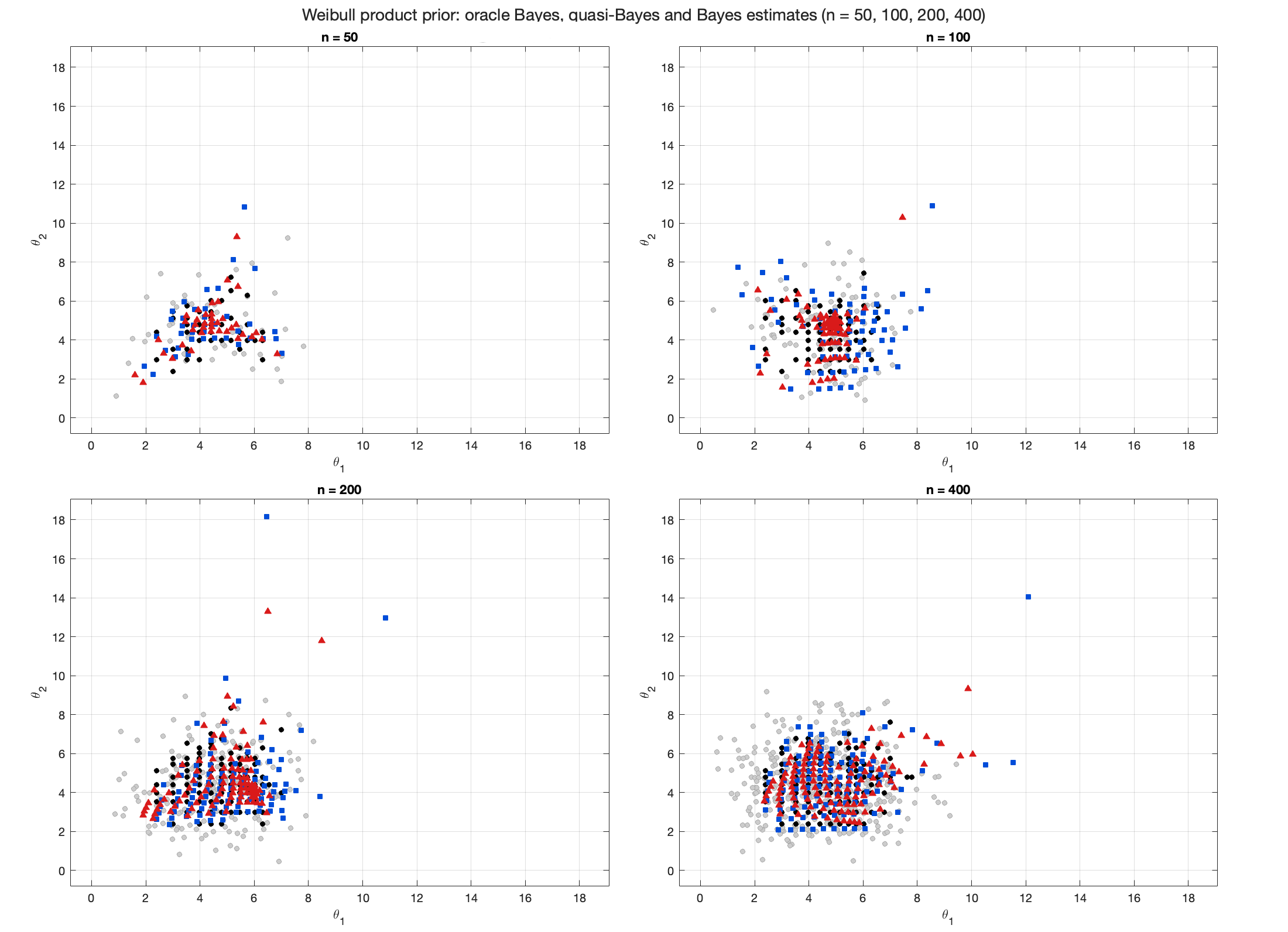}
\caption{\footnotesize{Weibull product prior, $n\in\{50,\,100,\,200,\,400\}$: data points plotted against the ``true'' parameters (grey), together with the corresponding oracle Bayes (black), Bayes (red), and quasi-Bayes (blue) estimates.}}
\label{weib_compare1_d}
\end{figure}

\begin{figure}
\centering
\includegraphics[width=.75\textwidth,height=.55\textheight,keepaspectratio]{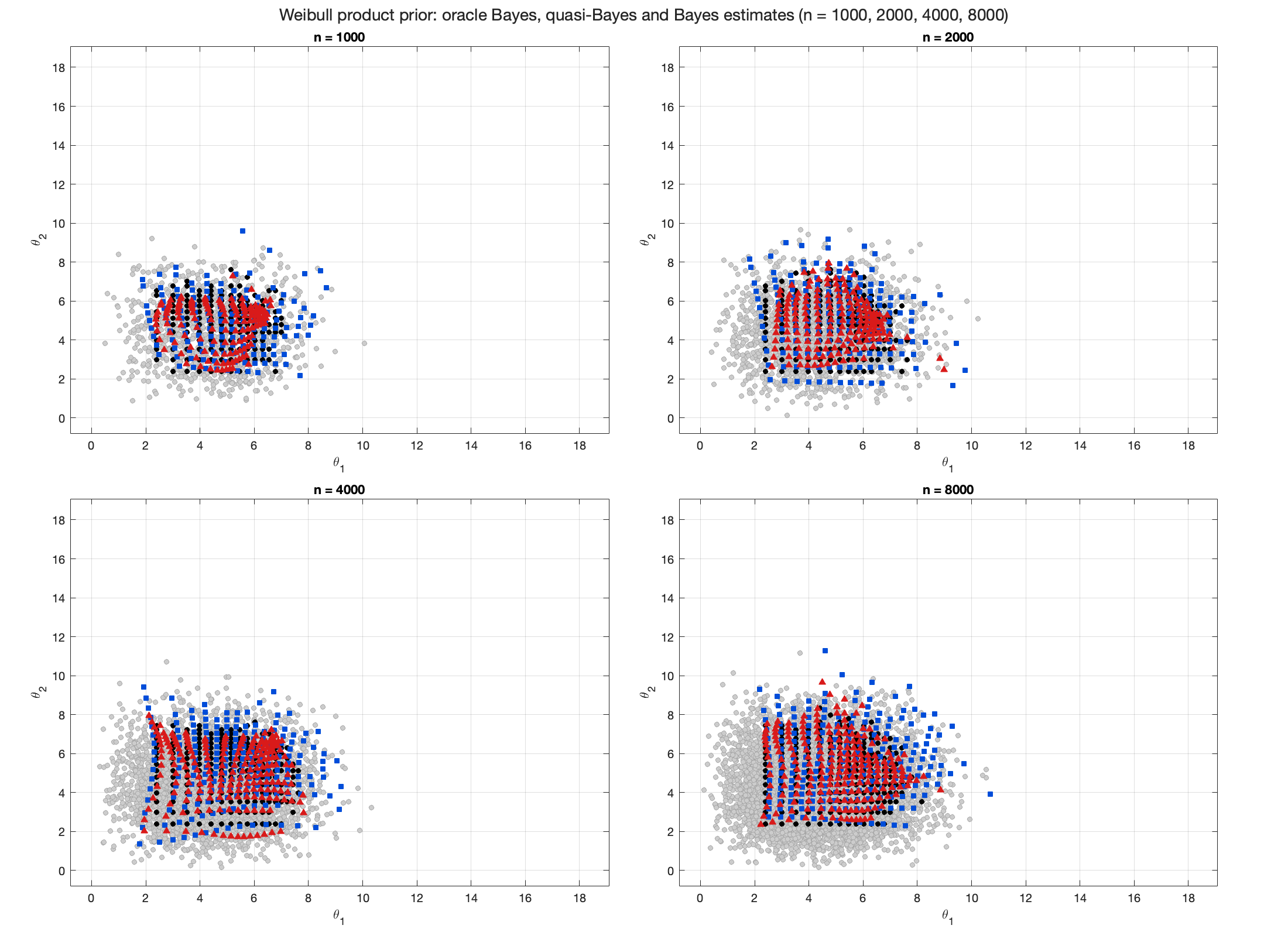}
\caption{\footnotesize{Weibull product prior, $n\in\{1,000,\,2,000,\,4,000,\,8,000\}$: data points plotted against the ``true'' parameters (grey), together with the corresponding oracle Bayes (black), Bayes (red), and quasi-Bayes (blue) estimates.}}
\label{weib_compare2_d}
\end{figure}

\begin{figure}
\centering
\includegraphics[width=.95\textwidth,height=.85\textheight,keepaspectratio]{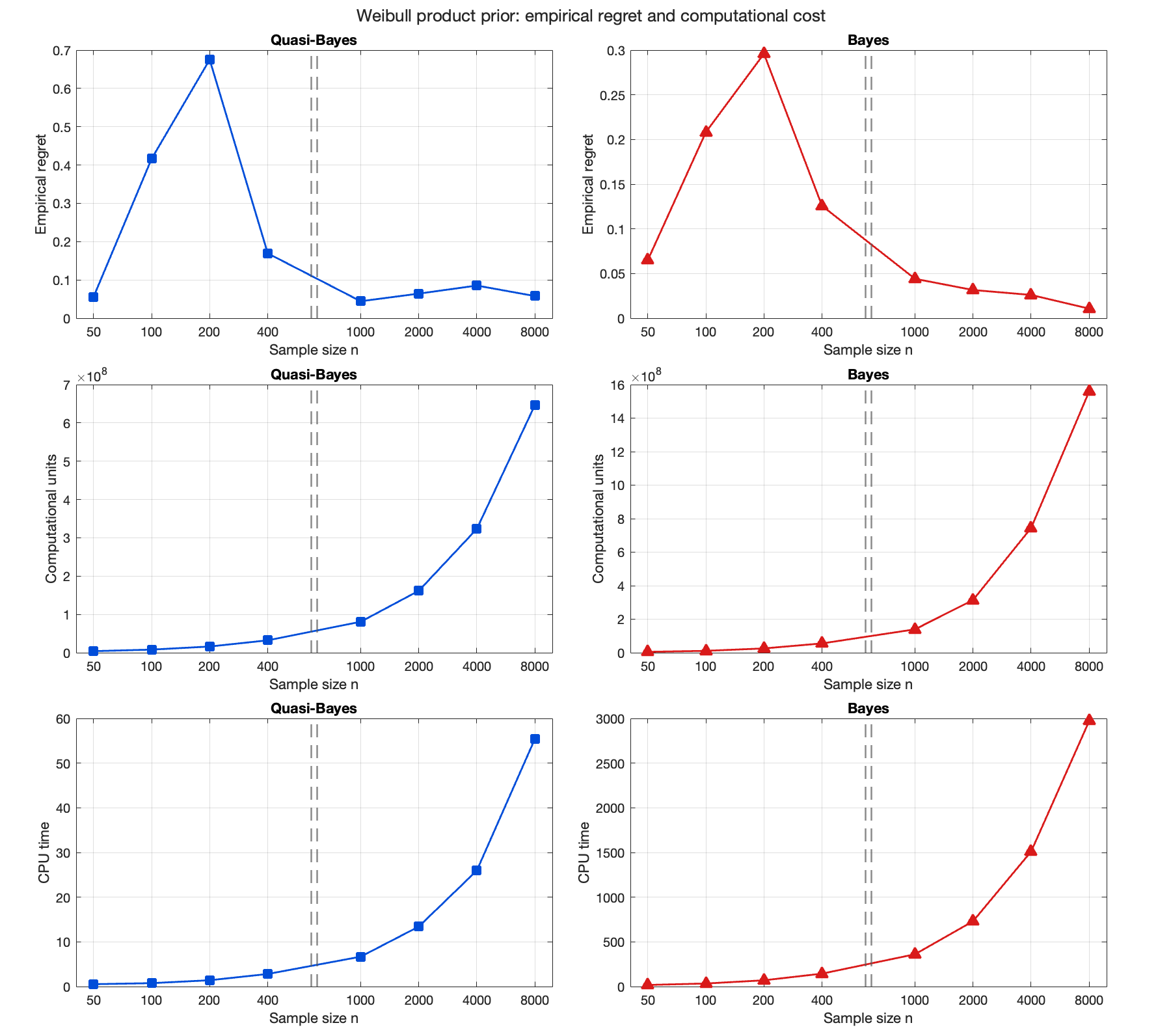}
\caption{\footnotesize{Weibull product prior: quasi-Bayes (blue) and Bayes (red) estimates compared by E-regret (top panels), computational units (middle panels), and CPU time (bottom panels).}}
\label{weib_cpu_d}
\end{figure}

\begin{figure}
\centering
\includegraphics[width=.95\textwidth,height=.85\textheight,keepaspectratio]{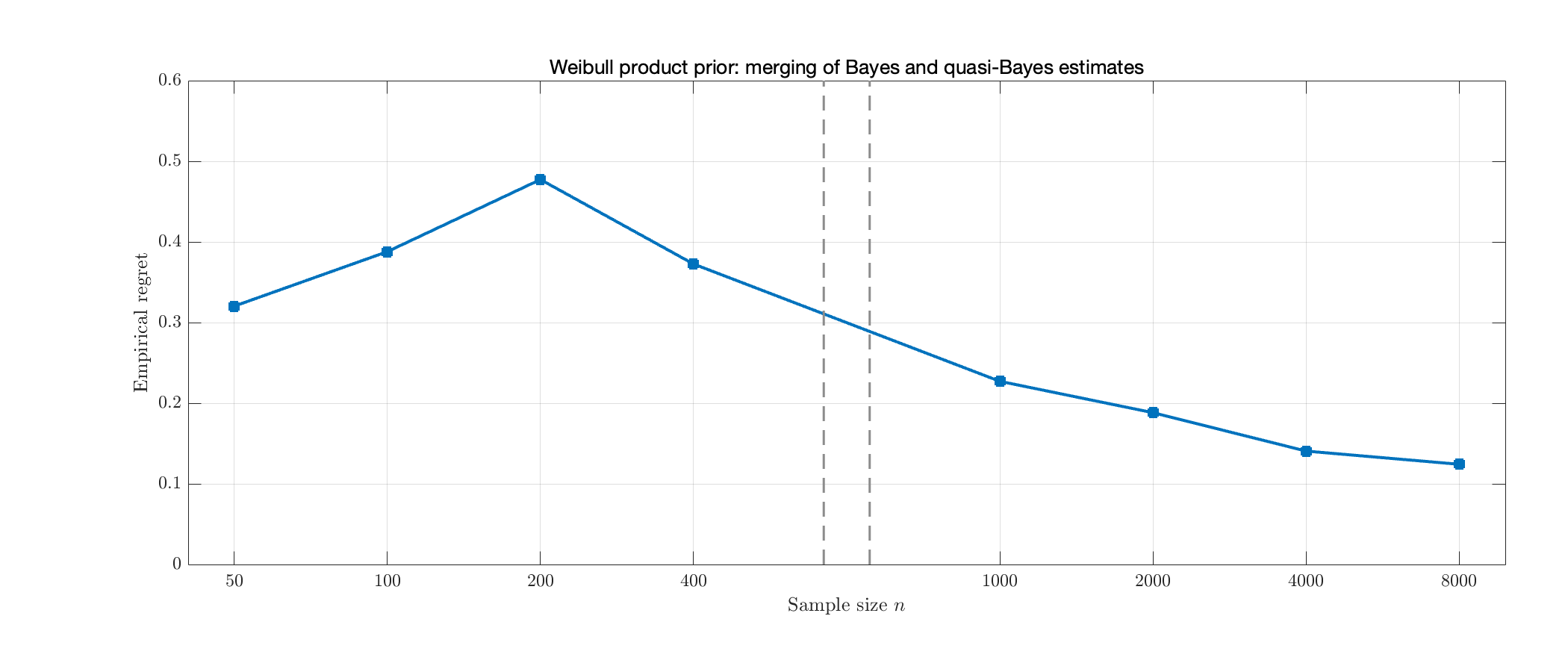}
\caption{\footnotesize{Weibull product prior: E-regret incurred by using the quasi-Bayes estimate in place of the Bayes estimate.}}
\label{regret_comparison_d}
\end{figure}

%%%%%%%%%%%%%%%%%%%%%%%%%%%%%%%%
%%%%%%%%%%%%%%%%%%%%%%%%%%%%%%%%
%%%%%%%%%%%%%%%%%%%%%%%%%%%%%%%%
%%%%%%%%%%%%%%%%%%%%%%%%%%%%%%%%

\section{Concluding remarks}\label{sec4}

Several directions remain open. First, the frequentist merging developed here is not specific to the Poisson kernel. The most natural next case is the Gaussian compound decision problem, where empirical Bayes estimates can be expressed in terms of the marginal density through Tweedie's formula. In the univariate Gaussian case, the main ingredients required for our analysis are available in \citet{MarTok(09)} and \citet{Ign(26)}. These results suggest that Gaussian analogue of our merging results could be obtained by means of an appropriate stability argument for Tweedie's formula. Extending this programme to other mixture kernels would require, in each case, analogous concentration results for the induced marginal distribution and a bound translating marginal concentration into regret control for the corresponding plug-in empirical Bayes rule.

A second question concerns the sharpness of the rates. The Bayesian rates obtained here are consistent with the known minimax benchmarks for the Poisson empirical Bayes problem. By contrast, the rates available for Newton's algorithm are those provided by existing stochastic-approximation theory, and we have not established their optimality. In particular, the dependence on the learning-rate exponent $\gamma$ reflects the specific convergence theory for the recursive update, rather than a matching lower bound for quasi-Bayesian empirical Bayes estimation. It would therefore be interesting to determine whether these rates are intrinsic to Newton's algorithm, or whether sharper
rates can be obtained either by a refined analysis or by alternative choices of the learning-rate sequence.

A third direction is to relax the compact-support assumptions on the oracle mixing distribution and on the parameter space. For the Bayesian strategy, such an extension should be possible by combining posterior concentration arguments with suitable sieve constructions and tail conditions on the base measure and on the oracle prior. For the quasi-Bayesian strategy, the issue is more delicate; see \citet{MarTok(09)}.  The convergence theory used here for Newton's algorithm relies on compactness assumptions, and in the Poisson case also on support bounded away from the boundary point zero. Extending the quasi-Bayesian analysis to unbounded parameter spaces, or to oracle priors with mass arbitrarily close to zero, would require new stochastic-approximation arguments.

Finally, an important extension concerns the estimation of sums of random variables and related functionals \citep{Zha(05)} The present paper focuses on estimating the individual Poisson means through plug-in empirical Bayes rules. In many applications, however, the target is not the collection of individual means itself, but a functional such as a sum, a thresholded sum, or a predictive aggregate depending on both the observations and future random variables. Developing Bayesian and quasi-Bayesian $g$-modeling procedures for such functionals, and studying whether the corresponding estimators merge in regret or in predictive risk, would provide a
natural continuation of the present work.

%%%%%%%%%%%%%%%%%%%%%%%%%%%%%%%%
%%%%%%%%%%%%%%%%%%%%%%%%%%%%%%%%
%%%%%%%%%%%%%%%%%%%%%%%%%%%%%%%%
%%%%%%%%%%%%%%%%%%%%%%%%%%%%%%%%

\phantomsection
\addcontentsline{toc}{section}{Acknowledgment}
\section*{Acknowledgment}
The authors are grateful to Sid Kankanala for bringing the work of \citet{Can(26)} to their attention.

%%%%%%%%%%%%%%%%%%%%%%%%%%%%%%%%
%%%%%%%%%%%%%%%%%%%%%%%%%%%%%%%%
%%%%%%%%%%%%%%%%%%%%%%%%%%%%%%%%
%%%%%%%%%%%%%%%%%%%%%%%%%%%%%%%%
\phantomsection

%%%%%%%%%%%%%%%%%%%%%%%%%%%%%%%%
%%%%%%%%%%%%%%%%%%%%%%%%%%%%%%%%
%%%%%%%%%%%%%%%%%%%%%%%%%%%%%%%%
%%%%%%%%%%%%%%%%%%%%%%%%%%%%%%%%
\newpage

\phantomsection
\addcontentsline{toc}{section}{Appendix}
\section*{Appendix}

\appendix

\renewcommand{\thesection}{\Alph{section}}
\renewcommand{\theequation}{\thesection.\arabic{equation}}
\renewcommand{\thefigure}{\thesection.\arabic{figure}}
\renewcommand{\thetable}{\thesection.\arabic{table}}

\setcounter{section}{0}
\setcounter{equation}{0}
\setcounter{figure}{0}
\setcounter{table}{0}
\setcounter{thm}{0}

%%%%%%%%%%%%%%%%%%%%%%%%%%%%%%%%%%%%%%%%%%%%%%%%%%%%%%%%%%%%%%%%%%%%%%%%%%%%%%%%%%%%%%%%%%%%%%%%%%

\section{Proofs: $1$-dimensional setting}\label{app1}

\subsection{Auxiliary lemmas}

\begin{lem}\label{lem:Poisson_kernel}
For $c\ge 0$ and $y\in\mathbb N_0$, let $q_c(y)=\operatorname{Poisson}(y\mid c)$. For every $\theta,\lambda\ge 0$,
\begin{displaymath}
\sum_{y\ge 0}\sqrt{q_\theta(y)q_\lambda(y)}=\exp\!\left\{-\frac{(\sqrt{\theta}-\sqrt{\lambda})^2}{2}\right\}.
\end{displaymath}
Hence
\begin{displaymath}
d_{\mathrm H}^2(q_\theta,q_\lambda)=2\left(1-\exp\!\left\{-\frac{(\sqrt{\theta}-\sqrt{\lambda})^2}{2}\right\}\right)\le(\sqrt{\theta}-\sqrt{\lambda})^2,
\end{displaymath}
and therefore
\begin{displaymath}
\|q_\theta-q_\lambda\|_1\le 2\,|\sqrt{\theta}-\sqrt{\lambda}|.
\end{displaymath}
\end{lem}

\begin{proof}
If $\theta=0$ or $\lambda=0$, the identity follows directly from the definition of $q_0$. Thus assume first that $\theta,\lambda>0$. For every $y\ge 0$,
\begin{displaymath}
\sqrt{q_\theta(y)q_\lambda(y)}=e^{-(\theta+\lambda)/2}\frac{(\theta\lambda)^{y/2}}{y!}.
\end{displaymath}
Therefore
\begin{displaymath}
\sum_{y\ge 0}\sqrt{q_\theta(y)q_\lambda(y)}=e^{-(\theta+\lambda)/2}\sum_{y\ge 0}\frac{(\sqrt{\theta\lambda})^y}{y!}=e^{-(\theta+\lambda)/2}e^{\sqrt{\theta\lambda}}=
\exp\!\left\{-\frac{(\sqrt{\theta}-\sqrt{\lambda})^2}{2}\right\}.
\end{displaymath}
It follows that
\begin{displaymath}
d_{\mathrm H}^2(q_\theta,q_\lambda)=2-2\sum_{y\ge 0}\sqrt{q_\theta(y)q_\lambda(y)}=2\left(1-\exp\!\left\{-\frac{(\sqrt{\theta}-\sqrt{\lambda})^2}{2}\right\}\right).
\end{displaymath}
Since $1-e^{-x}\le x$ for $x\ge 0$, we obtain
\begin{displaymath}
d_{\mathrm H}^2(q_\theta,q_\lambda)\le(\sqrt{\theta}-\sqrt{\lambda})^2.
\end{displaymath}
Finally, for any probability mass functions $p$ and $q$,
\begin{displaymath}
\|p-q\|_1=\sum_{y\ge 0}|\sqrt{p(y)}-\sqrt{q(y)}|(\sqrt{p(y)}+\sqrt{q(y)})\le2d_{\mathrm H}(p,q),
\end{displaymath}
by Cauchy--Schwarz. Hence
\begin{displaymath}
\|q_\theta-q_\lambda\|_1\le2d_{\mathrm H}(q_\theta,q_\lambda)\le2\,|\sqrt{\theta}-\sqrt{\lambda}|.
\end{displaymath}
\end{proof}

\begin{lem}
\label{lem:Poisson_moment_matching}
Let $0<h_0\le h<+\infty$. For every $G\in\mathcal P([h_{0},h])$ and every integer $K\ge 0$, there exists a discrete probability measure
\begin{displaymath}
G^{(K)}=\sum_{j=1}^{K+1}p_j\delta_{\theta_j},\qquad \theta_j\in[h_{0},h],
\end{displaymath}
such that
\begin{displaymath}
\int_{[h_{0},h]} \theta^r\,G^{(K)}(d\theta)=\int_{[h_{0},h]} \theta^r\,G(d\theta),\qquad r=0,1,\dots,K.
\end{displaymath}
Moreover,
\begin{displaymath}
\|p_G-p_{G^{(K)}}\|_1\le 2\sum_{r\geq K+1}\frac{(2h)^r}{r!}.
\end{displaymath}
In particular, for every sufficiently small $\varepsilon>0$, there exists a discrete $\tilde G$ with at most
\begin{displaymath}
N_\varepsilon \lesssim\frac{\log(1/\varepsilon)}{\log\log(1/\varepsilon)}
\end{displaymath}
support points such that $\|p_G-p_{\tilde G}\|_1\le\varepsilon$.
\end{lem}

\begin{proof}
If $K=0$, take $G^{(0)}=\delta_{\theta_0}$ for any $\theta_0\in[h_{0},h]$; then the moment identity holds for $r=0$, and the $L^1$ bound follows from the argument below, which
uses only the equality of moments up to order $K$. Thus assume $K\ge 1$ for the application of \citet[Lemma A.1]{Gho(01)}. By \citet[Lemma A.1]{Gho(01)} applied with $\psi_r(\theta)=\theta^r$, $r=1,\dots,K$, there exists a discrete probability measure $G^{(K)}$ with at most $K+1$ support points such that
\begin{displaymath}
\int_{[h_{0},h]} \theta^r\,G^{(K)}(d\theta)=\int_{[h_{0},h]} \theta^r\,G(d\theta),\qquad r=1,\dots,K.
\end{displaymath}
Since $G^{(K)}$ is a probability measure, the same identity also holds for $r=0$. Write $\mu_r(G):=\int_{[h_{0},h]} \theta^r\,G(d\theta)$. For every $y\ge 0$,
\begin{displaymath}
p_G(y)=\frac{1}{y!}\int_{[h_{0},h]} e^{-\theta}\theta^y\,G(d\theta).
\end{displaymath}
Since $G$ is supported on $[h_{0},h]$,
\begin{displaymath}
\sum_{m\ge 0}\int_{[h_{0},h]} \frac{\theta^{y+m}}{y!\,m!}\,G(d\theta)\le\frac{1}{y!}\sum_{m\ge 0}\frac{h^{y+m}}{m!}=\frac{h^y e^h}{y!}<+\infty.
\end{displaymath}
Thus, by Tonelli's theorem,
\begin{displaymath}
p_G(y)=\frac{1}{y!}\sum_{m\ge 0}\frac{(-1)^m}{m!}\mu_{y+m}(G),
\end{displaymath}
and the same expansion holds for $G^{(K)}$. Therefore
\begin{displaymath}
p_G(y)-p_{G^{(K)}}(y)=\frac{1}{y!}\sum_{m\ge 0}\frac{(-1)^m}{m!}\bigl\{\mu_{y+m}(G)-\mu_{y+m}(G^{(K)})\bigr\}.
\end{displaymath}
Because the moments agree up to order $K$, all terms with $y+m\le K$ vanish. Accordingly, we can write
\begin{displaymath}
|p_G(y)-p_{G^{(K)}}(y)|\le\frac{1}{y!}\sum_{m:\,y+m\ge K+1}\frac{1}{m!}\bigl|\mu_{y+m}(G)-\mu_{y+m}(G^{(K)})\bigr|.
\end{displaymath}
Since both $G$ and $G^{(K)}$ are supported on $[h_{0},h]$,
\begin{displaymath}
\bigl|\mu_r(G)-\mu_r(G^{(K)})\bigr|\le2h^r.
\end{displaymath}
Thus
\begin{displaymath}
\|p_G-p_{G^{(K)}}\|_1\le2\sum_{y\ge 0}\sum_{m:\,y+m\ge K+1}\frac{h^{y+m}}{y!\,m!}.
\end{displaymath}
Letting $r=y+m$ and using
\begin{displaymath}
\sum_{y=0}^r\frac{1}{y!(r-y)!}=\frac{2^r}{r!},
\end{displaymath}
we obtain
\begin{displaymath}
\|p_G-p_{G^{(K)}}\|_1\le2\sum_{r\geq K+1}\frac{(2h)^r}{r!}.
\end{displaymath}
Finally, Stirling's formula implies that, for large $K$,
\begin{displaymath}
\sum_{r\geq K+1}\frac{(2h)^r}{r!}\le C\left(\frac{2eh}{K+1}\right)^{K+1},
\end{displaymath}
for some constant $C>0$ depending only on $h$. Therefore one can choose
\begin{displaymath}
K\lesssim\frac{\log(1/\varepsilon)}{\log\log(1/\varepsilon)}
\end{displaymath}
so that $\|p_G-p_{G^{(K)}}\|_1\le\varepsilon$, which proves the last claim.
\end{proof}

The next lemma is the analogue of \citet[Lemma 5.1]{Gho(01)} adapted to Poisson mixtures.

\begin{lem}\label{lem:Poisson_L1_local}
Let $\tilde G=\sum_{j=1}^N p_j\delta_{\theta_j}$ be a discrete probability measure and let $\eta>0$ be such that
\begin{displaymath}
|\sqrt{\theta_j}-\sqrt{\theta_k}|>2\eta,
\qquad j\neq k.
\end{displaymath}
Define
\begin{displaymath}
U_j:=\{\theta\ge 0:\ |\sqrt{\theta}-\sqrt{\theta_j}|\le \eta\},\qquad U:=\bigcup_{j=1}^N U_j.
\end{displaymath}
Then, for every probability measure $G$ on $[0,\infty)$,
\begin{displaymath}
\|p_G-p_{\tilde G}\|_1\le 2\eta+\sum_{j=1}^N |G(U_j)-p_j|+G(U^c).
\end{displaymath}
\end{lem}

\begin{proof}
The separation assumption implies that the sets $U_1,\dots,U_N$ are pairwise disjoint. Write
\begin{displaymath}
w_j:=G(U_j),\qquad r:=G(U^c),
\end{displaymath}
so that $\sum_{j=1}^N w_j+r=1$. For each $j$ with $w_j>0$, define
\begin{displaymath}
G_j(A):=\frac{G(A\cap U_j)}{w_j}.
\end{displaymath}
If $w_j=0$, let $G_j$ be any fixed probability measure on $[0,\infty)$; the corresponding term below is then multiplied by $w_j=0$ and is irrelevant. Similarly, if $r>0$, define
\begin{displaymath}
G^c(A):=\frac{G(A\cap U^c)}{r},
\end{displaymath}
while if $r=0$, let $G^c$ be any fixed probability measure on $[0,\infty)$. Then
\begin{displaymath}
p_G=\sum_{j=1}^N w_j p_{G_j}+r p_{G^c}.
\end{displaymath}
Since $p_{\tilde G}=\sum_{j=1}^N p_j q_{\theta_j}$, the triangle inequality gives
\begin{displaymath}
\|p_G-p_{\tilde G}\|_1\le\sum_{j=1}^N \|w_jp_{G_j}-p_jq_{\theta_j}\|_1+r.
\end{displaymath}
Using $\|af-bg\|_1\le|a-b|+a\|f-g\|_1$ for $a,b\ge 0$ and probability mass functions $f$ and $g$, we obtain
\begin{displaymath}
\|p_G-p_{\tilde G}\|_1\le\sum_{j=1}^N |w_j-p_j|+\sum_{j=1}^N w_j\|p_{G_j}-q_{\theta_j}\|_1+r.
\end{displaymath}
By convexity of the $L^1$-norm and Lemma~\ref{lem:Poisson_kernel},
\begin{displaymath}
\|p_{G_j}-q_{\theta_j}\|_1\le\int_{[0,\infty)} \|q_\theta-q_{\theta_j}\|_1\,G_j(d\theta)\le2\eta,
\end{displaymath}
because $G_j(U_j)=1$ and $|\sqrt{\theta}-\sqrt{\theta_j}|\le\eta$ on $U_j$. Hence
\begin{displaymath}
\sum_{j=1}^N w_j\|p_{G_j}-q_{\theta_j}\|_1\le2\eta\sum_{j=1}^N w_j\le2\eta.
\end{displaymath}
Therefore
\begin{displaymath}
\|p_G-p_{\tilde G}\|_1\le\sum_{j=1}^N |G(U_j)-p_j|+2\eta+G(U^c),
\end{displaymath}
as claimed.
\end{proof}

The next lemma is the analogue of \citet[Lemma 4.1]{Gho(01)} adapted to Poisson mixtures.

\begin{lem}\label{lem:Poisson_WS}
For probability mass functions $p$ and $q$ on $\mathbb N_0$, define $K(p,q):=\sum_{y\ge 0}p(y)\log\frac{p(y)}{q(y)}$ and $V(p,q):=\sum_{y\ge 0}p(y)\left(\log\frac{p(y)}{q(y)}\right)^2$, with the usual convention that these quantities are infinite if $q(y)=0$ for some $y$ such that $p(y)>0$. For $\delta>0$, define the Kullback--Leibler type ball
\begin{displaymath}
B(\delta,p_{G^\ast}):=\left\{p:\ K(p_{G^\ast},p)\le \delta^2,\quad V(p_{G^\ast},p)\le \delta^2\right\}.
\end{displaymath}
Let $0<h_0\le h_0^\ast<h^\ast\le h<+\infty$, and assume that $G^\ast([h_0^\ast,h^\ast])=1$. Fix constants $0<b\le h<+\infty$ and $\alpha\in(0,1]$. Then there exist constants $C>0$ and $\eta_0>0$, depending only on $(h_0,h,h_0^\ast,h^\ast,b,\alpha)$, such that for every $G\in\mathcal P([h_0,h])$ satisfying $G([b,h])\ge \alpha$, the implication
$\eta:=d_{\mathrm H}(p_{G^\ast},p_G)<\eta_0$ yields
\begin{displaymath}
K(p_{G^\ast},p_G)\le C\,\eta^2\log\frac{1}{\eta},\qquad V(p_{G^\ast},p_G)\le C\,\eta^2\left(\log\frac{1}{\eta}\right)^2.
\end{displaymath}
Consequently, there exists a constant $A>0$, depending only on $(h_0,h,h_0^\ast,h^\ast,b,\alpha)$, such that for every $G\in\mathcal P([h_0,h])$ satisfying $G([b,h])\ge \alpha$, whenever
\begin{displaymath}
\|p_G-p_{G^\ast}\|_1\le \delta
\end{displaymath}
for some sufficiently small $\delta>0$, one has
\begin{displaymath}
p_G\in B\!\left(A\,\delta^{1/2}\log\frac{1}{\delta},\ p_{G^\ast}\right).
\end{displaymath}
\end{lem}

\begin{proof}
For every $y\ge 0$,
\begin{displaymath}
p_{G^\ast}(y)=\frac{1}{y!}\int_{[h_{0}^{\ast},h^\ast]} e^{-\theta}\theta^y\,G^\ast(d\theta)\le\frac{(h^\ast)^y}{y!}.
\end{displaymath}
Also, since $G([b,h])\ge \alpha$,
\begin{displaymath}
p_G(y)=\frac{1}{y!}\int_{[h_{0},h]} e^{-\theta}\theta^y\,G(d\theta)\ge\frac{1}{y!}\int_{[b,h]}e^{-\theta}\theta^y\,G(d\theta)\ge\alpha e^{-h}\frac{b^y}{y!}.
\end{displaymath}
Therefore
\begin{displaymath}
\frac{p_{G^\ast}(y)}{p_G(y)}\le\alpha^{-1}e^h\left(\frac{h^\ast}{b}\right)^y,\qquad y\ge 0.
\end{displaymath}
Let $C_0:=\alpha^{-1}e^h$, $r:=\frac{h^\ast}{b}$. Fix any $\delta_\ast\in(0,1]$. Then
\begin{displaymath}
\sum_{y\ge 0}\left(\frac{p_{G^\ast}(y)}{p_G(y)}\right)^{\delta_\ast}p_{G^\ast}(y)\le C_0^{\delta_\ast}\sum_{y\ge 0}r^{\delta_\ast y}p_{G^\ast}(y).
\end{displaymath}
If $Y\sim p_{G^\ast}$, then, for every $t>0$,
\begin{displaymath}
\mathbb E_{G^\ast}(t^Y)=\int_{[h_{0}^{\ast},h^\ast]}e^{\theta(t-1)}\,G^\ast(d\theta)\le e^{h^\ast|t-1|}.
\end{displaymath}
Taking $t=r^{\delta_\ast}$ gives
\begin{displaymath}
\sum_{y\ge 0}\left(\frac{p_{G^\ast}(y)}{p_G(y)}\right)^{\delta_\ast}p_{G^\ast}(y)\le C_0^{\delta_\ast}e^{h^\ast|r^{\delta_\ast}-1|}<+\infty.
\end{displaymath}
Thus the integrability condition needed to apply \citet[Theorem 5]{Won(95)} holds uniformly over all $G\in\mathcal P([h_{0},h])$ satisfying $G([b,h])\ge\alpha$. Arguing as in the proof of \citet[Lemma 4.1]{Gho(01)}, we obtain constants $C>0$ and $\eta_0>0$, depending only on $(h_0,h,h_0^\ast,h^\ast,b,\alpha)$, such that, for $\eta=d_{\mathrm H}(p_{G^\ast},p_G)<\eta_0$,
\begin{displaymath}
K(p_{G^\ast},p_G)\le C\,\eta^2\log\frac{1}{\eta},\qquad V(p_{G^\ast},p_G)\le C\,\eta^2\left(\log\frac{1}{\eta}\right)^2.
\end{displaymath}
It remains to prove the final claim. If $\|p_G-p_{G^\ast}\|_1\le\delta$, then
\begin{displaymath}
d_{\mathrm H}^2(p_{G^\ast},p_G)=\sum_{y\ge 0}\{\sqrt{p_{G^\ast}(y)}-\sqrt{p_G(y)}\}^2\le\sum_{y\ge 0}|p_{G^\ast}(y)-p_G(y)|\le\delta.
\end{displaymath}
Thus $d_{\mathrm H}(p_{G^\ast},p_G)\le\delta^{1/2}$. For all sufficiently small $\delta$, this is smaller than $\eta_0$. Therefore
\begin{displaymath}
K(p_{G^\ast},p_G)\le C\,\delta\log\frac{1}{\delta^{1/2}}\lesssim\delta\log\frac{1}{\delta},
\end{displaymath}
and
\begin{displaymath}
V(p_{G^\ast},p_G)\le C\,\delta\left(\log\frac{1}{\delta^{1/2}}\right)^2\lesssim\delta\left(\log\frac{1}{\delta}\right)^2.
\end{displaymath}
Choosing $A>0$ sufficiently large, we obtain
\begin{displaymath}
K(p_{G^\ast},p_G)\le A^2\delta\left(\log\frac{1}{\delta}\right)^2,\qquad V(p_{G^\ast},p_G)\le A^2\delta\left(\log\frac{1}{\delta}\right)^2.
\end{displaymath}
Equivalently,
\begin{displaymath}
p_G\in B\!\left(A\,\delta^{1/2}\log\frac{1}{\delta},\ p_{G^\ast}\right).
\end{displaymath}
This proves the lemma.
\end{proof}

\begin{lem}\label{lem:Poisson_local_sets}
Let $0<h_0\le h_0^\ast<h^\ast\le h<+\infty$ and assume that $G^\ast([h_0^\ast,h^\ast])=1$. Let $\Pi$ denote the law of a Dirichlet process $\operatorname{DP}(1,H)$, where $H$ is a finite measure on $[h_0,h]$ with density $a$ satisfying $0<m_H\le a(\theta)\le M_H<+\infty$ and $h_0\le \theta\le h$. Let $\varepsilon_n:=\frac{\log n}{\sqrt n}$. The Kullback--Leibler type ball $B(\delta,p_{G^\ast})$ is as defined in Lemma~\ref{lem:Poisson_WS}. Then there exist constants $L>0$ and $C_{\mathrm{KL}}>0$, depending only on $(h_0,h,h_0^\ast,h^\ast,G^\ast,\alpha)$, and sets $B_n\subset\mathcal P([h_0,h])$ such that
\begin{displaymath}
\sup_{G\in B_n}K(p_{G^\ast},p_G)\le L^2\varepsilon_n^2,\qquad\sup_{G\in B_n}V(p_{G^\ast},p_G)\le L^2\varepsilon_n^2,
\end{displaymath}
and
\begin{displaymath}
\Pi(B_n)\ge \exp(-C_{\mathrm{KL}}n\varepsilon_n^2)
\end{displaymath}
for all sufficiently large $n$.
\end{lem}

\begin{proof}
Choose $b>0$ such that $3b<h_0^\ast$, and choose $\alpha_0\in(0,1/4)$. Then $G^\ast([3b,h^\ast])=1>4\alpha_0$. Fix a sufficiently small $\varepsilon>0$. Choose an integer
$K=K_\varepsilon$ such that
\begin{displaymath}
2\sum_{r\geq K+1}\frac{(2h)^r}{r!}\le\frac{\varepsilon}{2}.
\end{displaymath}
Let $\varphi:[h_0^\ast,h^\ast]\to[0,1]$ be continuous and satisfy $\varphi(\theta)=1$ for all $\theta\in[h_0^\ast,h^\ast]$. Apply \citet[Lemma A.1]{Gho(01)} on the compact interval
$[h_0^\ast,h^\ast]$ to the functions $\psi_r(\theta)=\theta^r$, $r=1,\dots,K$, and $\psi_{K+1}(\theta)=\varphi(\theta)$. Then there exists a discrete
probability measure
\begin{displaymath}
\tilde{\tilde G}^\ast=\sum_{j=1}^{M_\varepsilon}\widetilde p_j\delta_{\widetilde\theta_j},\qquad\widetilde\theta_j\in[h_0^\ast,h^\ast],\qquad M_\varepsilon\le K+2,
\end{displaymath}
such that
\begin{displaymath}
\int_{[h_0^\ast,h^\ast]}\theta^r\,\tilde{\tilde G}^\ast(d\theta)=\int_{[h_0^\ast,h^\ast]}\theta^r\,G^\ast(d\theta),\qquad r=1,\dots,K,
\end{displaymath}
and
\begin{displaymath}
\int_{[h_0^\ast,h^\ast]}\varphi(\theta)\,\tilde{\tilde G}^\ast(d\theta)=\int_{[h_0^\ast,h^\ast]}\varphi(\theta)\,G^\ast(d\theta)=1.
\end{displaymath}
In particular, $\tilde{\tilde G}^\ast$ is supported on $[h_0^\ast,h^\ast]\subseteq[h_0,h]$ and $\tilde{\tilde G}^\ast([3b,h^\ast])=1>4\alpha_0$. Since $\tilde{\tilde G}^\ast$ is a probability measure, the moment identity also holds for $r=0$. Therefore, by the proof of Lemma~\ref{lem:Poisson_moment_matching},
\begin{displaymath}
\|p_{G^\ast}-p_{\tilde{\tilde G}^\ast}\|_1\le\frac{\varepsilon}{2}.
\end{displaymath}
Now move each support point of $\tilde{\tilde G}^\ast$ to a point of an $\varepsilon$-net in the $\sqrt{\theta}$-scale such that every point of $[\sqrt{h_0},\sqrt h]$ is within distance $\varepsilon/2$ of the net and distinct net points are separated by more than $\varepsilon/2$. Merge duplicate atoms. This yields a discrete probability measure
\begin{displaymath}
\tilde G^\ast=\sum_{j=1}^{N_\varepsilon}p_j\delta_{\theta_j}
\end{displaymath}
such that
\begin{displaymath}
N_\varepsilon\le M_\varepsilon\lesssim\frac{\log(1/\varepsilon)}{\log\log(1/\varepsilon)},
\end{displaymath}
and
\begin{displaymath}
|\sqrt{\theta_j}-\sqrt{\theta_k}|>\frac{\varepsilon}{2},\qquad j\neq k.
\end{displaymath}
Before merging duplicate atoms, each atom has been moved by at most $\varepsilon/2$ in the $\sqrt{\theta}$-scale. Hence, by Lemma~\ref{lem:Poisson_kernel},
\begin{displaymath}
\|p_{\tilde{\tilde G}^\ast}-p_{\tilde G^\ast}\|_1\le\varepsilon.
\end{displaymath}
Therefore
\begin{displaymath}
\|p_{G^\ast}-p_{\tilde G^\ast}\|_1\le\|p_{G^\ast}-p_{\tilde{\tilde G}^\ast}\|_1+\|p_{\tilde{\tilde G}^\ast}-p_{\tilde G^\ast}\|_1\le\frac{3\varepsilon}{2}.
\end{displaymath}
Since $\tilde{\tilde G}^\ast([2b,h])>4\alpha_0$, and each support point is moved by at most $\varepsilon/2$ in the $\sqrt{\theta}$-scale, every atom of $\tilde{\tilde G}^\ast$ lying in $[2b,h]$ is moved into $[b,h]$ for all sufficiently small $\varepsilon$. Merging duplicate atoms does not change the total mass. Hence
\begin{displaymath}
\tilde G^\ast([b,h])>4\alpha_0
\end{displaymath}
for all sufficiently small $\varepsilon$. Define
\begin{displaymath}
U_j:=\{\theta\in[h_{0},h]:\ |\sqrt{\theta}-\sqrt{\theta_j}|\le \varepsilon/4\},\qquad U:=\bigcup_{j=1}^{N_\varepsilon}U_j.
\end{displaymath}
Because the support points are separated by more than $\varepsilon/2$ in the $\sqrt{\theta}$-scale, the sets $U_1,\dots,U_{N_\varepsilon}$ are pairwise disjoint. Let
$U^c=[h_{0},h]\setminus U$ and define
\begin{displaymath}
\mathcal E_\varepsilon:=\left\{G\in\mathcal P([h_{0},h]):\sum_{j=1}^{N_\varepsilon}|G(U_j)-p_j|\le \varepsilon,\ G(U^c)\le \varepsilon\right\}.
\end{displaymath}
If $G\in\mathcal E_\varepsilon$, Lemma~\ref{lem:Poisson_L1_local} gives
\begin{displaymath}
\|p_G-p_{\tilde G^\ast}\|_1\le 2(\varepsilon/4)+\varepsilon+\varepsilon=\frac{5\varepsilon}{2}.
\end{displaymath}
Therefore
\begin{displaymath}
\|p_G-p_{G^\ast}\|_1\le\|p_G-p_{\tilde G^\ast}\|_1+\|p_{\tilde G^\ast}-p_{G^\ast}\|_1\le\frac{5\varepsilon}{2}+\frac{3\varepsilon}{2}=4\varepsilon.
\end{displaymath}
We next show that $G\in\mathcal E_\varepsilon$ implies $G([b/2,h])\ge\alpha_0$ for all sufficiently small $\varepsilon$. If $\theta_j\ge b$, then, for sufficiently small $\varepsilon$,
\begin{displaymath}
U_j\subset [b/2,h].
\end{displaymath}
Hence, for $G\in\mathcal E_\varepsilon$,
\begin{displaymath}
G([b/2,h])\ge\sum_{\theta_j\ge b}G(U_j)\ge\sum_{\theta_j\ge b}p_j-\varepsilon=\tilde G^\ast([b,h])-\varepsilon>4\alpha_0-\varepsilon>\alpha_0,
\end{displaymath}
provided $\varepsilon$ is sufficiently small. Now apply Lemma~\ref{lem:Poisson_WS} with $b/2$ in place of $b$ and $\alpha_0$ in place of $\alpha$. Since $\|p_G-p_{G^\ast}\|_1\le 4\varepsilon$, that lemma yields, for all sufficiently
small $\varepsilon$,
\begin{displaymath}
p_G\in B\!\left(A(4\varepsilon)^{1/2}\log\frac{1}{4\varepsilon},\ p_{G^\ast}\right),
\end{displaymath}
for a constant $A>0$ depending only on $(h_0,h,h_0^\ast,h^\ast,b,\alpha_0)$. Absorbing numerical constants into $A_1$, we obtain
\begin{displaymath}
\mathcal E_\varepsilon\subset B\!\left(A_1\,\varepsilon^{1/2}\log\frac{1}{\varepsilon},\ p_{G^\ast}\right).
\end{displaymath}
Now set $B_n:=\mathcal E_{1/n}$. Since
\begin{displaymath}
\left(\frac{1}{n}\right)^{1/2}\log n=\varepsilon_n,
\end{displaymath}
we have
\begin{displaymath}
B_n\subset B(A_1\varepsilon_n,p_{G^\ast}).
\end{displaymath}
By the definition of the KL-type ball,
\begin{displaymath}
\sup_{G\in B_n}K(p_{G^\ast},p_G)\le A_1^2\varepsilon_n^2,\qquad\sup_{G\in B_n}V(p_{G^\ast},p_G)\le A_1^2\varepsilon_n^2.
\end{displaymath}
It remains to lower bound $\Pi(B_n)$. Since $H$ has density bounded below on $[h_{0},h]$, there exists $c>0$ such that, for $\varepsilon=1/n$,
\begin{displaymath}
H(U_j)\ge c n^{-2},\qquad j=1,\dots,N_{1/n},
\end{displaymath}
for all sufficiently large $n$. Apply \citet[Lemma A.2]{Gho(01)} to the partition $U_1,\dots,U_{N_{1/n}},U^c$, after subdividing $U^c$ into finitely many pieces, if necessary, so that each Dirichlet parameter is at most $1$. This yields
\begin{displaymath}
\Pi(B_n)=\Pi(\mathcal E_{1/n})\ge\exp\{-c_1N_{1/n}\log n\}
\end{displaymath}
for some constant $c_1>0$. Since
\begin{displaymath}
N_{1/n}\lesssim\frac{\log n}{\log\log n},
\end{displaymath}
we get
\begin{displaymath}
N_{1/n}\log n\lesssim\frac{(\log n)^2}{\log\log n}\le C(\log n)^2=C n\varepsilon_n^2
\end{displaymath}
for all sufficiently large $n$, where the second inequality uses $(\log\log n)^{-1}\le 1$ for $n\ge e^e$. Hence there exists $C_{\mathrm{KL},1}>0$ such that
\begin{displaymath}
\Pi(B_n)\ge \exp(-C_{\mathrm{KL},1}n\varepsilon_n^2)
\end{displaymath}
for all sufficiently large $n$. The desired conclusion follows by taking $L:=A_1$ and $C_{\mathrm{KL}}:=C_{\mathrm{KL},1}$. This proves the lemma.
\end{proof}

\begin{lem}\label{lem:Poisson_bracketing_entropy}
Let $0<h_0\le h<+\infty$ and define $\mathcal G_h:=\{p_G:\ G\in\mathcal P([h_0,h])\}$. There exists a constant $C>0$, depending only on $h$, such that, for all
sufficiently small $\varepsilon>0$,
\begin{displaymath}
\log N_{[]}(\varepsilon,\mathcal G_h,d_{\mathrm H})\le C\,\frac{(\log(1/\varepsilon))^2}{\log\log(1/\varepsilon)}.
\end{displaymath}
\end{lem}

\begin{proof}
For every $G\in\mathcal P([h_{0},h])$ and every $y\ge 0$,
\begin{displaymath}
p_G(y)=\int_{[h_{0},h]}\frac{e^{-\theta}\theta^y}{y!}\,G(d\theta)\le\frac{h^y}{y!}.
\end{displaymath}
Let $b_h(y):=\frac{h^y}{y!}$ for $y\ge 0$. Choose an integer $T=T_\varepsilon$ such that $\sum_{y>T} b_h(y)\le \frac{\varepsilon^2}{8}$. By Stirling's formula, one can choose $T_\varepsilon$ so that
\begin{displaymath}
T_\varepsilon
\lesssim
\frac{\log(1/\varepsilon)}{\log\log(1/\varepsilon)}.
\end{displaymath}
We first bracket the coordinates $y=0,\dots,T$. Set $m:=T+1$, and $\delta:=\frac{\varepsilon^2}{8m}$. For every vector $x=(x_0,\dots,x_T)\in[0,1]^m$, define
\begin{displaymath}
a_y:=\delta\left\lfloor \frac{x_y}{\delta}\right\rfloor,
\qquad y=0,\dots,T.
\end{displaymath}
Then
\begin{displaymath}
a_y\le x_y\le a_y+\delta,\qquad y=0,\dots,T.
\end{displaymath}
Thus each $x\in[0,1]^m$ belongs to a bracket $[l,u]$ on $\{0,\dots,T\}$ of the form
\begin{displaymath}
l(y)=a_y,\qquad u(y)=a_y+\delta,\qquad y=0,\dots,T.
\end{displaymath}
The number of such brackets is at most $\left(\frac{1}{\delta}+1\right)^m$, and hence, for sufficiently small $\varepsilon$,
\begin{displaymath}
\log N_T\le m\log\left(\frac{2}{\delta}\right)\lesssim m\log\left(\frac{m}{\varepsilon^2}\right).
\end{displaymath}
Now let $p_G\in\mathcal G_h$. Apply the preceding construction to the vector $(p_G(0),\dots,p_G(T))$. This yields functions $l$ and $u$ on $\mathbb N_0$ defined by
\begin{displaymath}
l(y)=
\begin{cases}
a_y, & 0\le y\le T,\\
0, & y>T,
\end{cases}
\qquad
u(y)=
\begin{cases}
a_y+\delta, & 0\le y\le T,\\
b_h(y), & y>T.
\end{cases}
\end{displaymath}
Then $l(y)\le p_G(y)\le u(y)$ for all $y\ge 0$. It remains to bound the bracket width in Hellinger distance. Since
\begin{displaymath}
(\sqrt{u}-\sqrt{l})^2\le u-l,\qquad 0\le l\le u,
\end{displaymath}
we have
\begin{displaymath}
d_{\mathrm H}^2(l,u)=\sum_{y\ge 0}(\sqrt{u(y)}-\sqrt{l(y)})^2\le\sum_{y=0}^T\{u(y)-l(y)\}+\sum_{y>T}b_h(y).
\end{displaymath}
By construction,
\begin{displaymath}
\sum_{y=0}^T\{u(y)-l(y)\}=(T+1)\delta=\frac{\varepsilon^2}{8},
\end{displaymath}
and
\begin{displaymath}
\sum_{y>T}b_h(y)\le \frac{\varepsilon^2}{8}.
\end{displaymath}
Therefore
\begin{displaymath}
d_{\mathrm H}^2(l,u)\le \frac{\varepsilon^2}{4},\qquad d_{\mathrm H}(l,u)\le \frac{\varepsilon}{2}.
\end{displaymath}
Hence these brackets form an $\varepsilon$-bracketing of $\mathcal G_h$. Thus
\begin{displaymath}
\log N_{[]}(\varepsilon,\mathcal G_h,d_{\mathrm H})\lesssim (T_\varepsilon+1)\log\left(\frac{T_\varepsilon+1}{\varepsilon^2}\right).
\end{displaymath}
Using
\begin{displaymath}
T_\varepsilon\lesssim\frac{\log(1/\varepsilon)}{\log\log(1/\varepsilon)},
\end{displaymath}
we conclude that
\begin{displaymath}
\log N_{[]}(\varepsilon,\mathcal G_h,d_{\mathrm H})\le C\,\frac{(\log(1/\varepsilon))^2}{\log\log(1/\varepsilon)}
\end{displaymath}
for some constant $C>0$ depending only on $h$. Indeed, by Stirling's formula,
\begin{displaymath}
\sum_{y>T}\frac{h^y}{y!}\le C(h)\left(\frac{eh}{T+1}\right)^{T+1}
\end{displaymath}
for all sufficiently large $T$, which yields the stated order of $T_\varepsilon$.
\end{proof}

\begin{lem}\label{lem:Poisson_HP_En}
Let $0<h_0\le h_0^\ast<h^\ast\le h<+\infty$ and assume that $G^\ast([h_0^\ast,h^\ast])=1$. Let $\Pi$ be any probability measure on $\mathcal P([h_0,h])$. Define
$\varepsilon_n:=\frac{\log n}{\sqrt n}$, $r_n:=C\varepsilon_n$ and $U_n:=\{G\in\mathcal P([h_0,h]):d_{\mathrm H}(p_G,p_{G^\ast})\ge r_n\}$, where $C>0$ is a constant to be chosen. Define
\begin{displaymath}
R_n(G):=\prod_{i=1}^n\frac{p_G(Y_i)}{p_{G^\ast}(Y_i)}.
\end{displaymath}
Then, for every $D>0$, there exists a sufficiently large $C>0$ such that, for all sufficiently large $n$, there is an event $E_n$ satisfying
\begin{displaymath}
\P_{G^\ast}^n(E_n^c)\le \frac{1}{n^2},
\end{displaymath}
and, on $E_n$,
\begin{displaymath}
\int_{U_n}R_n(G)\,\Pi(dG)\le\exp(-Dn\varepsilon_n^2).
\end{displaymath}
\end{lem}

\begin{proof}
Fix $0<\kappa<1/2$ and set $\delta_n:=\kappa r_n^2$. Let $\mathcal U_n:=\left\{p_G\in\mathcal G_h:\ d_{\mathrm H}(p_G,p_{G^\ast})\ge r_n\right\}$. By Lemma~\ref{lem:Poisson_bracketing_entropy}, there exists a collection of
$\delta_n$-brackets $[l_k,u_k]$, $k=1,\dots,N_n$, covering $\mathcal G_h$ such that
\begin{displaymath}
\log N_n
\lesssim
\frac{(\log(1/\delta_n))^2}{\log\log(1/\delta_n)}\asymp\frac{(\log n)^2}{\log\log n}=o(\log^2 n)=o(n\varepsilon_n^2).
\end{displaymath}
By the proof of Lemma~\ref{lem:Poisson_bracketing_entropy}, these brackets may be chosen so that each $l_k$ and $u_k$ is nonnegative and $\sum_{y\ge 0}u_k(y)<+\infty$. Retain only those brackets that intersect $\mathcal U_n$. For each retained bracket, choose
$p_k\in\mathcal U_n$ such that
\begin{displaymath}
l_k\le p_k\le u_k.
\end{displaymath}
Then
\begin{displaymath}
d_{\mathrm H}(p_k,p_{G^\ast})\ge r_n.
\end{displaymath}
Moreover, since $l_k\le p_k\le u_k$, we have pointwise
\begin{displaymath}
0\le \sqrt{u_k}-\sqrt{p_k}\le \sqrt{u_k}-\sqrt{l_k}.
\end{displaymath}
Therefore
\begin{displaymath}
d_{\mathrm H}(u_k,p_k)\le d_{\mathrm H}(u_k,l_k)\le \delta_n.
\end{displaymath}
For every retained bracket,
\begin{displaymath}
\sum_{y\ge 0}\sqrt{u_k(y)p_{G^\ast}(y)}=\sum_{y\ge 0}\sqrt{p_k(y)p_{G^\ast}(y)}+\sum_{y\ge 0}\left\{\sqrt{u_k(y)}-\sqrt{p_k(y)}\right\}\sqrt{p_{G^\ast}(y)}.
\end{displaymath}
By Cauchy--Schwarz,
\begin{displaymath}
\sum_{y\ge 0}\left\{\sqrt{u_k(y)}-\sqrt{p_k(y)}\right\}\sqrt{p_{G^\ast}(y)}\le d_{\mathrm H}(u_k,p_k).
\end{displaymath}
Hence
\begin{displaymath}
\sum_{y\ge 0}\sqrt{u_k(y)p_{G^\ast}(y)}\le 1-\frac{1}{2}d_{\mathrm H}^2(p_k,p_{G^\ast})+d_{\mathrm H}(u_k,p_k)\le 1-\left(\frac{1}{2}-\kappa\right)r_n^2.
\end{displaymath}
Let $b:=\frac{1}{2}-\kappa>0$. Then
\begin{displaymath}
\E_{G^\ast}\!\left[\sqrt{R_n(u_k)}\right]=\left(\sum_{y\ge 0}\sqrt{u_k(y)p_{G^\ast}(y)}\right)^n\le(1-br_n^2)^n\le\exp(-bnr_n^2),
\end{displaymath}
where, for a nonnegative summable function $u$ on $\mathbb N_0$, we write
\begin{displaymath}
R_n(u):=\prod_{i=1}^n\frac{u(Y_i)}{p_{G^\ast}(Y_i)}.
\end{displaymath}
Define
\begin{displaymath}
E_n:=\left\{\max_{1\le k\le N_n}R_n(u_k)\le\exp\!\left(-\frac{b}{2}nr_n^2\right)\right\}.
\end{displaymath}
By Markov's inequality applied to $\sqrt{R_n(u_k)}$,
\begin{displaymath}
\P_{G^\ast}^n\!\left(R_n(u_k)>\exp\!\left(-\frac{b}{2}nr_n^2\right)\right)\le\exp\!\left(-\frac{b}{4}nr_n^2\right).
\end{displaymath}
Therefore, by the union bound,
\begin{displaymath}
\P_{G^\ast}^n(E_n^c)\le N_n\exp\!\left(-\frac{b}{4}nr_n^2\right)\le\exp\!\left\{\log N_n-\frac{b}{4}C^2n\varepsilon_n^2\right\}.
\end{displaymath}
Since $\log N_n=o(n\varepsilon_n^2)$, choosing $C$ sufficiently large gives
\begin{displaymath}
\P_{G^\ast}^n(E_n^c)\le \frac{1}{n^2}
\end{displaymath}
for all sufficiently large $n$. Finally, on $E_n$, for every $G\in U_n$, choose a retained bracket $[l_k,u_k]$ containing
$p_G$. Since $p_G\le u_k$ pointwise,
\begin{displaymath}
R_n(G)\le R_n(u_k).
\end{displaymath}
Thus
\begin{displaymath}
\int_{U_n}R_n(G)\,\Pi(dG)\le\sum_{k=1}^{N_n}R_n(u_k)\Pi\!\left(\{G:\ p_G\in [l_k,u_k]\cap\mathcal U_n\}\right)\le\sum_{k=1}^{N_n}R_n(u_k).
\end{displaymath}
Therefore, on $E_n$,
\begin{displaymath}
\int_{U_n}R_n(G)\,\Pi(dG)\le N_n\exp\!\left(-\frac{b}{2}nr_n^2\right).
\end{displaymath}
Again using $\log N_n=o(n\varepsilon_n^2)$, a sufficiently large choice of $C$ yields
\begin{displaymath}
N_n\exp\!\left(-\frac{b}{2}nr_n^2\right)\le\exp(-Dn\varepsilon_n^2),
\end{displaymath}
which proves the lemma.
\end{proof}

\begin{lem}\label{lem:Poisson_HP_exp_moment}
Assume the setup of Lemma~\ref{lem:Poisson_local_sets}, and let $B_n$ be the local set constructed therein. Then there exist constants $\lambda_0>0$ and $C_0<+\infty$, depending only on $(h_0,h,h_0^\ast,h^\ast,G^\ast)$, such that, for all sufficiently large $n$, 
\begin{displaymath}
 \sup_{G\in B_n} \E_{G^\ast}\!\left[ \exp\!\left\{ \lambda_0 \left| \log \frac{p_G(Y_1)}{p_{G^\ast}(Y_1)} \right| \right\} \right] \le C_0 .
 \end{displaymath}
\end{lem}

\begin{proof}
Choose $b>0$ such that $3b<h_0^\ast$, and choose $\alpha_0\in(0,1/4)$. Then $G^\ast([3b,h^\ast])=1>4\alpha_0$. By the construction in the proof of Lemma~\ref{lem:Poisson_local_sets}, every $G\in B_n$ satisfies
\begin{displaymath}
G([b/2,h])\ge \alpha_0
 \end{displaymath}
for all sufficiently large $n$. Hence, by the proof of Lemma~\ref{lem:Poisson_WS}, there exist constants $C_1,C_2>0$, depending only on $(h_0,h,h_0^\ast,h^\ast,b,\alpha_0)$, such that
\begin{displaymath}
\frac{p_{G^\ast}(y)}{p_G(y)}\le C_1C_2^y,\qquad y\ge 0,
 \end{displaymath}
uniformly over $G\in B_n$. Moreover, since $G^\ast([3b,h^\ast])=1$, for every $y\ge 0$,
\begin{displaymath}
p_{G^\ast}(y)=\frac{1}{y!}\int_{[h_0^\ast,h^\ast]}e^{-\theta}\theta^y\,G^\ast(d\theta)\ge e^{-h^\ast}\frac{(3b)^y}{y!}.
 \end{displaymath}
On the other hand, for every $G\in\mathcal P([h_0,h])$,
\begin{displaymath}
p_G(y)=\frac{1}{y!}\int_{[h_0,h]}e^{-\theta}\theta^y\,G(d\theta)\le\frac{h^y}{y!}.
\end{displaymath}
Therefore
\begin{displaymath}
\frac{p_G(y)}{p_{G^\ast}(y)}\le e^{h^\ast}\left(\frac{h}{3b}\right)^y,\qquad y\ge 0.
\end{displaymath}
Combining the preceding two bounds, there exist constants $a_0,a_1>0$, depending only on $(h_0,h,h_0^\ast,h^\ast)$, such that
\begin{displaymath}
\left|\log\frac{p_G(y)}{p_{G^\ast}(y)}\right|\le a_0+a_1y,\qquad y\ge 0,
\end{displaymath}
uniformly over $G\in B_n$ and all sufficiently large $n$. Since $Y_1\sim p_{G^\ast}$, for every $t>0$,
\begin{displaymath}
\E_{G^\ast}(e^{tY_1})=\int_{[h_0^\ast,h^\ast]}\exp\{\theta(e^t-1)\}\,G^\ast(d\theta)\le\exp\{h^\ast(e^t-1)\}<+\infty.
\end{displaymath}
Thus, for any fixed $\lambda_0>0$,
\begin{displaymath}
\sup_{G\in B_n}\E_{G^\ast}\!\left[\exp\!\left\{\lambda_0\left|\log\frac{p_G(Y_1)}{p_{G^\ast}(Y_1)}\right|\right\}\right]\le e^{\lambda_0 a_0}\E_{G^\ast}(e^{\lambda_0a_1Y_1})<+\infty .
\end{displaymath}
This proves the lemma.
\end{proof}

\begin{lem}\label{lem:Poisson_HP_Dn}
Let $0<h_0\le h<+\infty$, let $\varepsilon_n:=\frac{\log n}{\sqrt n}$, and define $R_n(G):=\prod_{i=1}^n \frac{p_G(Y_i)}{p_{G^\ast}(Y_i)}$. Let $B_n\subset \mathcal P([h_{0},h])$ satisfy $\Pi(B_n)\ge \exp(-C_{\mathrm{KL}}n\varepsilon_n^2)$ for some constant $C_{\mathrm{KL}}>0$. Assume further that there exists a constant $L>0$ such that
\begin{displaymath}
\sup_{G\in B_n}K(p_{G^\ast},p_G)\le L^2\varepsilon_n^2,\qquad\sup_{G\in B_n}V(p_{G^\ast},p_G)\le L^2\varepsilon_n^2,
\end{displaymath}
and that $B_n$ satisfies the uniform exponential moment condition of Lemma~\ref{lem:Poisson_HP_exp_moment}. For $D>0$, define
\begin{displaymath}
D_n:=\left\{\int_{\mathcal P([h_{0},h])} R_n(G)\,\Pi(dG)\ge\exp\!\left\{-(L^2+D)n\varepsilon_n^2\right\}\Pi(B_n)\right\}.
\end{displaymath}
Then there exists a constant $c_D>0$ such that, for all sufficiently large $n$,
\begin{displaymath}
\P_{G^\ast}^n(D_n^c)\le\exp(-c_D n\varepsilon_n^2)\le\frac{1}{n^2}.
\end{displaymath}
In particular, on $D_n$,
\begin{displaymath}
\int_{\mathcal P([h_{0},h])} R_n(G)\,\Pi(dG)\ge\exp(-C_{\mathrm{den}}n\varepsilon_n^2)
\end{displaymath}
for some constant $C_{\mathrm{den}}>0$.
\end{lem}

\begin{proof}
Let $\widetilde\Pi_n(\cdot):=\frac{\Pi(\cdot\cap B_n)}{\Pi(B_n)}$. By Jensen's inequality,
\begin{displaymath}
\int_{\mathcal P([h_{0},h])} R_n(G)\,\Pi(dG)\ge\Pi(B_n)\exp\!\left\{\int_{B_n}\log R_n(G)\,\widetilde\Pi_n(dG)\right\}.
\end{displaymath}
Hence it suffices to control
\begin{displaymath}
W_n:=\int_{B_n}\log R_n(G)\,\widetilde\Pi_n(dG).
\end{displaymath}
Write
\begin{displaymath}
W_n=\sum_{i=1}^n \xi_{n,i},\qquad\xi_{n,i}:=\int_{B_n}\log\frac{p_G(Y_i)}{p_{G^\ast}(Y_i)}\,\widetilde\Pi_n(dG).
\end{displaymath}
Since $Y_{1:n}$ are i.i.d. under $P_{G^\ast}^n$, the variables $\xi_{n,1},\dots,\xi_{n,n}$ are i.i.d. By Fubini's theorem,
\begin{displaymath}
\E_{G^\ast}(\xi_{n,1})=\int_{B_n}\E_{G^\ast}\!\left[\log\frac{p_G(Y_1)}{p_{G^\ast}(Y_1)}\right]\,\widetilde\Pi_n(dG)=-\int_{B_n}K(p_{G^\ast},p_G)\,\widetilde\Pi_n(dG)\ge
-L^2\varepsilon_n^2.
\end{displaymath}
Moreover, by Jensen's inequality,
\begin{displaymath}
\xi_{n,1}^2\le\int_{B_n}\left(\log\frac{p_G(Y_1)}{p_{G^\ast}(Y_1)}\right)^2\,\widetilde\Pi_n(dG).
\end{displaymath}
Taking expectations gives
\begin{displaymath}
\E_{G^\ast}(\xi_{n,1}^2)\le\int_{B_n}V(p_{G^\ast},p_G)\,\widetilde\Pi_n(dG)\le L^2\varepsilon_n^2.
\end{displaymath}
Thus
\begin{displaymath}
\operatorname{Var}_{G^\ast}(\xi_{n,1})\le L^2\varepsilon_n^2.
\end{displaymath}
Next, by Jensen's inequality and the uniform exponential moment condition of Lemma~\ref{lem:Poisson_HP_exp_moment},
\begin{displaymath}
\begin{split}
\E_{G^\ast}\!\left[e^{\lambda_0|\xi_{n,1}|}\right]
&\le\E_{G^\ast}\!\left[\int_{B_n}\exp\!\left\{\lambda_0\left|\log\frac{p_G(Y_1)}{p_{G^\ast}(Y_1)}\right|\right\}\,\widetilde\Pi_n(dG)\right]  \\
&\le\sup_{G\in B_n}\E_{G^\ast}\!\left[\exp\!\left\{\lambda_0\left|\log\frac{p_G(Y_1)}{p_{G^\ast}(Y_1)}\right|\right\}\right]\le C_0.
\end{split}
\end{displaymath}
Define $Z_{n,i}:=\xi_{n,i}-\E_{G^\ast}(\xi_{n,1})$, $i=1,\dots,n$. Then $Z_{n,1},\dots,Z_{n,n}$ are i.i.d., mean zero, and $W_n-\E_{G^\ast}(W_n)=
\sum_{i=1}^n Z_{n,i}$. Moreover,
\begin{displaymath}
\sum_{i=1}^n \E_{G^\ast}(Z_{n,i}^2)=n\,\operatorname{Var}_{G^\ast}(\xi_{n,1})\le L^2n\varepsilon_n^2.
\end{displaymath}
The preceding exponential moment bound implies that the centered variables $Z_{n,i}$ have uniformly bounded sub-exponential norms. Indeed,
\begin{displaymath}
\E_{G^\ast}|\xi_{n,1}|\le\lambda_0^{-1}\log \E_{G^\ast}\!\left[e^{\lambda_0|\xi_{n,1}|}\right]\le\lambda_0^{-1}\log C_0,
\end{displaymath}
and hence
\begin{displaymath}
\E_{G^\ast}\!\left[e^{(\lambda_0/2)|Z_{n,1}|}\right]\le\exp\!\left\{\frac{\lambda_0}{2}\E_{G^\ast}|\xi_{n,1}|\right\}\E_{G^\ast}\!\left[e^{(\lambda_0/2)|\xi_{n,1}|}\right]\le C_1
\end{displaymath}
for some constant $C_1<+\infty$ independent of $n$. We now apply Bernstein's inequality for independent mean-zero random variables with uniformly bounded sub-exponential norms and variance proxy $\sum_{i=1}^n\E_{G^\ast}(Z_{n,i}^2)$. There exists a constant $c>0$, depending only on
$(\lambda_0,C_0)$, such that, for every $t>0$,
\begin{displaymath}
\P_{G^\ast}^n\!\left(\sum_{i=1}^nZ_{n,i}\le -t\right)\le\exp\!\left[-c\,\min\!\left\{\frac{t^2}{\sum_{i=1}^n\E_{G^\ast}(Z_{n,i}^2)},\, t\right\}\right].
\end{displaymath}
Taking $t=Dn\varepsilon_n^2$ and using $\sum_{i=1}^n\E_{G^\ast}(Z_{n,i}^2)\le L^2n\varepsilon_n^2$, we obtain
\begin{displaymath}
\P_{G^\ast}^n\!\left(W_n-\E_{G^\ast}(W_n)<-Dn\varepsilon_n^2\right)\le\exp\!\left[-c\min\!\left\{\frac{D^2}{L^2}n\varepsilon_n^2,\,Dn\varepsilon_n^2\right\}\right].
\end{displaymath}
Therefore there exists $c_D>0$ such that
\begin{displaymath}
\P_{G^\ast}^n\!\left(W_n-\E_{G^\ast}(W_n)<-Dn\varepsilon_n^2\right)\le\exp(-c_Dn\varepsilon_n^2).
\end{displaymath}
Since $\E_{G^\ast}(W_n)=n\E_{G^\ast}(\xi_{n,1})\ge -L^2n\varepsilon_n^2$, we have
\begin{displaymath}
\left\{W_n<-(L^2+D)n\varepsilon_n^2\right\}\subset\left\{W_n-\E_{G^\ast}(W_n)<-Dn\varepsilon_n^2\right\}.
\end{displaymath}
Thus
\begin{displaymath}
\P_{G^\ast}^n\!\left(W_n<-(L^2+D)n\varepsilon_n^2\right)\le\exp(-c_Dn\varepsilon_n^2).
\end{displaymath}
By the Jensen lower bound,
\begin{displaymath}
\int_{\mathcal{P}([h_{0},h])} R_n(G)\,\Pi(dG)\ge\Pi(B_n)\exp(W_n).
\end{displaymath}
Hence
\begin{displaymath}
\left\{W_n\ge -(L^2+D)n\varepsilon_n^2\right\}\subset D_n.
\end{displaymath}
Consequently,
\begin{displaymath}
\P_{G^\ast}^n(D_n^c)\le\P_{G^\ast}^n\!\left(W_n<-(L^2+D)n\varepsilon_n^2\right)\le\exp(-c_Dn\varepsilon_n^2).
\end{displaymath}
Since $n\varepsilon_n^2=\log^2 n$, the last quantity is bounded by $1/n^2$ for all sufficiently large $n$. Finally, on $D_n$,
\begin{displaymath}
\int_{\mathcal{P}([h_{0},h])} R_n(G)\,\Pi(dG)\ge\exp\!\left\{-(L^2+D)n\varepsilon_n^2\right\}\Pi(B_n)\ge\exp\!\left\{-(L^2+D+C_{\mathrm{KL}})n\varepsilon_n^2\right\}.
\end{displaymath}
Thus the final claim holds with
\begin{displaymath}
C_{\mathrm{den}}:=L^2+D+C_{\mathrm{KL}}.
\end{displaymath}
\end{proof}

\subsection{Proof of Proposition~\ref{pcrdp_1d}}

For notational simplicity, write $\varepsilon_n:=\frac{\log n}{\sqrt n}$. For a constant $C>0$ to be chosen later, set $U_n:=\{G\in\mathcal P([h_0,h]):d_{\mathrm H}(p_G,p_{G^\ast})\ge C\varepsilon_n\}$, and define $R_n(G):=\prod_{i=1}^n\frac{p_G(Y_i)}{p_{G^\ast}(Y_i)}$. Then
\begin{displaymath}
\Pi(U_n\mid Y_{1:n})=\frac{\int_{U_n}R_n(G)\,\Pi(dG)}{\int_{\mathcal P([h_0,h])}R_n(G)\,\Pi(dG)}.
\end{displaymath}
By Lemma~\ref{lem:Poisson_local_sets}, there exists a local set $B_n\subset\mathcal P([h_0,h])$ and constants $L>0$ and $C_{\mathrm{KL}}>0$, depending only on $(h_0,h,h_0^\ast,h^\ast,G^\ast,\alpha)$, such that, for all sufficiently large $n$,
\begin{displaymath}
\sup_{G\in B_n}K(p_{G^\ast},p_G)\le L^2\varepsilon_n^2,\qquad\sup_{G\in B_n}V(p_{G^\ast},p_G)\le L^2\varepsilon_n^2,
\end{displaymath}
and
\begin{displaymath}
\Pi(B_n)\ge \exp(-C_{\mathrm{KL}}n\varepsilon_n^2).
\end{displaymath}
Moreover, by Lemma~\ref{lem:Poisson_HP_exp_moment}, the sets $B_n$ satisfy the uniform exponential moment condition required in Lemma~\ref{lem:Poisson_HP_Dn}. Fix, for instance, $D_0=1$. Define
\begin{displaymath}
D_n:=\left\{\int_{\mathcal P([h_0,h])}R_n(G)\,\Pi(dG)\ge\exp\{-(L^2+D_0)n\varepsilon_n^2\}\Pi(B_n)\right\}.
\end{displaymath}
By Lemma~\ref{lem:Poisson_HP_Dn}, there exists a constant $C_{\mathrm{den}}>0$, depending only on $(h_0,h,h_0^\ast,h^\ast,G^\ast,\alpha)$, such that, for all sufficiently
large $n$,
\begin{displaymath}
\P_{G^\ast}^n(D_n^c)\le \frac{1}{n^2},
\end{displaymath}
and, on $D_n$,
\begin{displaymath}
\int_{\mathcal P([h_0,h])}R_n(G)\,\Pi(dG)\ge\exp(-C_{\mathrm{den}}n\varepsilon_n^2).
\end{displaymath}
Now choose $D_1>C_{\mathrm{den}}$. By Lemma~\ref{lem:Poisson_HP_En}, there exists a sufficiently large constant $C>0$, depending only on $(h_0,h,h_0^\ast,h^\ast,G^\ast,\alpha)$, and, for all sufficiently large $n$, an event $E_n$ such that
\begin{displaymath}
\P_{G^\ast}^n(E_n^c)\le \frac{1}{n^2},
\end{displaymath}
and, on $E_n$,
\begin{displaymath}
\int_{U_n}R_n(G)\,\Pi(dG)\le\exp(-D_1n\varepsilon_n^2).
\end{displaymath}
Therefore, on $D_n\cap E_n$,
\begin{displaymath}
\Pi(U_n\mid Y_{1:n})=\frac{\int_{U_n}R_n(G)\,\Pi(dG)}{\int_{\mathcal P([h_0,h])}R_n(G)\,\Pi(dG)}\le\exp\{(C_{\mathrm{den}}-D_1)n\varepsilon_n^2\}.
\end{displaymath}
Let $c:=D_1-C_{\mathrm{den}}>0$. Since $n\varepsilon_n^2=\log^2 n$, on $D_n\cap E_n$ we obtain $\Pi(U_n\mid Y_{1:n})\le\exp(-c\log^2 n)$. Hence, for all sufficiently large $n$,
\begin{displaymath}
\P_{G^\ast}^n\!\left[\Pi(U_n\mid Y_{1:n})>\exp(-c\log^2 n)\right]\le\P_{G^\ast}^n(D_n^c)+\P_{G^\ast}^n(E_n^c)\le\frac{2}{n^2}\le\frac{1}{n}.
\end{displaymath}
This proves the first assertion \eqref{pcr_dp_1}. Now define the posterior mean marginal probability mass function
\begin{displaymath}
\bar p_n(y):=\int_{\mathcal P([h_0,h])}p_G(y)\,\Pi(dG\mid Y_{1:n}),\qquad y\in\mathbb N_0.
\end{displaymath}
Let $A_n:=\{\Pi(U_n\mid Y_{1:n})\le \exp(-c\log^2 n)\}$. By the first part,
\begin{displaymath}
\P_{G^\ast}^n(A_n)\ge 1-\frac{1}{n}
\end{displaymath}
for all sufficiently large $n$. On $A_n$, by convexity of squared Hellinger distance in its first argument,
\begin{displaymath}
d_{\mathrm H}^2(\bar p_n,p_{G^\ast})\le\int_{\mathcal P([h_0,h])}d_{\mathrm H}^2(p_G,p_{G^\ast})\,\Pi(dG\mid Y_{1:n}).
\end{displaymath}
Splitting the integral over $U_n^c$ and $U_n$, and using $d_{\mathrm H}^2(p,q)\le 2$ for probability mass functions under the present convention, we obtain
\begin{displaymath}
d_{\mathrm H}^2(\bar p_n,p_{G^\ast})\le C^2\varepsilon_n^2+2\,\Pi(U_n\mid Y_{1:n})\le C^2\varepsilon_n^2+2e^{-c\log^2 n}.
\end{displaymath}
Since $e^{-c\log^2 n}=o(\varepsilon_n^2)$, there exists a constant $C'>0$ such that, on $A_n$, $d_{\mathrm H}(\bar p_n,p_{G^\ast})\le C'\varepsilon_n$. Therefore
\begin{displaymath}
\P_{G^\ast}^n\!\left[d_{\mathrm H}(\bar p_n,p_{G^\ast})\ge C'\varepsilon_n\right]\le\P_{G^\ast}^n(A_n^c)\le\frac{1}{n}
\end{displaymath}
for all sufficiently large $n$. This proves the second assertion \eqref{pcr_dp_2}.

\subsection{Proof of Proposition \ref{pcrnew_1d}}\label{proof:pcrnew_1d}

Since $\Theta=[h_0,h]$ is a compact subset of $(0,+\infty)$, assumptions A1--A6 of \cite{MarTok(09)} hold. Following \cite{MarTok(09)}, we denote by $\mathbb F$ the class of probability measures on $[h_0,h]$ that are absolutely continuous with respect to Lebesgue measure, and by $\overline{\mathbb F}$ its closure in the weak topology. Then $G^*\in\mathbb F$ and
$$
G^\ast=\text{argmin}_{G\in \overline {\mathbb F}}KL(p_{G^\ast},p_G),
$$
where $KL$ denotes the Kullback--Leibler divergence. The claim for $\gamma\in (2/3,1)$ is then a direct consequence of Corollary 4.10 in \cite{MarTok(09)}. For $\gamma=1$, apply Theorem 4.8 in \cite{MarTok(09)}, with $a_n=\sum_{k=1}^n\alpha_k\asymp \log n$. Then
$\log n \;KL(p_{G^\ast},p_{G_{1,n}^{\tiny{[Q-B]}}} )$ converges to zero almost surely with respect to the probability measure under which the $Y_i$'s are i.i.d. according to $p_{G^\ast}$. Since the squared Hellinger distance is dominated by the Kullback--Leibler divergence, the claim follows.

%%%%%%%%%%%%%%%%%%%%%%%%%%%%%%%%
%%%%%%%%%%%%%%%%%%%%%%%%%%%%%%%%
%%%%%%%%%%%%%%%%%%%%%%%%%%%%%%%%
%%%%%%%%%%%%%%%%%%%%%%%%%%%%%%%%

\section{Proofs: $d$-dimensional setting, $d>1$}\label{app2}

\subsection{Auxiliary lemmas}

For $\boldsymbol y=(y_1,\ldots,y_d)\in\mathbb N_0^d$, write $|\boldsymbol y|:=y_1+\cdots+y_d$ and $\boldsymbol y!:=y_1!\cdots y_d!$. For $\boldsymbol\theta=(\theta_1,\ldots,\theta_d)\in[0,\infty)^d$, write $|\boldsymbol\theta|:=\theta_1+\cdots+\theta_d$ and $\boldsymbol\theta^{\boldsymbol y}:=\prod_{\ell=1}^d \theta_\ell^{y_\ell}$. The square root $\sqrt{\boldsymbol\theta}$ is understood componentwise.

\begin{lem}\label{lem:Poisson_kernel_d}
For $\boldsymbol\theta,\boldsymbol\lambda\in[0,\infty)^d$, let $q_{\boldsymbol\theta}$ and $q_{\boldsymbol\lambda}$ denote the $d$-dimensional Poisson kernels.
Then
\begin{displaymath}
\sum_{\boldsymbol y\in\mathbb N_0^d}\sqrt{q_{\boldsymbol\theta}(\boldsymbol y)q_{\boldsymbol\lambda}(\boldsymbol y)}=\exp\!\left\{-\frac{1}{2}\|\sqrt{\boldsymbol\theta}-\sqrt{\boldsymbol\lambda}\|^2\right\},
\end{displaymath}
where $\|\cdot\|$ denotes the Euclidean norm and the square root is understood componentwise. Consequently,
\begin{displaymath}
d_{\mathrm H}^2(q_{\boldsymbol\theta},q_{\boldsymbol\lambda})=2\left[1-\exp\!\left\{-\frac{1}{2}\|\sqrt{\boldsymbol\theta}-\sqrt{\boldsymbol\lambda}\|^2\right\}\right]\le\|\sqrt{\boldsymbol\theta}-\sqrt{\boldsymbol\lambda}\|^2,
\end{displaymath}
and
\begin{displaymath}
\|q_{\boldsymbol\theta}-q_{\boldsymbol\lambda}\|_1\le 2\|\sqrt{\boldsymbol\theta}-\sqrt{\boldsymbol\lambda}\|.
\end{displaymath}
\end{lem}

\begin{proof}
By the product structure of the $d$-dimensional Poisson kernel, $q_{\boldsymbol\theta}(\boldsymbol y)=\prod_{\ell=1}^d q_{\theta_\ell}(y_\ell)$. Since all terms are nonnegative, Tonelli's theorem gives
\begin{displaymath}
\sum_{\boldsymbol y\in\mathbb N_0^d}\sqrt{q_{\boldsymbol\theta}(\boldsymbol y)q_{\boldsymbol\lambda}(\boldsymbol y)}=\prod_{\ell=1}^d\sum_{y_\ell\ge 0}\sqrt{q_{\theta_\ell}(y_\ell)q_{\lambda_\ell}(y_\ell)}.
\end{displaymath}
By Lemma~\ref{lem:Poisson_kernel}, the $\ell$-th factor equals $\exp\!\left\{-\frac{(\sqrt{\theta_\ell}-\sqrt{\lambda_\ell})^2}{2}\right\}$. Multiplying over $\ell=1,\ldots,d$ gives
\begin{displaymath}
\sum_{\boldsymbol y\in\mathbb N_0^d}\sqrt{q_{\boldsymbol\theta}(\boldsymbol y)q_{\boldsymbol\lambda}(\boldsymbol y)}=\exp\!\left\{-\frac{1}{2}\sum_{\ell=1}^d
(\sqrt{\theta_\ell}-\sqrt{\lambda_\ell})^2\right\}=\exp\!\left\{-\frac{1}{2}\|\sqrt{\boldsymbol\theta}-\sqrt{\boldsymbol\lambda}\|^2\right\}.
\end{displaymath}
The expression for $d_{\mathrm H}^2(q_{\boldsymbol\theta},q_{\boldsymbol\lambda})$ follows from the definition of squared Hellinger distance. The upper bound follows from $1-e^{-y}\le y$ for $y\ge0$. Finally, as in the 1-dimensional case, $\|p-q\|_1\le 2d_{\mathrm H}(p,q)$ for probability mass functions $p$ and $q$. Hence
\begin{displaymath}
\|q_{\boldsymbol\theta}-q_{\boldsymbol\lambda}\|_1\le2d_{\mathrm H}(q_{\boldsymbol\theta},q_{\boldsymbol\lambda})\le2\|\sqrt{\boldsymbol\theta}-\sqrt{\boldsymbol\lambda}\|.
\end{displaymath}
This proves the claim.
\end{proof}

\begin{lem}\label{lem:Poisson_moment_matching_d}
Let $0<h_0\le h<+\infty$. For every $G\in\mathcal P([h_0,h]^d)$ and every integer $K\ge 0$, there exists a discrete probability measure 
\begin{displaymath}
G^{(K)} = \sum_{j=1}^{N_K}p_j\delta_{\boldsymbol\theta_j}, \qquad \boldsymbol\theta_j\in[h_0,h]^d, 
\end{displaymath}
with
\begin{displaymath}
N_K\le \binom{K+d}{d},
\end{displaymath}
such that
\begin{displaymath}
\int_{[h_0,h]^d}\boldsymbol\theta^{\boldsymbol r}\, G^{(K)}(d\boldsymbol\theta) = \int_{[h_0,h]^d}\boldsymbol\theta^{\boldsymbol r}\, G(d\boldsymbol\theta), \qquad |\boldsymbol r|\le K.
\end{displaymath}
Moreover, 
\begin{displaymath}
\|p_G-p_{G^{(K)}}\|_1 \le 2\sum_{s\ge K+1}\frac{(2dh)^s}{s!}.
\end{displaymath}
In particular, for every sufficiently small $\varepsilon>0$, there exists a discrete $\tilde G$ with at most 
\begin{displaymath}
N_\varepsilon \lesssim \left( \frac{\log(1/\varepsilon)}{\log\log(1/\varepsilon)} \right)^d
\end{displaymath}
support points such that $\|p_G-p_{\tilde G}\|_1\le\varepsilon$.
\end{lem}

\begin{proof}
If $K=0$, take $G^{(0)}=\delta_{\boldsymbol\theta_0}$ for any $\boldsymbol\theta_0\in[h_0,h]^d$; then the moment identity holds for $\boldsymbol r=\boldsymbol 0$, and the $L^1$ bound follows from the argument below, which uses only the equality of moments up to total degree $K$. Thus assume $K\ge 1$ for the application of \citet[Lemma A.1]{Gho(01)}. Apply \citet[Lemma A.1]{Gho(01)} to the collection of monomials $\boldsymbol\theta^{\boldsymbol r}$ with $1\le |\boldsymbol r|\le K$. The number of such nonconstant monomials is $\binom{K+d}{d}-1$. Hence \citet[Lemma A.1]{Gho(01)} gives a discrete probability measure $G^{(K)}$ with at most $\binom{K+d}{d}$ support points such that the displayed moment identities hold for $1\le |\boldsymbol r|\le K$. The identity for $\boldsymbol r=\boldsymbol 0$ follows because both measures are probability measures. Write
\begin{displaymath}
\mu_{\boldsymbol r}(G):=\int_{[h_{0},h]^d}\boldsymbol\theta^{\boldsymbol r}\,G(d\boldsymbol\theta).
\end{displaymath}
For $\boldsymbol y\in\mathbb N_0^d$,
\begin{displaymath}
p_G(\boldsymbol y)=\frac{1}{\boldsymbol y!}\int_{[h_{0},h]^d}e^{-|\boldsymbol\theta|}\boldsymbol\theta^{\boldsymbol y}\,G(d\boldsymbol\theta).
\end{displaymath}
Since $G$ is supported on $[h_{0},h]^d$,
\begin{displaymath}
\sum_{\boldsymbol m\in\mathbb N_0^d}\int_{[h_{0},h]^d}\frac{\boldsymbol\theta^{\boldsymbol y+\boldsymbol m}}{\boldsymbol y!\,\boldsymbol m!}\,G(d\boldsymbol\theta)\le
\frac{h^{|\boldsymbol y|}}{\boldsymbol y!}\sum_{\boldsymbol m\in\mathbb N_0^d}\frac{h^{|\boldsymbol m|}}{\boldsymbol m!}=\frac{h^{|\boldsymbol y|}e^{dh}}{\boldsymbol y!}<
+\infty.
\end{displaymath}
Therefore the following expansion is justified by absolute convergence:
\begin{displaymath}
e^{-|\boldsymbol\theta|}=\prod_{\ell=1}^d e^{-\theta_\ell}=\sum_{\boldsymbol m\in\mathbb N_0^d}(-1)^{|\boldsymbol m|}\frac{\boldsymbol\theta^{\boldsymbol m}}{\boldsymbol m!}.
\end{displaymath}
Thus
\begin{displaymath}
p_G(\boldsymbol y)=\frac{1}{\boldsymbol y!}\sum_{\boldsymbol m\in\mathbb N_0^d}\frac{(-1)^{|\boldsymbol m|}}{\boldsymbol m!}\mu_{\boldsymbol y+\boldsymbol m}(G),
\end{displaymath}
and the same expansion holds for $G^{(K)}$. Therefore
\begin{displaymath}
p_G(\boldsymbol y)-p_{G^{(K)}}(\boldsymbol y)=\frac{1}{\boldsymbol y!}\sum_{\boldsymbol m\in\mathbb N_0^d}\frac{(-1)^{|\boldsymbol m|}}{\boldsymbol m!}\left\{
\mu_{\boldsymbol y+\boldsymbol m}(G)-\mu_{\boldsymbol y+\boldsymbol m}(G^{(K)})\right\}.
\end{displaymath}
All terms with $|\boldsymbol y+\boldsymbol m|\le K$ vanish. Since both measures are supported
on $[h_{0},h]^d$,
\begin{displaymath}
\left|\mu_{\boldsymbol r}(G)-\mu_{\boldsymbol r}(G^{(K)})\right|\le2h^{|\boldsymbol r|}.
\end{displaymath}
Hence
\begin{displaymath}
\|p_G-p_{G^{(K)}}\|_1\le2\sum_{\boldsymbol y\in\mathbb N_0^d}\sum_{\boldsymbol m:\,|\boldsymbol y+\boldsymbol m|\ge K+1}\frac{h^{|\boldsymbol y+\boldsymbol m|}}
{\boldsymbol y!\,\boldsymbol m!}.
\end{displaymath}
Let $\boldsymbol r=\boldsymbol y+\boldsymbol m$. Then
\begin{displaymath}
\sum_{\boldsymbol 0\le \boldsymbol y\le\boldsymbol r}\frac{1}{\boldsymbol y!(\boldsymbol r-\boldsymbol y)!}=\frac{2^{|\boldsymbol r|}}{\boldsymbol r!}.
\end{displaymath}
Thus
\begin{displaymath}
\|p_G-p_{G^{(K)}}\|_1\le2\sum_{\boldsymbol r:\,|\boldsymbol r|\ge K+1}\frac{(2h)^{|\boldsymbol r|}}{\boldsymbol r!}.
\end{displaymath}
Using the multinomial identity
\begin{displaymath}
\sum_{\boldsymbol r:\,|\boldsymbol r|=s}\frac{1}{\boldsymbol r!}=\frac{d^s}{s!},
\end{displaymath}
we get
\begin{displaymath}
\|p_G-p_{G^{(K)}}\|_1\le2\sum_{s\geq K+1}\frac{(2dh)^s}{s!}.
\end{displaymath}
Stirling's formula implies that this tail is at most $\varepsilon$ for $K\lesssim\frac{\log(1/\varepsilon)}{\log\log(1/\varepsilon)}$. Since
\begin{displaymath}
N_K\le \binom{K+d}{d}\lesssim K^d,
\end{displaymath}
the final claim follows.
\end{proof}

\begin{lem}\label{lem:Poisson_L1_local_d}
Let $\tilde G=\sum_{j=1}^N p_j\delta_{\boldsymbol\theta_j}$ be a discrete probability measure on $[0,\infty)^d$, and let $\eta>0$ be such that
\begin{displaymath}
\|\sqrt{\boldsymbol\theta_j}-\sqrt{\boldsymbol\theta_k}\|>2\eta,
\qquad j\neq k.
\end{displaymath}
Define
\begin{displaymath}
U_j:=\left\{\boldsymbol\theta\in[0,\infty)^d:\|\sqrt{\boldsymbol\theta}-\sqrt{\boldsymbol\theta_j}\|\le \eta\right\},\qquad U:=\bigcup_{j=1}^N U_j.
\end{displaymath}
Then, for every probability measure $G$ on $[0,\infty)^d$,
\begin{displaymath}
\|p_G-p_{\tilde G}\|_1\le 2\eta+\sum_{j=1}^N |G(U_j)-p_j|+G(U^c).
\end{displaymath}
\end{lem}

\begin{proof}
The proof is the $d$-dimensional analogue of Lemma~\ref{lem:Poisson_L1_local}, with $|\sqrt{\theta}-\sqrt{\lambda}|$ replaced by $\|\sqrt{\boldsymbol\theta}-\sqrt{\boldsymbol\lambda}\|$. The separation assumption implies that the sets $U_1,\ldots,U_N$ are pairwise disjoint. Write $w_j:=G(U_j)$ and $r:=G(U^c)$ so that $\sum_{j=1}^N w_j+r=1$. For each $j$ with $w_j>0$, define
\begin{displaymath}
G_j(A):=\frac{G(A\cap U_j)}{w_j}.
\end{displaymath}
If $w_j=0$, let $G_j$ be any fixed probability measure on $[0,\infty)^d$; the corresponding term below is multiplied by $w_j=0$ and is irrelevant. Similarly, if $r>0$, define
\begin{displaymath}
G^c(A):=\frac{G(A\cap U^c)}{r},
\end{displaymath}
while if $r=0$, let $G^c$ be any fixed probability measure on $[0,\infty)^d$. Then
\begin{displaymath}
p_G=\sum_{j=1}^N w_jp_{G_j}+rp_{G^c}.
\end{displaymath}
Since $p_{\tilde G}=\sum_{j=1}^N p_jq_{\boldsymbol\theta_j}$, the triangle inequality gives
\begin{displaymath}
\|p_G-p_{\tilde G}\|_1\le\sum_{j=1}^N\|w_jp_{G_j}-p_jq_{\boldsymbol\theta_j}\|_1+r.
\end{displaymath}
Using $\|af-bg\|_1\le |a-b|+a\|f-g\|_1$ for $a,b\ge0$ and probability mass functions $f$ and $g$, we obtain
\begin{displaymath}
\|p_G-p_{\tilde G}\|_1\le\sum_{j=1}^N |w_j-p_j|+\sum_{j=1}^N w_j\|p_{G_j}-q_{\boldsymbol\theta_j}\|_1+r.
\end{displaymath}
By convexity of the $L^1$ norm and Lemma~\ref{lem:Poisson_kernel_d},
\begin{displaymath}
\|p_{G_j}-q_{\boldsymbol\theta_j}\|_1\le\int_{[0,\infty)^d}\|q_{\boldsymbol\theta}-q_{\boldsymbol\theta_j}\|_1\,G_j(d\boldsymbol\theta)\le2\eta,
\end{displaymath}
because $G_j(U_j)=1$ and $\|\sqrt{\boldsymbol\theta}-\sqrt{\boldsymbol\theta_j}\|\le\eta$ on $U_j$. Hence
\begin{displaymath}
\sum_{j=1}^N w_j\|p_{G_j}-q_{\boldsymbol\theta_j}\|_1\le2\eta\sum_{j=1}^N w_j\le2\eta.
\end{displaymath}
Therefore
\begin{displaymath}
\|p_G-p_{\tilde G}\|_1\le\sum_{j=1}^N |G(U_j)-p_j|+2\eta+G(U^c),
\end{displaymath}
as claimed.
\end{proof}

\begin{lem}\label{lem:Poisson_WS_d}
For probability mass functions $p$ and $q$ on $\mathbb N_0^d$, define $K(p,q):=\sum_{\boldsymbol y\in\mathbb N_0^d} p(\boldsymbol y) \log\frac{p(\boldsymbol y)}{q(\boldsymbol y)}$ and $V(p,q):=\sum_{\boldsymbol y\in\mathbb N_0^d} p(\boldsymbol y) \left(\log\frac{p(\boldsymbol y)}{q(\boldsymbol y)}\right)^2$ with the usual convention that these quantities are infinite if $q(\boldsymbol y)=0$ for some $\boldsymbol y$ such that $p(\boldsymbol y)>0$. For $\delta>0$, define the Kullback-Leibler type ball
\begin{displaymath}
B(\delta,p_{G^\ast}) := \left\{ p:\ K(p_{G^\ast},p)\le\delta^2,\quad V(p_{G^\ast},p)\le\delta^2 \right\}. 
\end{displaymath}
Let $0<h_0\le h_0^\ast<h^\ast\le h<+\infty$, and assume that $G^\ast([h_0^\ast,h^\ast]^d)=1$. Fix constants $0<b\le h<+\infty$ and $\alpha\in(0,1]$. Then there exist constants $C>0$ and $\eta_0>0$, depending only on $(d,h_0,h,h_0^\ast,h^\ast,b,\alpha)$, such that for every $G\in\mathcal P([h_0,h]^d)$ satisfying $G([b,h]^d)\ge\alpha$, the implication $\eta:=d_{\mathrm H}(p_{G^\ast},p_G)<\eta_0$ yields 
\begin{displaymath}
K(p_{G^\ast},p_G)\le C\eta^2\log\frac{1}{\eta}, \qquad V(p_{G^\ast},p_G)\le C\eta^2\left(\log\frac{1}{\eta}\right)^2. 
\end{displaymath}
Consequently, there exists $A>0$, depending only on $(d,h_0,h,h_0^\ast,h^\ast,b,\alpha)$, such that for every $G\in\mathcal P([h_0,h]^d)$ satisfying $G([b,h]^d)\ge\alpha$, whenever $\|p_G-p_{G^\ast}\|_1\le\delta$ for some sufficiently small $\delta>0$, one has 
\begin{displaymath}
p_G\in B\!\left( A\delta^{1/2}\log\frac{1}{\delta}, p_{G^\ast} \right).
\end{displaymath}
\end{lem}

\begin{proof}
For every $\boldsymbol y\in\mathbb N_0^d$,
\begin{displaymath}
p_{G^\ast}(\boldsymbol y)=\frac{1}{\boldsymbol y!}\int_{[h_{0}^{\ast},h^\ast]^d}e^{-|\boldsymbol\theta|}\boldsymbol\theta^{\boldsymbol y}\,G^\ast(d\boldsymbol\theta)\le\frac{(h^\ast)^{|\boldsymbol y|}}{\boldsymbol y!}.
\end{displaymath}
Since $G([b,h]^d)\ge\alpha$,
\begin{displaymath}
p_G(\boldsymbol y)\ge\frac{1}{\boldsymbol y!}\int_{[b,h]^d}e^{-|\boldsymbol\theta|}\boldsymbol\theta^{\boldsymbol y}\,G(d\boldsymbol\theta)\ge\alpha e^{-dh}\frac{b^{|\boldsymbol y|}}{\boldsymbol y!}.
\end{displaymath}
Therefore
\begin{displaymath}
\frac{p_{G^\ast}(\boldsymbol y)}{p_G(\boldsymbol y)}\le\alpha^{-1}e^{dh}\left(\frac{h^\ast}{b}\right)^{|\boldsymbol y|}.
\end{displaymath}
Let $C_0:=\alpha^{-1}e^{dh}$ and $r:=\frac{h^\ast}{b}$. For any $\delta_\ast\in(0,1]$,
\begin{displaymath}
\sum_{\boldsymbol y\in\mathbb N_0^d}\left(\frac{p_{G^\ast}(\boldsymbol y)}{p_G(\boldsymbol y)}\right)^{\delta_\ast}p_{G^\ast}(\boldsymbol y)\le C_0^{\delta_\ast}
\E_{G^\ast}\!\left[r^{\delta_\ast|\boldsymbol Y|}\right].
\end{displaymath}
If $\boldsymbol Y\sim p_{G^\ast}$, then for every $t>0$,
\begin{displaymath}
\E_{G^\ast}(t^{|\boldsymbol Y|})=\int_{[h_{0}^{\ast},h^\ast]^d}\exp\{(t-1)|\boldsymbol\theta|\}\,G^\ast(d\boldsymbol\theta)\le\exp\{dh^\ast |t-1|\}.
\end{displaymath}
Thus the integrability condition in \citet[Theorem 5]{Won(95)} holds uniformly over all $G\in\mathcal P([h_{0},h]^d)$ satisfying $G([b,h]^d)\ge\alpha$. Arguing as in the proof of \citet[Lemma 4.1]{Gho(01)}, we obtain constants $C>0$ and $\eta_0>0$, depending only on $(d,h_0,h,h_0^\ast,h^\ast,b,\alpha)$, such that, for
$\eta=d_{\mathrm H}(p_{G^\ast},p_G)<\eta_0$,
\begin{displaymath}
K(p_{G^\ast},p_G)\le C\eta^2\log\frac{1}{\eta}, \qquad V(p_{G^\ast},p_G)\le C\eta^2\left(\log\frac{1}{\eta}\right)^2.
\end{displaymath}
Finally, if $\|p_G-p_{G^\ast}\|_1\le\delta$, then $d_{\mathrm H}^2(p_G,p_{G^\ast})\le\|p_G-p_{G^\ast}\|_1\le\delta$. Hence $d_{\mathrm H}(p_G,p_{G^\ast})\le\delta^{1/2}$. For all sufficiently small $\delta$, this is smaller than $\eta_0$. Therefore
\begin{displaymath}
K(p_{G^\ast},p_G)\le C\delta\log\frac{1}{\delta^{1/2}}\lesssim\delta\log\frac{1}{\delta},
\end{displaymath}
and
\begin{displaymath}
V(p_{G^\ast},p_G)\le C\delta\left(\log\frac{1}{\delta^{1/2}}\right)^2\lesssim\delta\left(\log\frac{1}{\delta}\right)^2.
\end{displaymath}
Choosing $A>0$ sufficiently large, we obtain
\begin{displaymath}
K(p_{G^\ast},p_G)\le A^2\delta\left(\log\frac{1}{\delta}\right)^2,\qquad V(p_{G^\ast},p_G)\le A^2\delta\left(\log\frac{1}{\delta}\right)^2.
\end{displaymath}
Equivalently,
\begin{displaymath}
p_G\in B\!\left(A\delta^{1/2}\log\frac{1}{\delta},\ p_{G^\ast}\right).
\end{displaymath}
This proves the lemma.
\end{proof}

\begin{lem}\label{lem:Poisson_bracketing_entropy_d}
Let $0<h_0\le h<+\infty$ and define \[ \mathcal G_{h,d}:= \{p_G:\ G\in\mathcal P([h_0,h]^d)\}. \] There exists a constant $C>0$, depending only on $(d,h)$, such that, for all sufficiently small $\varepsilon>0$,
\begin{displaymath}
\log N_{[]}(\varepsilon,\mathcal G_{h,d},d_{\mathrm H}) \le C\frac{(\log(1/\varepsilon))^{d+1}} {(\log\log(1/\varepsilon))^d}.
 \end{displaymath}
\end{lem}

\begin{proof}
For every $G\in\mathcal P([h_{0},h]^d)$ and every $\boldsymbol y\in\mathbb N_0^d$,
\begin{displaymath}
p_G(\boldsymbol y)=\int_{[h_{0},h]^d}q_{\boldsymbol\theta}(\boldsymbol y)\,G(d\boldsymbol\theta)\le\frac{h^{|\boldsymbol y|}}{\boldsymbol y!}=:b_h(\boldsymbol y).
\end{displaymath}
Choose $T=T_\varepsilon$ such that $\sum_{|\boldsymbol y|>T} b_h(\boldsymbol y)\le\frac{\varepsilon^2}{8}$.  Since $\sum_{|\boldsymbol y|=s}\frac{h^s}{\boldsymbol y!}
=\frac{(dh)^s}{s!}$, Stirling's formula implies that one can take
\begin{displaymath}
T_\varepsilon\lesssim\frac{\log(1/\varepsilon)}{\log\log(1/\varepsilon)}.
\end{displaymath}
Let $\mathcal I_T:=\{\boldsymbol y\in\mathbb N_0^d:\ |\boldsymbol y|\le T\}$ and $m_T:=|\mathcal I_T|$. Then $m_T\lesssim T^d$. Set $\delta:=\frac{\varepsilon^2}{8m_T}$. For each $\boldsymbol y\in\mathcal I_T$, bracket the coordinate
$p_G(\boldsymbol y)\in[0,1]$ by intervals of length $\delta$. The number of such coordinatewise
brackets is at most
\begin{displaymath}
\left(\frac{1}{\delta}+1\right)^{m_T}.
\end{displaymath}
For $\boldsymbol y\notin\mathcal I_T$, use the bracket $[0,b_h(\boldsymbol y)]$. For every resulting bracket $[l,u]$, where $l$ and $u$ are nonnegative summable functions on $\mathbb N_0^d$, we have
\begin{displaymath}
(\sqrt{u(\boldsymbol y)}-\sqrt{l(\boldsymbol y)})^2\le u(\boldsymbol y)-l(\boldsymbol y).
\end{displaymath}
Therefore
\begin{displaymath}
d_{\mathrm H}^2(l,u)\le\sum_{\boldsymbol y\in\mathcal I_T}\{u(\boldsymbol y)-l(\boldsymbol y)\}+\sum_{|\boldsymbol y|>T}b_h(\boldsymbol y)\le\frac{\varepsilon^2}{8}+\frac{\varepsilon^2}{8}=\frac{\varepsilon^2}{4}.
\end{displaymath}
Hence $d_{\mathrm H}(l,u)\le\varepsilon/2$, and the constructed brackets are in particular $\varepsilon$-brackets for $\mathcal G_{h,d}$. Thus
\begin{displaymath}
\log N_{[]}(\varepsilon,\mathcal G_{h,d},d_{\mathrm H})\lesssim m_T\log\left(\frac{2}{\delta}\right)\lesssim m_T\log\left(\frac{m_T}{\varepsilon^2}\right).
\end{displaymath}
Using $m_T\lesssim T_\varepsilon^d$ and
\begin{displaymath}
T_\varepsilon\lesssim\frac{\log(1/\varepsilon)}{\log\log(1/\varepsilon)},
\end{displaymath}
we obtain
\begin{displaymath}
\log N_{[]}(\varepsilon,\mathcal G_{h,d},d_{\mathrm H})\le C\frac{(\log(1/\varepsilon))^{d+1}}{(\log\log(1/\varepsilon))^d}.
\end{displaymath}
This proves the lemma.
\end{proof}

\begin{lem}\label{lem:Poisson_local_sets_d}
Let $0<h_0\le h_0^\ast<h^\ast\le h<+\infty$, and assume that $G^\ast([h_0^\ast,h^\ast]^d)=1$. Let $\Pi$ denote the law of a Dirichlet process $\operatorname{DP}(1,H)$, where $H$ is a finite measure on $[h_0,h]^d$ with density $a$ satisfying $0<m_H\le a(\boldsymbol\theta)\le M_H<+\infty, \qquad \boldsymbol\theta\in[h_0,h]^d$. Let $\varepsilon_{n,d}:=\frac{(\log n)^{(d+1)/2}}{\sqrt n}$. Then there exist constants $L>0$ and $C_{\mathrm{KL}}>0$, depending only on $(d,h_0,h,h_0^\ast,h^\ast,G^\ast,\alpha)$, and sets $B_n\subset\mathcal P([h_0,h]^d)$ such that 
\begin{displaymath}
\sup_{G\in B_n}K(p_{G^\ast},p_G)\le L^2\varepsilon_{n,d}^2, \qquad \sup_{G\in B_n}V(p_{G^\ast},p_G)\le L^2\varepsilon_{n,d}^2,
\end{displaymath}
and 
\begin{displaymath}
\Pi(B_n)\ge\exp(-C_{\mathrm{KL}}n\varepsilon_{n,d}^2)
\end{displaymath}
for all sufficiently large $n$.
\end{lem}

\begin{proof}
Choose $b>0$ such that $3b<h_0^\ast$, and choose $\alpha_0\in(0,1/4)$. Then $G^\ast([3b,h^\ast]^d)=1>4\alpha_0$. Fix a sufficiently small $\varepsilon>0$. Choose $K=K_\varepsilon$ such that
\begin{displaymath}
2\sum_{s\geq K+1}\frac{(2dh)^s}{s!}\le\frac{\varepsilon}{2}.
\end{displaymath}
Apply \citet[Lemma A.1]{Gho(01)} on the compact set $[h_0^\ast,h^\ast]^d$ to all monomials $\boldsymbol\theta^{\boldsymbol r}$ with $1\le |\boldsymbol r|\le K$. Then there exists a discrete probability measure
\begin{displaymath}
\tilde{\tilde G}^\ast=\sum_{j=1}^{M_\varepsilon}\widetilde p_j\delta_{\widetilde{\boldsymbol\theta}_j},\qquad\widetilde{\boldsymbol\theta}_j\in[h_0^\ast,h^\ast]^d,
\end{displaymath}
with
\begin{displaymath}
M_\varepsilon\lesssim K^d\lesssim\left(\frac{\log(1/\varepsilon)}{\log\log(1/\varepsilon)}\right)^d,
\end{displaymath}
such that
\begin{displaymath}
\int_{[h_0^\ast,h^\ast]^d}\boldsymbol\theta^{\boldsymbol r}\,\tilde{\tilde G}^\ast(d\boldsymbol\theta)=\int_{[h_0^\ast,h^\ast]^d}\boldsymbol\theta^{\boldsymbol r}\,
G^\ast(d\boldsymbol\theta),\qquad 1\le |\boldsymbol r|\le K.
\end{displaymath}
Since both measures are probability measures, the identity also holds for $\boldsymbol r=\boldsymbol 0$. Moreover, $\tilde{\tilde G}^\ast([3b,h^\ast]^d)=1>4\alpha_0$. Therefore, by the proof of Lemma~\ref{lem:Poisson_moment_matching_d},
\begin{displaymath}
\|p_{G^\ast}-p_{\tilde{\tilde G}^\ast}\|_1\le\frac{\varepsilon}{2}.
\end{displaymath}
Now move each support point of $\tilde{\tilde G}^\ast$ to a point of an $\varepsilon$-net in the $\sqrt{\boldsymbol\theta}$-scale, chosen so that every point of $[\sqrt{h_{0}},\sqrt h]^d$ is within Euclidean distance $\varepsilon/2$ of the net and distinct retained net points are separated by more than $\varepsilon/2$. Merge duplicate atoms. The resulting
measure
\begin{displaymath}
\tilde G^\ast=\sum_{j=1}^{N_\varepsilon}p_j\delta_{\boldsymbol\theta_j}
\end{displaymath}
satisfies
\begin{displaymath}
N_\varepsilon\lesssim\left(\frac{\log(1/\varepsilon)}{\log\log(1/\varepsilon)}\right)^d,
\end{displaymath}
and
\begin{displaymath}
\|\sqrt{\boldsymbol\theta_j}-\sqrt{\boldsymbol\theta_k}\|>\frac{\varepsilon}{2},\qquad j\neq k.
\end{displaymath}
Moreover, before merging duplicate atoms, each support point has been moved by at most $\varepsilon/2$ in the $\sqrt{\boldsymbol\theta}$-scale. Hence, by Lemma~\ref{lem:Poisson_kernel_d},
\begin{displaymath}
\|p_{\tilde{\tilde G}^\ast}-p_{\tilde G^\ast}\|_1\le\varepsilon.
\end{displaymath}
Therefore
\begin{displaymath}
\|p_{G^\ast}-p_{\tilde G^\ast}\|_1\le\frac{3\varepsilon}{2}.
\end{displaymath}
Furthermore, since $\tilde{\tilde G}^\ast([3b,h^\ast]^d)=1>4\alpha_0$, and since each support point is moved by at most $\varepsilon/2$ in the $\sqrt{\boldsymbol\theta}$-scale, every support point of $\tilde{\tilde G}^\ast$ is moved into $[b,h]^d$ for all sufficiently small $\varepsilon$. Merging duplicate atoms does not change the total mass.
Thus $\tilde G^\ast([b,h]^d)>4\alpha_0$. Define
\begin{displaymath}
U_j:=\left\{\boldsymbol\theta\in[h_{0},h]^d:\|\sqrt{\boldsymbol\theta}-\sqrt{\boldsymbol\theta_j}\|\le\frac{\varepsilon}{4}\right\},\qquad U:=\bigcup_{j=1}^{N_\varepsilon}U_j,
\end{displaymath}
and let
\begin{displaymath}
\mathcal E_\varepsilon:=\left\{G\in\mathcal P([h_{0},h]^d):\sum_{j=1}^{N_\varepsilon}|G(U_j)-p_j|\le\varepsilon,\quad G(U^c)\le\varepsilon\right\}.
\end{displaymath}
By Lemma~\ref{lem:Poisson_L1_local_d}, every $G\in\mathcal E_\varepsilon$ satisfies
\begin{displaymath}
\|p_G-p_{\tilde G^\ast}\|_1\le 2(\varepsilon/4)+\varepsilon+\varepsilon=\frac{5\varepsilon}{2}.
\end{displaymath}
Therefore
\begin{displaymath}
\|p_G-p_{G^\ast}\|_1\le 4\varepsilon.
\end{displaymath}
We next show that $G\in\mathcal E_\varepsilon$ implies $G([b/2,h]^d)\ge\alpha_0$ for all sufficiently small $\varepsilon$. If $\boldsymbol\theta_j\in[b,h]^d$, then $U_j\subset [b/2,h]^d$ for all sufficiently small $\varepsilon$. Hence, for $G\in\mathcal E_\varepsilon$,
\begin{displaymath}
G([b/2,h]^d)\ge\sum_{\boldsymbol\theta_j\in[b,h]^d}G(U_j)\ge\sum_{\boldsymbol\theta_j\in[b,h]^d}p_j-\varepsilon=\tilde G^\ast([b,h]^d)-\varepsilon>4\alpha_0-\varepsilon>\alpha_0,
\end{displaymath}
provided $\varepsilon$ is sufficiently small. Applying Lemma~\ref{lem:Poisson_WS_d} with $b/2$ in place of $b$ and $\alpha_0$ in place of
$\alpha$ gives
\begin{displaymath}
\mathcal E_\varepsilon\subset B\!\left(A\varepsilon^{1/2}\log\frac{1}{\varepsilon},p_{G^\ast}\right)
\end{displaymath}
for a constant $A>0$ depending only on $(d,h_0,h,h_0^\ast,h^\ast,b,\alpha_0)$. Now set $B_n:=\mathcal E_{1/n}$. Applying Lemma~\ref{lem:Poisson_WS_d} with $\delta=4/n$ gives $\mathcal{E}_{1/n}\subset B\!\left(A(4/n)^{1/2}\log(n/4),\,p_{G^\ast}\right)$. The resulting radius satisfies
\begin{displaymath}
(4/n)^{1/2}\log(n/4)\asymp \frac{\log n}{\sqrt{n}}=\frac{(\log n)^1}{n^{1/2}}\le\frac{(\log n)^{(d+1)/2}}{n^{1/2}}=\varepsilon_{n,d},
\end{displaymath}
where the inequality holds for all $d\ge 1$ and all $n\ge 3$. Hence, for a suitable constant $A_1>0$, $B_n\subset B(A_1\varepsilon_{n,d},p_{G^\ast})$. It remains to lower bound $\Pi(B_n)$. Since $H$ has density bounded below on $[h_{0},h]^d$, there exists $c>0$ such that, for $\varepsilon=1/n$,
\begin{displaymath}
H(U_j)\ge c n^{-2d},\qquad j=1,\ldots,N_{1/n},
\end{displaymath}
for all sufficiently large $n$. Indeed, in each coordinate, a ball of radius $1/(4n)$ in the $\sqrt{\theta_\ell}$-scale contains an interval in the $\theta_\ell$-scale of length at least of order $n^{-2}$, uniformly over centers in $[h_{0},h]$. Hence each $U_j$ has Lebesgue measure at least of order $n^{-2d}$. Applying \citet[Lemma A.2]{Gho(01)} to the partition $U_1,\ldots,U_{N_{1/n}},U^c$, after subdividing $U^c$ if necessary, yields
\begin{displaymath}
\Pi(B_n)\ge\exp\{-c_1N_{1/n}\log n\}.
\end{displaymath}
Since
\begin{displaymath}
N_{1/n}\lesssim\left(\frac{\log n}{\log\log n}\right)^d,
\end{displaymath}
we have
\begin{displaymath}
N_{1/n}\log n\lesssim\frac{(\log n)^{d+1}}{(\log\log n)^d}\le C n\varepsilon_{n,d}^2.
\end{displaymath}
Thus
\begin{displaymath}
\Pi(B_n)\ge \exp(-C_{\mathrm{KL},1}n\varepsilon_{n,d}^2).
\end{displaymath}
Taking $L:=A_1$, $C_{\mathrm{KL}}:=C_{\mathrm{KL},1}$ gives
\begin{displaymath}
\sup_{G\in B_n}K(p_{G^\ast},p_G)\le L^2\varepsilon_{n,d}^2,\qquad \sup_{G\in B_n}V(p_{G^\ast},p_G)\le L^2\varepsilon_{n,d}^2,
\end{displaymath}
and
\begin{displaymath}
\Pi(B_n)\ge \exp(-C_{\mathrm{KL}}n\varepsilon_{n,d}^2)
\end{displaymath}
for all sufficiently large $n$. This proves the lemma.
\end{proof}

\begin{lem}\label{lem:Poisson_HP_exp_moment_d}
Assume the setup of Lemma~\ref{lem:Poisson_local_sets_d}, and let $B_n$ be the local set constructed therein. Then there exist constants $\lambda_0>0$ and $C_0<+\infty$, depending only on $(d,h_0,h,h_0^\ast,h^\ast,G^\ast)$, such that, for all sufficiently large $n$, 
\begin{displaymath}
\sup_{G\in B_n} \E_{G^\ast}\!\left[ \exp\!\left\{ \lambda_0 \left| \log\frac{p_G(\boldsymbol Y_1)}{p_{G^\ast}(\boldsymbol Y_1)} \right| \right\} \right] \le C_0.
\end{displaymath}
\end{lem}

\begin{proof}
Choose $b>0$ such that $3b<h_0^\ast$, and choose $\alpha_0\in(0,1/4)$. Then $G^\ast([3b,h^\ast]^d)=1>4\alpha_0$. By the construction in the proof of Lemma~\ref{lem:Poisson_local_sets_d}, every $G\in B_n$ satisfies
\begin{displaymath}
G([b/2,h]^d)\ge \alpha_0
\end{displaymath}
for all sufficiently large $n$. Hence, as in the proof of Lemma~\ref{lem:Poisson_WS_d}, there exist constants $C_1,C_2>0$, depending only on $(d,h_0,h,h_0^\ast,h^\ast,b,\alpha_0)$, such that
\begin{displaymath}
\frac{p_{G^\ast}(\boldsymbol y)}{p_G(\boldsymbol y)}\le C_1 C_2^{|\boldsymbol y|},\qquad\boldsymbol y\in\mathbb N_0^d,
\end{displaymath}
uniformly over $G\in B_n$. Moreover, since $G^\ast([3b,h^\ast]^d)=1$, for every $\boldsymbol y\in\mathbb N_0^d$,
\begin{displaymath}
p_{G^\ast}(\boldsymbol y)=\frac{1}{\boldsymbol y!}\int_{[h_0^\ast,h^\ast]^d}e^{-|\boldsymbol\theta|}\boldsymbol\theta^{\boldsymbol y}\,G^\ast(d\boldsymbol\theta)\ge
e^{-dh^\ast}\frac{(3b)^{|\boldsymbol y|}}{\boldsymbol y!}.
\end{displaymath}
On the other hand, for every $G\in\mathcal P([h_0,h]^d)$,
\begin{displaymath}
p_G(\boldsymbol y)\le\frac{h^{|\boldsymbol y|}}{\boldsymbol y!}.
\end{displaymath}
Therefore
\begin{displaymath}
\frac{p_G(\boldsymbol y)}{p_{G^\ast}(\boldsymbol y)}\le e^{dh^\ast}\left(\frac{h}{3b}\right)^{|\boldsymbol y|},\qquad\boldsymbol y\in\mathbb N_0^d .
\end{displaymath}
Combining the two ratio bounds, there exist constants $a_0,a_1>0$, depending only on $(d,h_0,h,h_0^\ast,h^\ast)$, such that
\begin{displaymath}
\left|\log\frac{p_G(\boldsymbol y)}{p_{G^\ast}(\boldsymbol y)}\right|\le a_0+a_1|\boldsymbol y|,\qquad\boldsymbol y\in\mathbb N_0^d,
\end{displaymath}
uniformly over $G\in B_n$ and all sufficiently large $n$. Since $\boldsymbol Y_1\sim p_{G^\ast}$, for every $t>0$,
\begin{displaymath}
\E_{G^\ast}(e^{t|\boldsymbol Y_1|})=\int_{[h_0^\ast,h^\ast]^d}\exp\{(e^t-1)|\boldsymbol\theta|\}\,G^\ast(d\boldsymbol\theta)\le\exp\{dh^\ast(e^t-1)\}<+\infty.
\end{displaymath}
Thus, for any fixed $\lambda_0>0$,
\begin{displaymath}
\sup_{G\in B_n}\E_{G^\ast}\!\left[\exp\!\left\{\lambda_0\left|\log\frac{p_G(\boldsymbol Y_1)}{p_{G^\ast}(\boldsymbol Y_1)}\right|\right\}\right]\le e^{\lambda_0 a_0}
\E_{G^\ast}(e^{\lambda_0 a_1|\boldsymbol Y_1|})<+\infty .
\end{displaymath}
This proves the lemma.
\end{proof}

The numerator-event lemma and the denominator-event lemma used in the $1$-dimensional proof remain valid verbatim in the present $d$-dimensional setting, after replacing $\mathbb N_0$ by $\mathbb N_0^d$, $\mathcal P([h_0,h])$ by $\mathcal P([h_0,h]^d)$, $Y_i$ by $\boldsymbol Y_i$, and $\varepsilon_n$ by $\varepsilon_{n,d}$. Their proofs
use only the bracketing entropy bound, the KL and $V$ bounds on the local set, and the uniform exponential moment condition. The key analytic condition required by both lemmas is that the $\delta_n$-bracketing number satisfies $\log N_n=o(n\varepsilon_{n,d}^2)$ for $\delta_n\asymp\varepsilon_{n,d}^2$. By Lemma~\ref{lem:Poisson_bracketing_entropy_d},
\begin{displaymath}
\log N_n\lesssim\frac{(\log(1/\delta_n))^{d+1}}{(\log\log(1/\delta_n))^d}\asymp\frac{(\log n)^{d+1}}{(\log\log n)^d}=\frac{n\varepsilon_{n,d}^2}{(\log\log n)^d}=o(n\varepsilon_{n,d}^2),
\end{displaymath}
since $(\log\log n)^{-d}\to0$. This verifies the condition and confirms that the lemmas apply with $\varepsilon_{n,d}$ in place of $\varepsilon_n$.

\subsection{Proof of Proposition~\ref{pcrdp_d}}

For notational simplicity, write $\varepsilon_{n,d}:=\frac{(\log n)^{(d+1)/2}}{\sqrt n}$. For a constant $C>0$ to be chosen later, set $U_n:=\{G\in\mathcal P([h_0,h]^d):
d_{\mathrm H}(p_G,p_{G^\ast})\ge C\varepsilon_{n,d}\}$, and define $R_n(G):=\prod_{i=1}^n\frac{p_G(\boldsymbol Y_i)}{p_{G^\ast}(\boldsymbol Y_i)}$. Then
\begin{displaymath}
\Pi(U_n\mid \boldsymbol Y_{1:n})=\frac{\int_{U_n}R_n(G)\,\Pi(dG)}{\int_{\mathcal P([h_0,h]^d)}R_n(G)\,\Pi(dG)}.
\end{displaymath}
By Lemma~\ref{lem:Poisson_local_sets_d}, there exists a local set $B_n\subset\mathcal P([h_0,h]^d)$ and constants $L>0$ and $C_{\mathrm{KL}}>0$, depending only on
$(d,h_0,h,h_0^\ast,h^\ast,G^\ast,\alpha)$, such that, for all sufficiently large $n$,
\begin{displaymath}
\sup_{G\in B_n}K(p_{G^\ast},p_G)\le L^2\varepsilon_{n,d}^2,\qquad\sup_{G\in B_n}V(p_{G^\ast},p_G)\le L^2\varepsilon_{n,d}^2,
\end{displaymath}
and
\begin{displaymath}
\Pi(B_n)\ge\exp(-C_{\mathrm{KL}}n\varepsilon_{n,d}^2).
\end{displaymath}
Moreover, by Lemma~\ref{lem:Poisson_HP_exp_moment_d}, the sets $B_n$ satisfy the uniform exponential moment condition required in the denominator-event argument. Fix, for instance, $D_0=1$ and define
\begin{displaymath}
D_n:=\left\{\int_{\mathcal P([h_0,h]^d)}R_n(G)\,\Pi(dG)\ge\exp\{-(L^2+D_0)n\varepsilon_{n,d}^2\}\Pi(B_n)\right\}.
\end{displaymath}
By the $d$-dimensional analogue of Lemma~\ref{lem:Poisson_HP_Dn}, there exists a constant $C_{\mathrm{den}}>0$, depending only on $(d,h_0,h,h_0^\ast,h^\ast,G^\ast,\alpha)$, such that, for all sufficiently large $n$,
\begin{displaymath}
\P_{G^\ast}^n(D_n^c)\le \frac{1}{n^2},
\end{displaymath}
and, on $D_n$,
\begin{displaymath}
\int_{\mathcal P([h_0,h]^d)}R_n(G)\,\Pi(dG)\ge\exp(-C_{\mathrm{den}}n\varepsilon_{n,d}^2).
\end{displaymath}
Now choose $D_1>C_{\mathrm{den}}$. By the $d$-dimensional analogue of Lemma~\ref{lem:Poisson_HP_En}, there exists a sufficiently large constant $C>0$, depending only on $(d,h_0,h,h_0^\ast,h^\ast,G^\ast,\alpha)$, and, for all sufficiently large $n$, an event $E_n$ such that
\begin{displaymath}
\P_{G^\ast}^n(E_n^c)\le \frac{1}{n^2},
\end{displaymath}
and, on $E_n$,
\begin{displaymath}
\int_{U_n}R_n(G)\,\Pi(dG)\le\exp(-D_1n\varepsilon_{n,d}^2).
\end{displaymath}
Therefore, on $D_n\cap E_n$,
\begin{displaymath}
\Pi(U_n\mid \boldsymbol Y_{1:n})=\frac{\int_{U_n}R_n(G)\,\Pi(dG)}{\int_{\mathcal P([h_0,h]^d)}R_n(G)\,\Pi(dG)}\le\exp\{(C_{\mathrm{den}}-D_1)n\varepsilon_{n,d}^2\}.
\end{displaymath}
Let $c:=D_1-C_{\mathrm{den}}>0$. Since $n\varepsilon_{n,d}^2=(\log n)^{d+1}$, on $D_n\cap E_n$ we obtain $\Pi(U_n\mid \boldsymbol Y_{1:n})\le\exp\{-c(\log n)^{d+1}\}$.
Hence, for all sufficiently large $n$,
\begin{displaymath}
\P_{G^\ast}^n\!\left[\Pi(U_n\mid \boldsymbol Y_{1:n})>\exp\{-c(\log n)^{d+1}\}\right]\le\P_{G^\ast}^n(D_n^c)+\P_{G^\ast}^n(E_n^c)\le\frac{2}{n^2}\le\frac{1}{n}.
\end{displaymath}
This proves the first assertion \eqref{pcr_dp_d_1}. Now define the posterior mean marginal probability mass function
\begin{displaymath}
\bar p_n(\boldsymbol y):=
\int_{\mathcal P([h_0,h]^d)}
p_G(\boldsymbol y)\,\Pi(dG\mid \boldsymbol Y_{1:n}),\qquad\boldsymbol y\in\mathbb N_0^d.
\end{displaymath}
Let $A_n:=\{\Pi(U_n\mid \boldsymbol Y_{1:n})\le\exp\{-c(\log n)^{d+1}\}\}$. By the first part,
\begin{displaymath}
\P_{G^\ast}^n(A_n)\ge 1-\frac{1}{n}
\end{displaymath}
for all sufficiently large $n$. On $A_n$, by convexity of squared Hellinger distance in its first argument,
\begin{displaymath}
d_{\mathrm H}^2(\bar p_n,p_{G^\ast})\le\int_{\mathcal P([h_0,h]^d)}d_{\mathrm H}^2(p_G,p_{G^\ast})\,\Pi(dG\mid \boldsymbol Y_{1:n}).
\end{displaymath}
Splitting the integral over $U_n^c$ and $U_n$, and using $d_{\mathrm H}^2(p,q)\le 2$ for probability mass functions under the present convention, we obtain
\begin{displaymath}
d_{\mathrm H}^2(\bar p_n,p_{G^\ast})\le C^2\varepsilon_{n,d}^2+2\,\Pi(U_n\mid \boldsymbol Y_{1:n})\le C^2\varepsilon_{n,d}^2+2e^{-c(\log n)^{d+1}}.
\end{displaymath}
Since
\begin{displaymath}
e^{-c(\log n)^{d+1}}=o(\varepsilon_{n,d}^2),
\end{displaymath}
there exists a constant $C'>0$ such that, on $A_n$,
\begin{displaymath}
d_{\mathrm H}(\bar p_n,p_{G^\ast})\le C'\varepsilon_{n,d}.
\end{displaymath}
Therefore
\begin{displaymath}
\P_{G^\ast}^n\!\left[d_{\mathrm H}(\bar p_n,p_{G^\ast})\ge C'\varepsilon_{n,d}\right]\le\P_{G^\ast}^n(A_n^c)\le\frac{1}{n}
\end{displaymath}
for all sufficiently large $n$. This proves the second assertion \eqref{pcr_dp_d_2}.

\subsection{Proof of Proposition \ref{pcrnew_d}}\label{proof:pcrnew_d}
%%%%%%%%%%%%%%%%%%%%%%%%%%%%%%%%%%%%

Since $\Theta=[h_0,h]^d$ is a compact subset of $(0,+\infty)^d$, assumptions A1--A6 of \cite{MarTok(09)} hold. Following \cite{MarTok(09)}, we denote by $\mathbb F$ the class of probability measures on $[h_0,h]^d$ that are absolutely continuous with respect to Lebesgue measure, and by $\overline{\mathbb F}$ its closure in the weak topology. Then $G^*\in\mathbb F$ and
$$
G^\ast=\text{argmin}_{G\in \overline {\mathbb F}}KL(p_{G^\ast},p_G),
$$
where $KL$ denotes the Kullback--Leibler divergence. The claim for $\gamma\in (2/3,1)$ is then a direct consequence of Corollary 4.10 in \cite{MarTok(09)}. For $\gamma=1$, apply Theorem 4.8 in \cite{MarTok(09)}, with $a_n=\sum_{k=1}^n\alpha_k\asymp \log n$. Then
$\log n \;KL(p_{G^\ast},p_{G_{1,n}^{\tiny{[Q-B]}}} )$ converges to zero almost surely with respect to the probability measure under which the $\boldsymbol{Y}_i$'s are i.i.d. according to $p_{G^\ast}$. Since the squared Hellinger distance is dominated by the Kullback--Leibler divergence, the claim follows.

%%%%%%%%%%%%%%%%%%%%%%%%%%%%%%%%%%%%%

\subsection{Multivariate extension of Lemma 4 in \cite{Jan(24)}.}\label{dextension_jana}

The results in Section \ref{sec:merging_d} rely on the following
inequality, which extends Lemma~4 of \cite{Jan(24)} to the $d$-dimensional
setting. The constants are slightly sharper than those in Lemma~4 of
\cite{Jan(24)}, since in the present compactly supported setting no
truncation step is required. The proof follows essentially the
same argument as in \cite{Jan(24)}; we include the details for convenience.

%%%%%%%%%%%%%%%%%%%%%%%%%%%%%%%%

%%%%%%
\begin{lem}
\label{lem:regret-density-d}
Assume that $G^\ast([0,h^\ast]^d)=1$ for some $0<h^\ast<+\infty$,
and let $\widehat G$ be any probability measure such that
$\widehat G([0,h]^d)=1$ for some $h$ satisfying
$h^\ast\leq h<+\infty$. Then, for every integer $K\geq 1$,
\[
\operatorname{Regret}(\widehat G,G^\ast)
\leq
d\Big\{
6(h^2+(h^\ast)^2)+24(h+h^\ast)K
\Big\}
d_H^2(p_{\widehat G},p_{G^\ast})
+
(h+h^\ast)^2
\sum_{\ell=1}^d
\sum_{\boldsymbol y\in\mathbb N_0^d:\,y_\ell\geq K}
p_{G^\ast}(\boldsymbol y).
\]
\end{lem}

\begin{proof}
By definition,
\begin{align*}
    \operatorname{Regret}(\widehat G,G^\ast)
&=
\sum_{\boldsymbol y\in\mathbb N_0^d}
\left\|
\hat{\boldsymbol\theta}_{\widehat G}(\boldsymbol y)
-
\hat{\boldsymbol\theta}^\ast(\boldsymbol y)
\right\|_2^2
p_{G^\ast}(\boldsymbol y)\\
&=
\sum_{\ell=1}^d
\sum_{\boldsymbol y\in\mathbb N_0^d}
\left\{
\hat\theta_{\widehat G,\ell}(\boldsymbol y)
-
\hat\theta_\ell^\ast(\boldsymbol y)
\right\}^2
p_{G^\ast}(\boldsymbol y)\\
&=A_K+B_K,
\end{align*}
where
\[
A_K
=
\sum_{\ell=1}^d
\sum_{\boldsymbol y\in\mathbb N_0^d:\,y_\ell<K}
\left\{
\hat\theta_{\widehat G,\ell}(\boldsymbol y)
-
\hat\theta_\ell^\ast(\boldsymbol y)
\right\}^2
p_{G^\ast}(\boldsymbol y)
,\quad
B_K
=
\sum_{\ell=1}^d
\sum_{\boldsymbol y\in\mathbb N_0^d:\,y_\ell\geq K}
\left\{
\hat\theta_{\widehat G,\ell}(\boldsymbol y)
-
\hat\theta_\ell^\ast(\boldsymbol y)
\right\}^2
p_{G^\ast}(\boldsymbol y).
\]
Since $\widehat G([0,h]^d)=1$ and $G^\ast([0,h^\ast]^d)=1$, then
\[
0\leq \hat\theta_{\widehat G,\ell}(\boldsymbol y)\leq h,
\qquad
0\leq \hat\theta_\ell^\ast(\boldsymbol y)\leq h^\ast,
\qquad
\ell=1,\ldots,d,\quad \boldsymbol y\in\mathbb N_0^d.
\]
which entails 
\[
B_K
\leq
(h+h^\ast)^2
\sum_{\ell=1}^d
\sum_{\boldsymbol y\in\mathbb N_0^d:\,y_\ell\geq K}
p_{G^\ast}(\boldsymbol y).
\]
It remains to bound $A_K$. Fix $\ell\in\{1,\ldots,d\}$ and
$\boldsymbol y\in\mathbb N_0^d$ with $y_\ell<K$. By the multidimensional
Robbins formula,
\[
\hat\theta_{\widehat G,\ell}(\boldsymbol y)
-
\hat\theta_\ell^\ast(\boldsymbol y)
=
(y_\ell+1)
\left\{
\frac{p_{\widehat G}(\boldsymbol y+\boldsymbol e_\ell)}
     {p_{\widehat G}(\boldsymbol y)}
-
\frac{p_{G^\ast}(\boldsymbol y+\boldsymbol e_\ell)}
     {p_{G^\ast}(\boldsymbol y)}
\right\}.
\]
By $(a+b+c)^2\leq 3(a^2+b^2+c^2)$, we can write

\begin{align*}
\left\{
\hat\theta_{\widehat G,\ell}(\boldsymbol y)
-
\hat\theta_\ell^\ast(\boldsymbol y)
\right\}^2
p_{G^\ast}(\boldsymbol y)
&\leq
3\{h^2+(h^\ast)^2\}
\frac{
\{p_{\widehat G}(\boldsymbol y)-p_{G^\ast}(\boldsymbol y)\}^2
}{
p_{\widehat G}(\boldsymbol y)+p_{G^\ast}(\boldsymbol y)
}
\\
&\quad
+
12(y_\ell+1)^2
\frac{
\{p_{\widehat G}(\boldsymbol y+\boldsymbol e_\ell)
-
p_{G^\ast}(\boldsymbol y+\boldsymbol e_\ell)\}^2
}{
p_{\widehat G}(\boldsymbol y)+p_{G^\ast}(\boldsymbol y)
}.
\end{align*}
Summing over $\boldsymbol y\in\mathbb N_0^d$ such that $y_\ell<K$ gives
\[
\begin{aligned}
\sum_{\boldsymbol y\in\mathbb N_0^d:\,y_\ell<K}
\left\{
\hat\theta_{\widehat G,\ell}(\boldsymbol y)
-
\hat\theta_\ell^\ast(\boldsymbol y)
\right\}^2
p_{G^\ast}(\boldsymbol y)
&\leq
3\{h^2+(h^\ast)^2\}
\sum_{\boldsymbol y\in\mathbb N_0^d:\,y_\ell<K}
\frac{
\{p_{\widehat G}(\boldsymbol y)-p_{G^\ast}(\boldsymbol y)\}^2
}{
p_{\widehat G}(\boldsymbol y)+p_{G^\ast}(\boldsymbol y)
}
\\
&\quad
+
12
\sum_{\boldsymbol y\in\mathbb N_0^d:\,y_\ell<K}
(y_\ell+1)^2
\frac{
\{p_{\widehat G}(\boldsymbol y+\boldsymbol e_\ell)
-
p_{G^\ast}(\boldsymbol y+\boldsymbol e_\ell)\}^2
}{
p_{\widehat G}(\boldsymbol y)+p_{G^\ast}(\boldsymbol y)
}.
\end{aligned}
\]
We now control the two sums. Since, for $a,b\geq0$,
\[
\frac{(a-b)^2}{a+b}
=
\frac{(\sqrt a-\sqrt b)^2(\sqrt a+\sqrt b)^2}{a+b}
\leq
2(\sqrt a-\sqrt b)^2,
\]
we have
\begin{align*}
\sum_{\boldsymbol y\in\mathbb N_0^d:\,y_\ell<K}
\frac{
\{p_{\widehat G}(\boldsymbol y)-p_{G^\ast}(\boldsymbol y)\}^2
}{
p_{\widehat G}(\boldsymbol y)+p_{G^\ast}(\boldsymbol y)
}
\leq
2
\sum_{\boldsymbol y\in\mathbb N_0^d:\,y_\ell<K}
\left\{
\sqrt{p_{\widehat G}(\boldsymbol y)}
-
\sqrt{p_{G^\ast}(\boldsymbol y)}
\right\}^2
\leq
2d_H^2(p_{\widehat G},p_{G^\ast}).
\end{align*}
For the second sum, observe that
\[
(y_\ell+1)
\frac{
p_{\widehat G}(\boldsymbol y+\boldsymbol e_\ell)
+
p_{G^\ast}(\boldsymbol y+\boldsymbol e_\ell)
}{
p_{\widehat G}(\boldsymbol y)+p_{G^\ast}(\boldsymbol y)
}
\leq \hat\theta_{\widehat G,\ell}(\boldsymbol y)+ \hat\theta_\ell^\ast(\boldsymbol y)\leq 
h+h^\ast.
\]
Therefore, using also $y_\ell+1\leq K$ on the set $\{y_\ell<K\}$,
\begin{align*}
&\sum_{\boldsymbol y\in\mathbb N_0^d:\,y_\ell<K}
(y_\ell+1)^2
\frac{
\{p_{\widehat G}(\boldsymbol y+\boldsymbol e_\ell)
-
p_{G^\ast}(\boldsymbol y+\boldsymbol e_\ell)\}^2
}{
p_{\widehat G}(\boldsymbol y)+p_{G^\ast}(\boldsymbol y)
}
\\
&\leq
2(h+h^\ast)K
\sum_{\boldsymbol y\in\mathbb N_0^d:\,y_\ell<K}
\left\{
\sqrt{p_{\widehat G}(\boldsymbol y+\boldsymbol e_\ell)}
-
\sqrt{p_{G^\ast}(\boldsymbol y+\boldsymbol e_\ell)}
\right\}^2\\
&\leq 2(h+h^\ast)K
d_H^2(p_{\widehat G},p_{G^\ast}).
\end{align*}
Summing over $\ell=1,\ldots,d$ gives
\[
A_K
\leq
d\Big\{
6(h^2+(h^\ast)^2)
+
24(h+h^\ast)K
\Big\}
d_H^2(p_{\widehat G},p_{G^\ast}).
\]
Finally, combining the bounds for $A_K$ and $B_K$ yields
\[
\operatorname{Regret}(\widehat G,G^\ast)
\leq
d\Big\{
6(h^2+(h^\ast)^2)+24(h+h^\ast)K
\Big\}
d_H^2(p_{\widehat G},p_{G^\ast})
+
(h+h^\ast)^2
\sum_{\ell=1}^d
\sum_{\boldsymbol y\in\mathbb N_0^d:\,y_\ell\geq K}
p_{G^\ast}(\boldsymbol y).
\]
\end{proof}

%%%%%%

%%%%%%%%%%%%%%%%%%%%%%%%%%%%%%%%
%%%%%%%%%%%%%%%%%%%%%%%%%%%%%%%%
%%%%%%%%%%%%%%%%%%%%%%%%%%%%%%%%

\section{Additional synthetic-data illustrations: $1$-dimensional setting}\label{app3}

\subsection{Preliminaries}
We generate synthetic data from a Poisson mixture model, for various choices of the prior (mixing) distribution $G$. For $n\geq1$ and let $(Y_{1},\theta_{1}),\ldots,(Y_{n},\theta_{n})$, with $Y_{i}\in\mathbb{N}_{0}$ and $\theta_{i}\in\mathbb{R}_{+}$, for $i=1,\ldots,n$, be distributed as follows:
\begin{equation}\label{eq:mixture_model_poisson}
\begin{aligned}
Y_i \mid \theta_i &\quad \simind \quad \text{Poisson}(\cdot \mid \theta_i) \qquad i=1,\ldots,n,\\[0.2cm]
\theta_i &\quad \simiid \quad G.
\end{aligned}
\end{equation}
First, we assume $G$ to be a Uniform distribution on $[a,b]$; precisely, we set $a=0$ and $b=3$. Then, we consider an examples of $G$ belonging to the class $\mathcal{G}$ of sub-exponential distribution of order $s$, which is defined as
\begin{displaymath}
\mathcal{G}=\left\{G\text{ on }\mathbb{R}_{+}:\;G([t,\infty)) \leq 2e^{-t/s}\ \text{for all } t>0\right\},\qquad s>0.
\end{displaymath}
In particular, we assume $G$ to be a half-Gaussian distribution, namely the distribution of the positive part of a Gaussian random variable with mean $0$ and variance $\sigma^{2}$, which belongs $\mathcal{G}$ for $\sigma>0$; precisely, we set $\sigma=1$. The Weibull distribution considered in Section~\ref{sec3} also belongs to the class $\mathcal G$; moreover, its tail is lighter than that of the half-Gaussian distribution. Finally, we consider an example of $G$ belonging to the moment class of distributions $\mathcal{M}$ defined, which is defined, for any real $M_{p}$, as
\begin{displaymath}
\mathcal{M}=\left\{G\text{ on }\mathbb{R}_{+}:\;\int_{\mathbb{R}_{+}}\theta^{p} G(d\theta)<M_{p}\right\},\qquad p>0;
\end{displaymath}
see \citet[Section 1]{She(24)}. In particular, we assume $G$ to be square-root of half-Cauchy distribution, namely the distribution of the square-root of the positive part of a standard Cauchy random variable. This distribution has heavier tail than the Weibull and half-Gaussian distributions.

\subsection{Uniform prior}

For sample sizes $n\in\{50,\,100,\,200,\,400,\,1,000,\,2,000,\,4,000,\,8,000\}$, we generate i.i.d. data $Y_{1:n}=(Y_{1},\ldots,Y_{n})$ from a Poisson mixture model \eqref{eq:mixture_model_poisson} with a Uniform prior $G$ on the set $[0,3]$. We compare the quasi-Bayes estimate $\hat{\theta}^{\text{\tiny{[Q-B]}}}_{\gamma,n}$ and the Bayes estimate $\hat{\theta}^{\text{\tiny{[B]}}}_{n}$ with the oracle Bayes estimate $\hat{\theta}^{\ast}$. In particular, the oracle $\hat{\theta}^{\ast}$ is obtained from \eqref{eq:oracle} with $G^{\ast}$ being the Uniform prior distribution that generates the $\theta_{i}$'s, and evaluating the marginal likelihood $p_{G^\ast}$ numerically through the trapezoidal rule.

With regards to the quasi-Bayes estimate $\hat{\theta}^{\text{\tiny{[Q-B]}}}_{\gamma,n}$,  the implementation of Newton's algorithm is the same as in the synthetic-data analysis with the Weibull prior: i) the density function of $G_{n}$ is represented through its values on a fixed uniform grid of $d\in\{5,000,\,1,000,\,500,\,100,\,50,\,10\}$ quadrature points over $\Theta=(0,U_{\Theta})$, where $U_\Theta=\max\{\max\{Y_{1:n}\},\lceil Q_{n,0.99}+4\sqrt{\max\{Q_{n,0.99},1\}}\rceil\}$, with $Q_{n,0.99}=\text{Quantile}(Y_{1:n};0.99)$; ii) $G_{0}$ is Uniform over $\Theta$; iii) the learning rate is $\alpha_{n}=(1+n)^{-0.99}$. Table \ref{unif_tab_sens} reports the $\mathrm{E\text{-}mse}(\hat{G}_{\gamma,n}^{\text{\tiny{[Q-B]}}})$ and $\mathrm{E\text{-}regret}(\hat{G}_{\gamma,n}^{\text{\tiny{[Q-B]}}})$ as the sample size $n$ and the grid resolution $d$ vary.

\begin{table}[ht]
\centering
\caption{\footnotesize{Uniform prior: $\mathrm{E\text{-}mse}(\hat{G}_{\gamma,n}^{\text{\tiny{[Q-B]}}})$ and $\mathrm{E\text{-}regret}(G_{\gamma,n}^{\text{\tiny{[Q-B]}}})$ as $n$ and $d$ vary.}}
{
\setlength{\tabcolsep}{0pt}
\begin{tabular}{@{}l@{\hspace{0.5cm}}*{6}{>{\centering\arraybackslash}p{2.1cm}}@{}}
\hline
\hline
 & $d=5{,}000$ & $d=1{,}000$ & $d=500$ & $d=100$ & $d=50$ & $d=10$ \\[0.1cm]
\hline
\multicolumn{7}{@{}l}{\underline{$n=50$}} \\[0.05cm]
$\mathrm{E\text{-}mse}(\hat{G}_{\gamma,n}^{\text{\tiny{[Q-B]}}})$   & 0.515 & 0.515 & 0.515 & 0.517 & 0.522 & 0.719 \\
$\mathrm{E\text{-}regret}(\hat{G}_{\gamma,n}^{\text{\tiny{[Q-B]}}})$ & 0.045 & 0.045 & 0.045 & 0.047 & 0.052 & 0.249 \\[0.4cm]

\multicolumn{7}{@{}l}{\underline{$n=100$}} \\[0.05cm]
$\mathrm{E\text{-}mse}(\hat{G}_{\gamma,n}^{\text{\tiny{[Q-B]}}})$   & 1.045 & 1.045 & 1.045 & 1.045 & 1.047 & 1.637 \\
$\mathrm{E\text{-}regret}(\hat{G}_{\gamma,n}^{\text{\tiny{[Q-B]}}})$ & 0.539 & 0.539 & 0.539 & 0.539 & 0.541 & 1.131 \\[0.4cm]

\multicolumn{7}{@{}l}{\underline{$n=200$}} \\[0.05cm]
$\mathrm{E\text{-}mse}(\hat{G}_{\gamma,n}^{\text{\tiny{[Q-B]}}})$    & 0.597 & 0.596 & 0.595 & 0.591 & 0.589 & 0.841 \\
$\mathrm{E\text{-}regret}(\hat{G}_{\gamma,n}^{\text{\tiny{[Q-B]}}})$ & 0.113 & 0.113 & 0.112 & 0.108 & 0.106 & 0.357 \\[0.4cm]

\multicolumn{7}{@{}l}{\underline{$n=400$}} \\[0.05cm]
$\mathrm{E\text{-}mse}(\hat{G}_{\gamma,n}^{\text{\tiny{[Q-B]}}})$    & 0.646 & 0.646 & 0.646 & 0.646 & 0.647 & 0.961 \\
$\mathrm{E\text{-}regret}(\hat{G}_{\gamma,n}^{\text{\tiny{[Q-B]}}})$ & 0.159 & 0.159 & 0.159 & 0.159 & 0.160 & 0.474 \\[0.1cm]
\hline
\hline
\end{tabular}
}
\label{unif_tab_sens}
\end{table}

With regards to the Bayes estimate $\hat{\theta}^{\text{\tiny{[B]}}}_{n}$, the implementation of Algorithm 8 is the same as in the synthetic-data analysis with the Weibull prior: we set the strength parameter $c=1$, use the same Gamma base probability measure, take $m=5$ auxiliary components, and use the same MCMC settings.

Table \ref{unif_tab_sens} shows that the empirical performance of Newton's algorithm is robust to the choice of the grid resolution $d\in\{5,000,\,1,000,\,500,\,100,\,50,\,10\}$. Thus, for the evaluation of $\hat{\theta}^{\text{\tiny{[Q-B]}}}_{\gamma,n}$ we set $d=1,000$. Figure \ref{unif_compare1}-\ref{unif_compare2} display the quasi-Bayes, Bayes and oracle Bayes estimates. Figure \ref{unif_cpu} compares the quasi-Bayes and Bayes estimates in terms of both empirical performance, measured by the E-regret, and computational cost, measured by the number of computational units and by CPU time. Finally, Figure \ref{unif_regret_comparison_1d} reports the empirical regret incurred by using the quasi-Bayes estimate $\hat{\theta}^{\text{\tiny{[Q-B]}}}_{\gamma,n}$ in place of the Bayes estimate $\hat{\theta}^{\text{\tiny{[B]}}}_{n}$, providing an empirical validation of the merging as the sample size $n$ grows.

\begin{figure}
\centering
\includegraphics[width=.75\textwidth,height=.55\textheight,keepaspectratio]{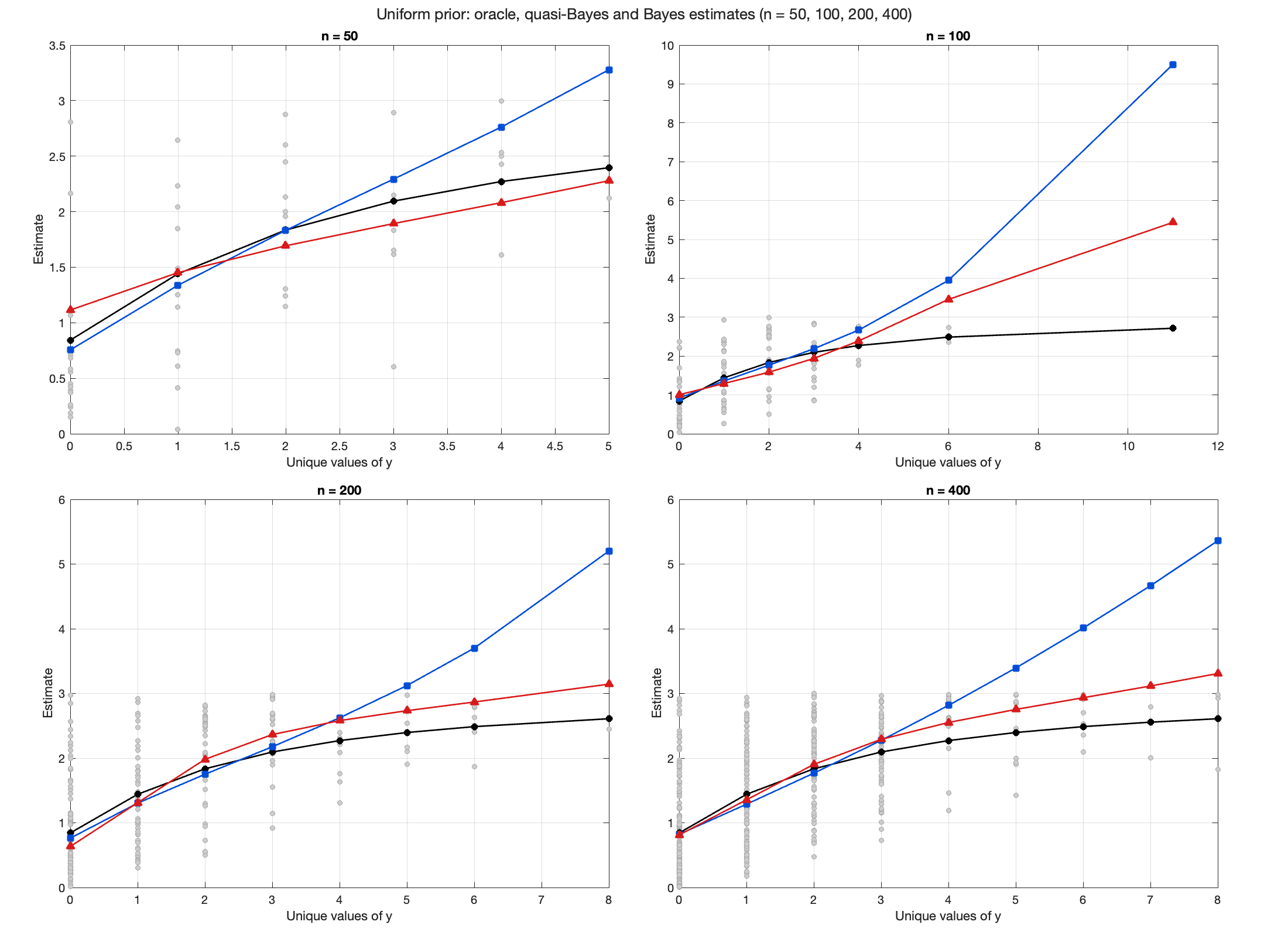}
\caption{\footnotesize{Uniform prior, $n\in\{50,\,100,\,200,\,400\}$: data points plotted against the ``true'' parameters (grey), together with the corresponding oracle Bayes (black), Bayes (red), and quasi-Bayes (blue) estimates.}}
\label{unif_compare1}
\end{figure}

\begin{figure}
\centering
\includegraphics[width=.75\textwidth,height=.55\textheight,keepaspectratio]{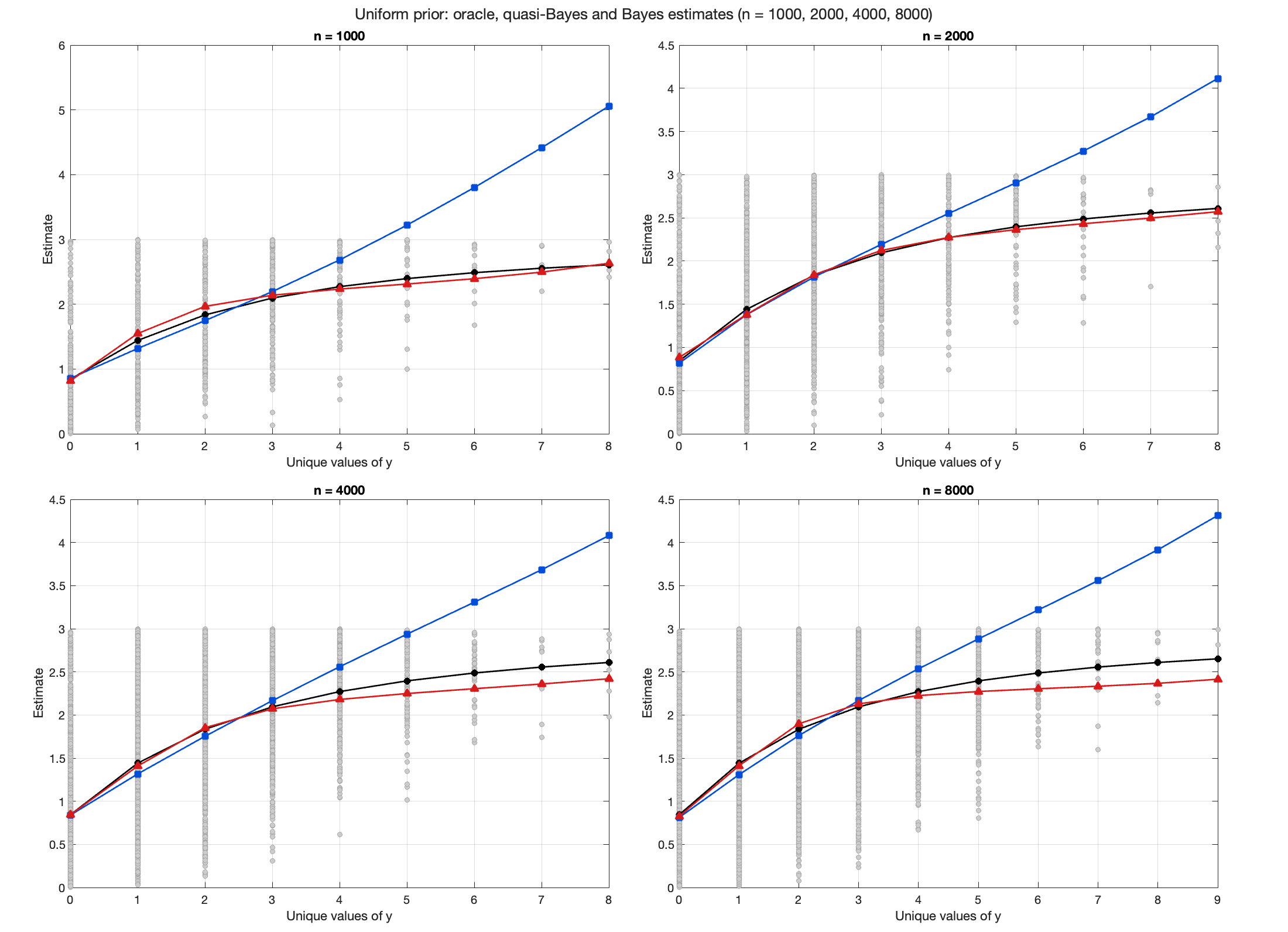}
\caption{\footnotesize{Uniform prior, $n\in\{1,000,\,2,000,\,4,000,\,8,000\}$: data points plotted against the ``true'' parameters (grey), together with the corresponding oracle Bayes (black), Bayes (red), and quasi-Bayes (blue) estimates.}}
\label{unif_compare2}
\end{figure}

\begin{figure}
\centering
\includegraphics[width=.95\textwidth,height=.85\textheight,keepaspectratio]{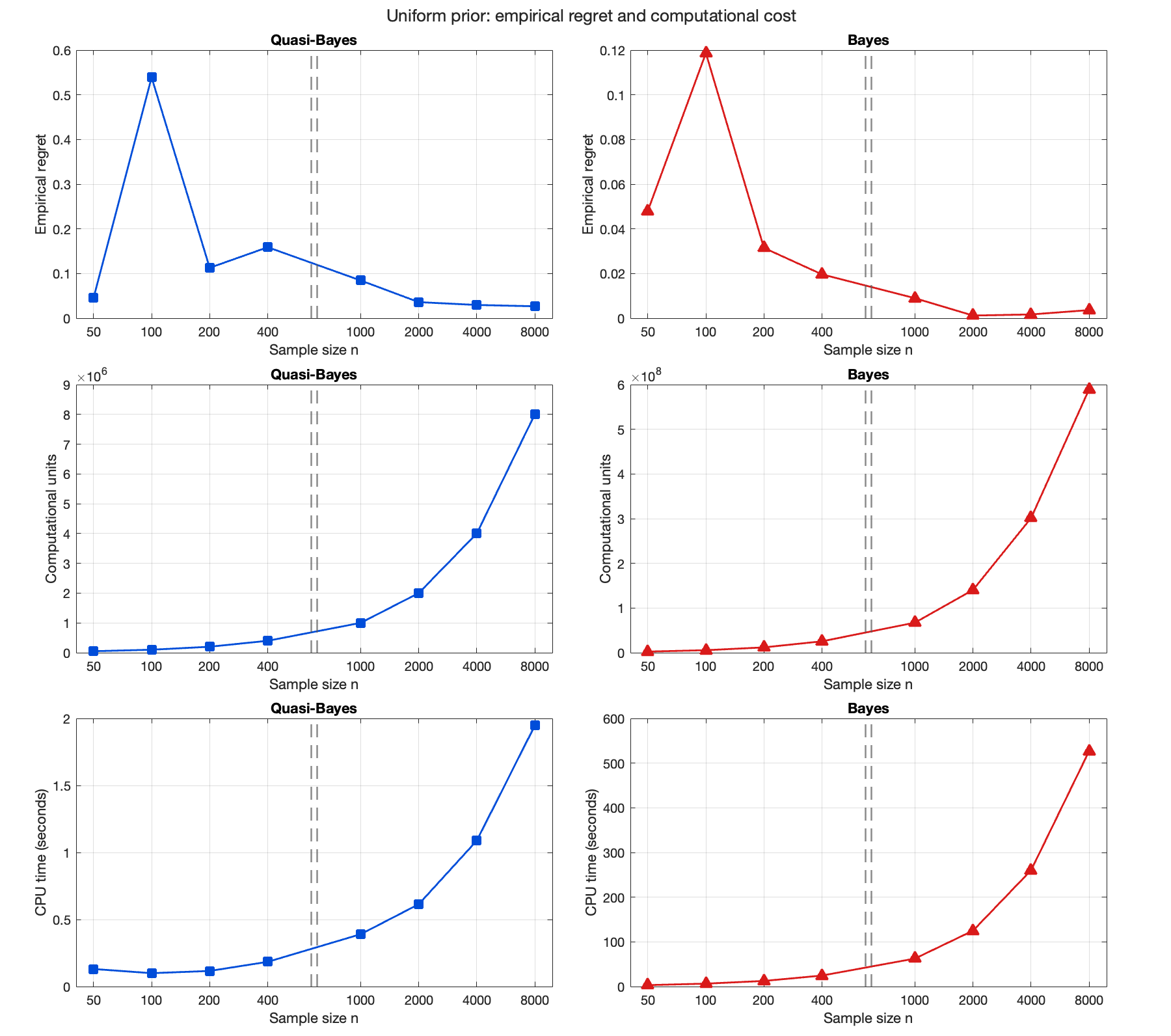}
\caption{\footnotesize{Uniform prior: quasi-Bayes (blue) and Bayes (red) estimates compared by E-regret (top panels), computational units (middle panels), and CPU time (bottom panels).}}
\label{unif_cpu}
\end{figure}

\begin{figure}
\centering
\includegraphics[width=.95\textwidth,height=.85\textheight,keepaspectratio]{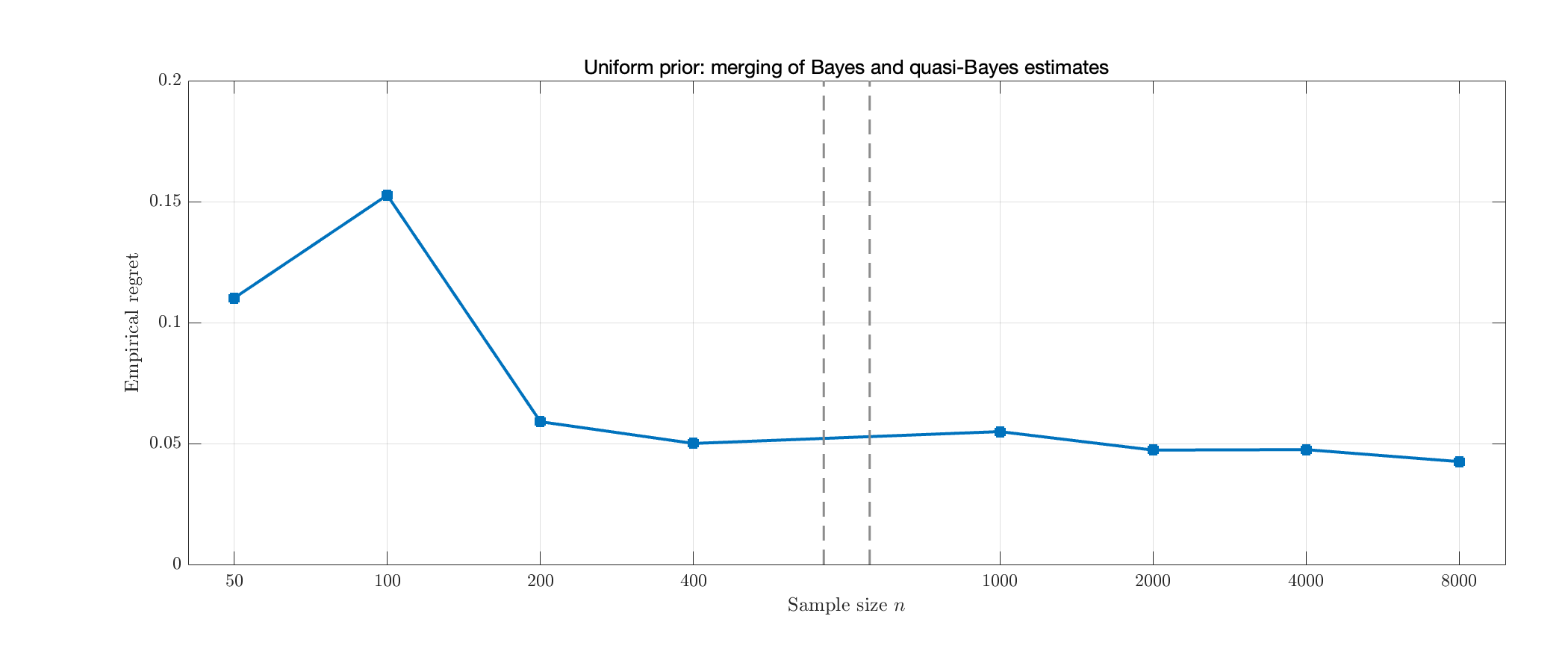}
\caption{\footnotesize{Uniform prior: E-regret incurred by using the quasi-Bayes estimate in place of the Bayes estimate.}}
\label{unif_regret_comparison_1d}
\end{figure}

\subsection{Half-Gaussian prior}

For sample sizes $n\in\{50,\,100,\,200,\,400,\,1,000,\,2,000,\,4,000,\,8,000\}$, we generate i.i.d. data $Y_{1:n}=(Y_{1},\ldots,Y_{n})$ from a Poisson mixture model \eqref{eq:mixture_model_poisson} with a half-Gaussian prior $G$ with $\sigma=1$. We compare the quasi-Bayes estimate $\hat{\theta}^{\text{\tiny{[Q-B]}}}_{\gamma,n}$ and the Bayes estimate $\hat{\theta}^{\text{\tiny{[B]}}}_{n}$ with the oracle Bayes estimate $\hat{\theta}^{\ast}$. In particular, the oracle $\hat{\theta}^{\ast}$ is obtained from \eqref{eq:oracle} with $G^{\ast}$ being the half-Gaussian prior distribution that generates the $\theta_{i}$'s, and evaluating the marginal likelihood $p_{G^\ast}$ numerically through the trapezoidal rule.

With regards to the quasi-Bayes estimate $\hat{\theta}^{\text{\tiny{[Q-B]}}}_{\gamma,n}$,  the implementation of Newton's algorithm is the same as in the synthetic-data analysis with the Weibull prior: i) the density function of $G_{n}$ is represented through its values on a fixed uniform grid of $d\in\{5,000,\,1,000,\,500,\,100,\,50,\,10\}$ quadrature points over $\Theta=(0,U_{\Theta})$, where $U_\Theta=\max\{\max\{Y_{1:n}\},\lceil Q_{n,0.99}+4\sqrt{\max\{Q_{n,0.99},1\}}\rceil\}$, with $Q_{n,0.99}=\text{Quantile}(Y_{1:n};0.99)$; ii) $G_{0}$ is Uniform over $\Theta$; iii) the learning rate is $\alpha_{n}=(1+n)^{-0.99}$. Table \ref{gauss_tab_sens} reports the $\mathrm{E\text{-}mse}(\hat{G}_{\gamma,n}^{\text{\tiny{[Q-B]}}})$ and $\mathrm{E\text{-}regret}(\hat{G}_{\gamma,n}^{\text{\tiny{[Q-B]}}})$ as the sample size $n$ and the grid resolution $d$ vary.

\begin{table}[ht]
\centering
\caption{\footnotesize{Half-Gaussian prior: $\mathrm{E\text{-}mse}(\hat{G}_{\gamma,n}^{\text{\tiny{[Q-B]}}})$ and $\mathrm{E\text{-}regret}(\hat{G}_{\gamma,n}^{\text{\tiny{[Q-B]}}})$ as $n$ and $d$ vary.}}
{
\setlength{\tabcolsep}{0pt}
\begin{tabular}{@{}l@{\hspace{0.5cm}}*{6}{>{\centering\arraybackslash}p{2.1cm}}@{}}
\hline
\hline
 & $d=5{,}000$ & $d=1{,}000$ & $d=500$ & $d=100$ & $d=50$ & $d=10$ \\[0.1cm]
\hline
\multicolumn{7}{@{}l}{\underline{$n=50$}} \\[0.05cm]
$\mathrm{E\text{-}mse}(\hat{G}_{\gamma,n}^{\text{\tiny{[Q-B]}}})$   & 0.280 & 0.281 & 0.281 & 0.286 & 0.294 & 0.651 \\
 $\mathrm{E\text{-}regret}(\hat{G}_{\gamma,n}^{\text{\tiny{[Q-B]}}})$ & 0.108 & 0.109 & 0.109 & 0.114 & 0.122 & 0.478 \\[0.4cm]

\multicolumn{7}{@{}l}{\underline{$n=100$}} \\[0.05cm]
$\mathrm{E\text{-}mse}(\hat{G}_{\gamma,n}^{\text{\tiny{[Q-B]}}})$    & 0.361 & 0.362 & 0.363 & 0.376 & 0.398 & 1.030 \\
 $\mathrm{E\text{-}regret}(\hat{G}_{\gamma,n}^{\text{\tiny{[Q-B]}}})$ & 0.117 & 0.118 & 0.120 & 0.132 & 0.154 & 0.786 \\[0.4cm]

\multicolumn{7}{@{}l}{\underline{$n=200$}} \\[0.05cm]
$\mathrm{E\text{-}mse}(\hat{G}_{\gamma,n}^{\text{\tiny{[Q-B]}}})$    & 0.296 & 0.295 & 0.295 & 0.293 & 0.291 & 0.492 \\
 $\mathrm{E\text{-}regret}(\hat{G}_{\gamma,n}^{\text{\tiny{[Q-B]}}})$& 0.004 & 0.004 & 0.003 & 0.001 & -0.001 & 0.200 \\[0.4cm]

\multicolumn{7}{@{}l}{\underline{$n=400$}} \\[0.05cm]
$\mathrm{E\text{-}mse}(\hat{G}_{\gamma,n}^{\text{\tiny{[Q-B]}}})$   & 0.280 & 0.281 & 0.281 & 0.285 & 0.291 & 0.597 \\
 $\mathrm{E\text{-}regret}(\hat{G}_{\gamma,n}^{\text{\tiny{[Q-B]}}})$ & 0.010 & 0.010 & 0.011 & 0.014 & 0.020 & 0.327 \\[0.1cm]
\hline
\hline
\end{tabular}
}
\label{gauss_tab_sens}
\end{table}

With regards to the Bayes estimate $\hat{\theta}^{\text{\tiny{[B]}}}_{n}$, the implementation of Algorithm 8 is the same as in the synthetic-data analysis with the Weibull prior: we set the strength parameter $c=1$, use the same Gamma base probability measure, take $m=5$ auxiliary components, and use the same MCMC settings.

Table \ref{gauss_tab_sens} shows that the empirical performance of Newton's algorithm is robust to the choice of the grid resolution $d\in\{5,000,\,1,000,\,500,\,100,\,50,\,10\}$. Thus, for the evaluation of $\hat{\theta}^{\text{\tiny{[Q-B]}}}_{n}$ we set $d=1,000$. Figure \ref{gauss_compare1}-\ref{gauss_compare2} display the quasi-Bayes, Bayes and oracle Bayes estimates. Figure \ref{gauss_cpu} compares the quasi-Bayes and Bayes estimates in terms of both empirical performance, measured by the E-regret, and  computational cost, measured by the number of computational units and by CPU time. Finally, Figure \ref{gauss_regret_comparison_1d} reports the empirical regret incurred by using the quasi-Bayes estimate $\hat{\theta}^{\text{\tiny{[Q-B]}}}_{\gamma,n}$ in place of the Bayes estimate $\hat{\theta}^{\text{\tiny{[B]}}}_{n}$, providing an empirical validation of the merging as the sample size $n$ grows.

\begin{figure}
\centering
\includegraphics[width=.75\textwidth,height=.55\textheight,keepaspectratio]{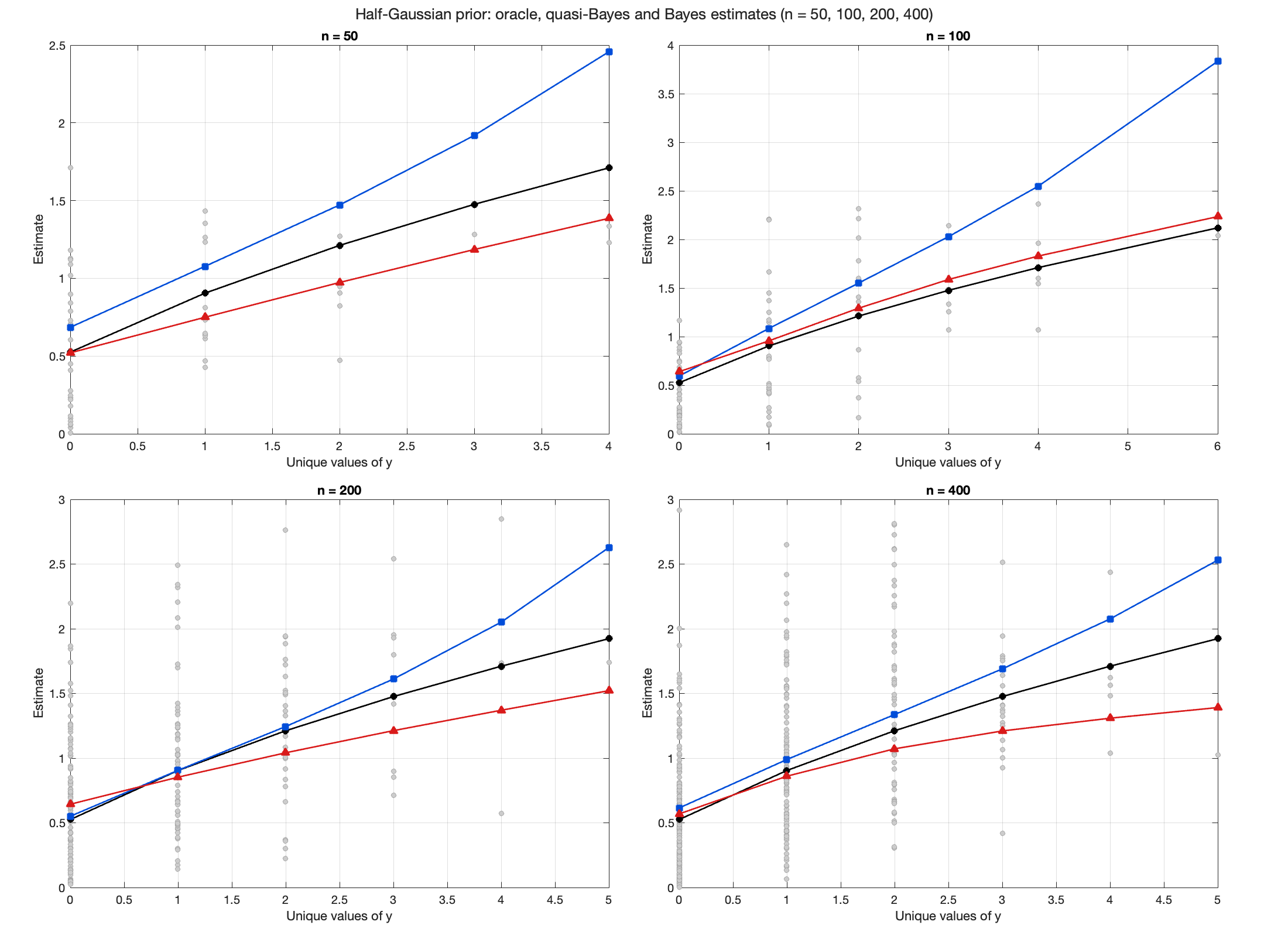}
\caption{\footnotesize{Half-Gaussian prior, $n\in\{50,\,100,\,200,\,400\}$: data points plotted against the ``true'' parameters (grey), together with the corresponding oracle Bayes (black), Bayes (red), and quasi-Bayes (blue) estimates.}}
\label{gauss_compare1}
\end{figure}

\begin{figure}
\centering
\includegraphics[width=.75\textwidth,height=.55\textheight,keepaspectratio]{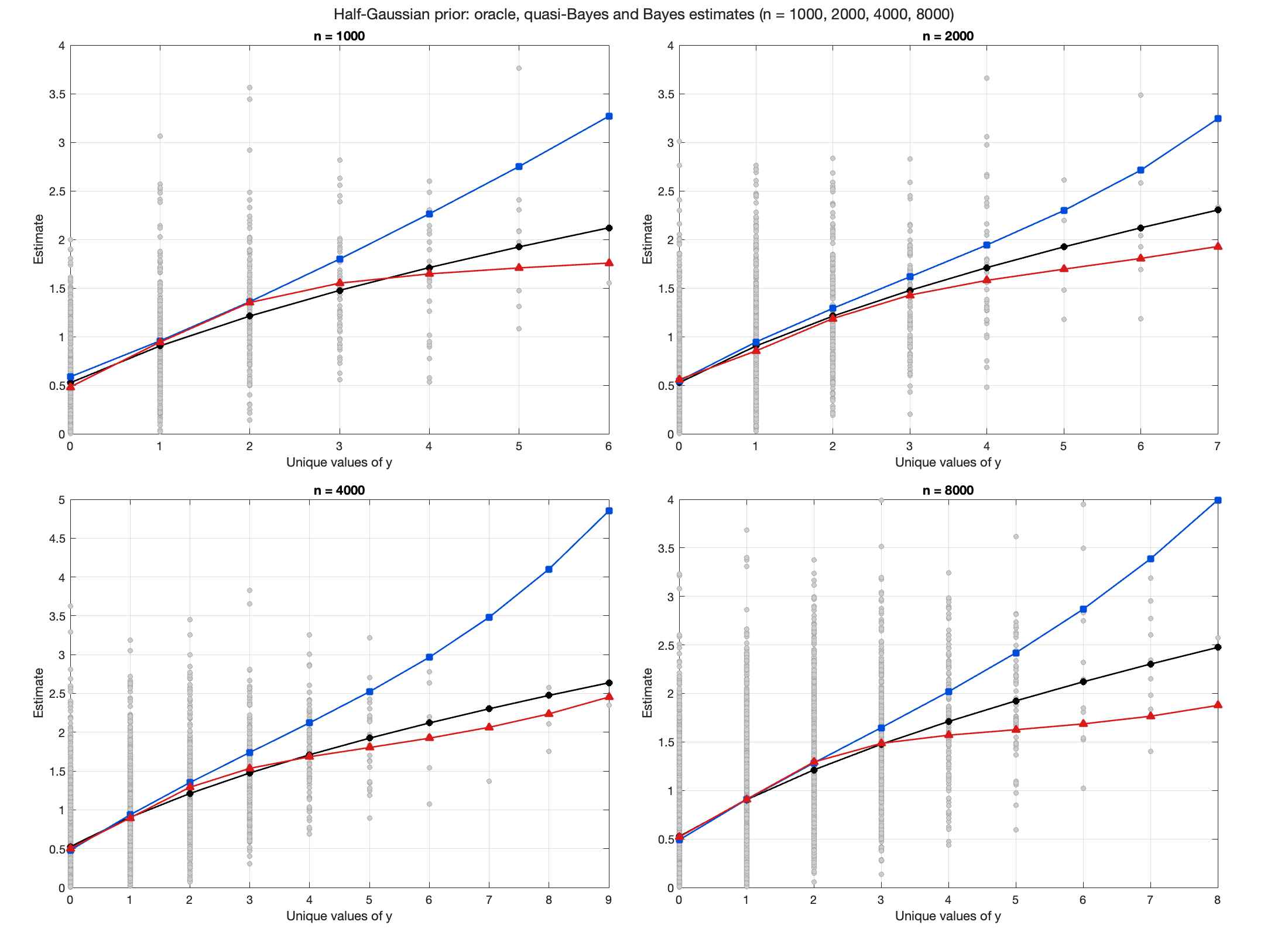}
\caption{\footnotesize{Half-Gaussian prior, $n\in\{1,000,\,2,000,\,4,000,\,8,000\}$: data points plotted against the ``true'' parameters (grey), together with the corresponding oracle Bayes (black), Bayes (red), and quasi-Bayes (blue) estimates.}}
\label{gauss_compare2}
\end{figure}

\begin{figure}
\centering
\includegraphics[width=.95\textwidth,height=.85\textheight,keepaspectratio]{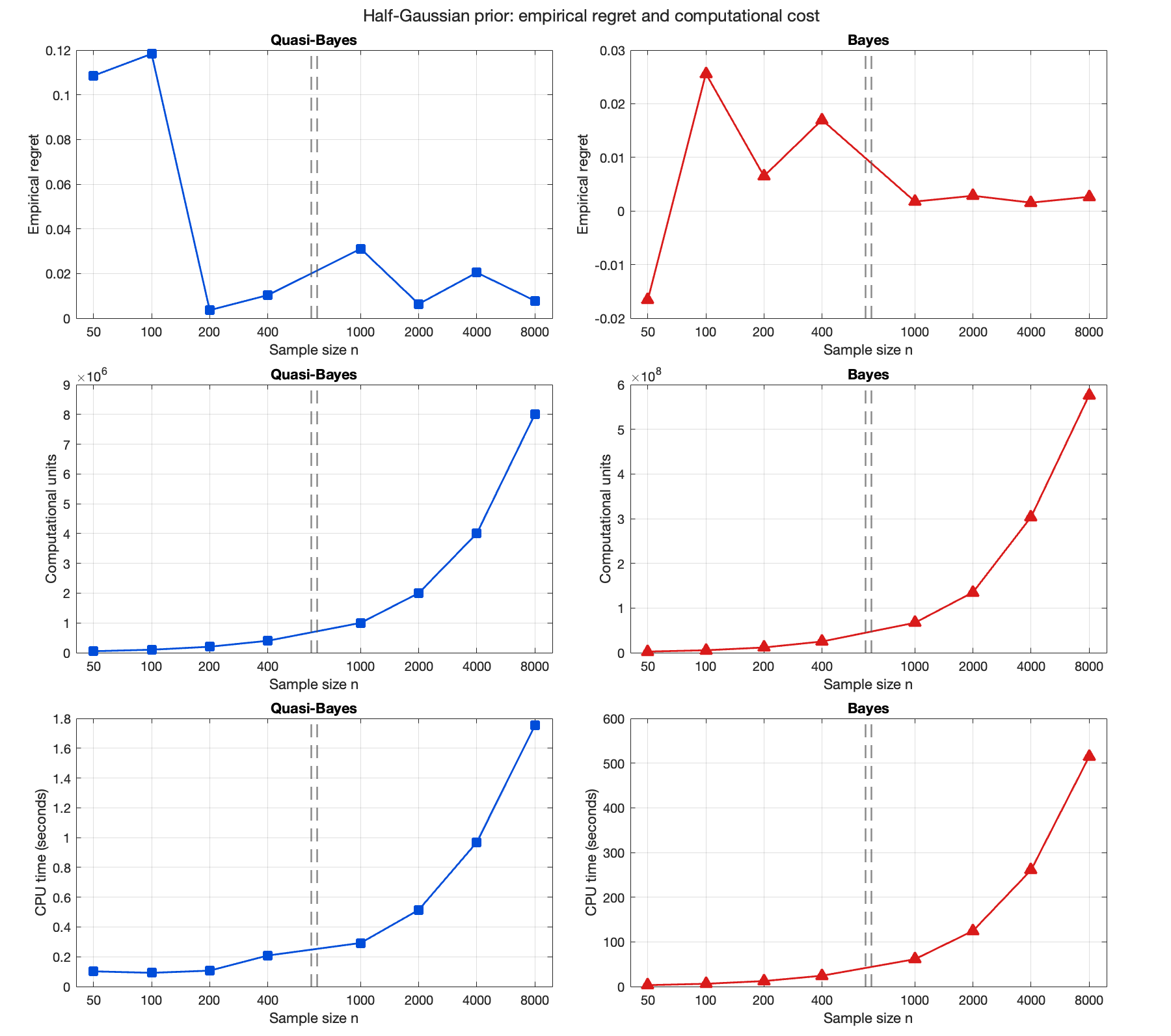}
\caption{\footnotesize{Half-Gaussian prior: quasi-Bayes (blue) and Bayes (red) estimates compared by E-regret (top panels), computational units (middle panels), and CPU time (bottom panels).}}
\label{gauss_cpu}
\end{figure}

\begin{figure}
\centering
\includegraphics[width=.95\textwidth,height=.85\textheight,keepaspectratio]{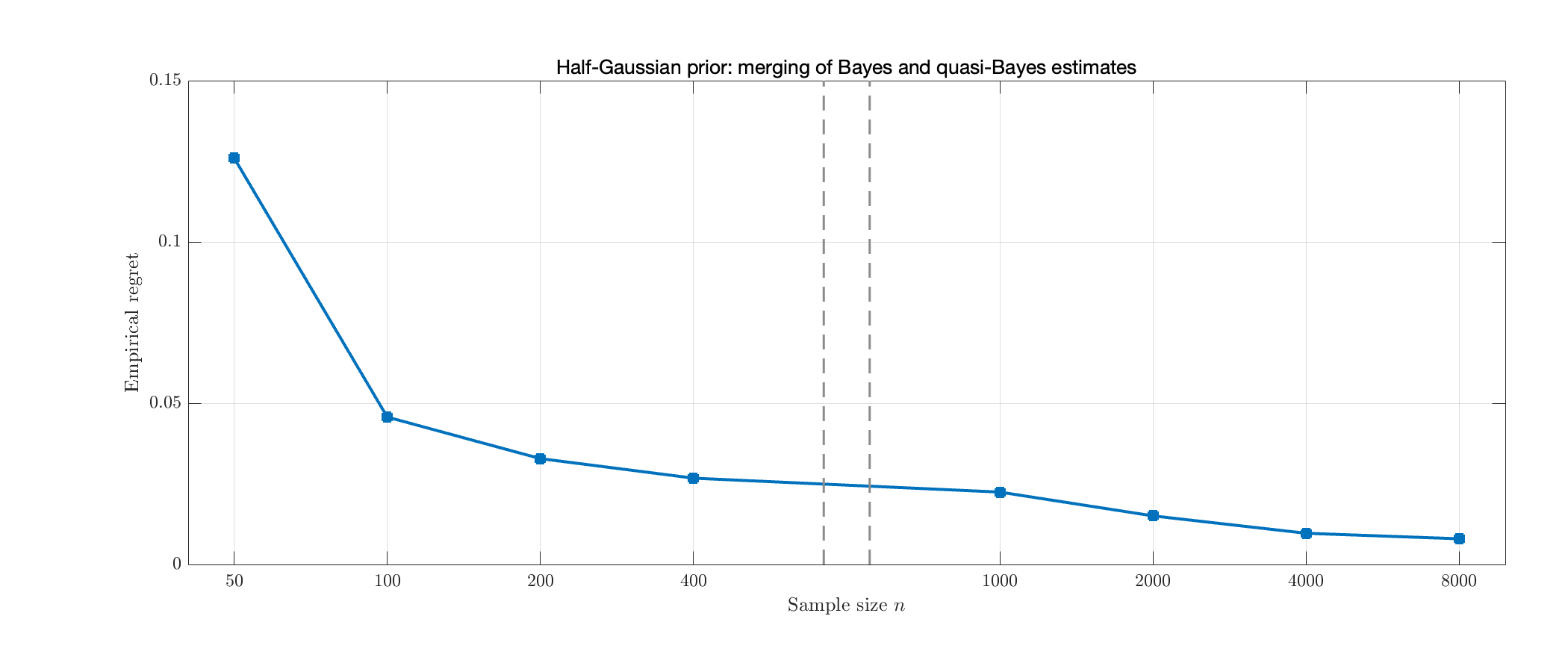}
\caption{\footnotesize{Half-Gaussian prior: E-regret incurred by using the quasi-Bayes estimate in place of the Bayes estimate.}}
\label{gauss_regret_comparison_1d}
\end{figure}

\subsection{Square-root of half-Cauchy prior}

For sample sizes $n\in\{50,\,100,\,200,\,400,\,1,000,\,2,000,\,4,000,\,8,000\}$, we generate i.i.d. data $Y_{1:n}=(Y_{1},\ldots,Y_{n})$ from a Poisson mixture model \eqref{eq:mixture_model_poisson} with a square-root of half-Cauchy prior $G$. We compare the quasi-Bayes estimate $\hat{\theta}^{\text{\tiny{[Q-B]}}}_{n}$ and the Bayes estimate $\hat{\theta}^{\text{\tiny{[B]}}}_{n}$ with the oracle Bayes estimate $\hat{\theta}^{\ast}$. In particular, the oracle $\hat{\theta}^{\ast}$ is obtained from \eqref{eq:oracle} with $G^{\ast}$ being the square-root of half-Cauchy prior distribution that generates the $\theta_{i}$'s, and evaluating the marginal likelihood $p_{G^\ast}$ numerically through the trapezoidal rule.

With regards to the quasi-Bayes estimate $\hat{\theta}^{\text{\tiny{[Q-B]}}}_{n}$,  the implementation of Newton's algorithm is the same as in the synthetic-data analysis with the Weibull prior: i) the density function of $G_{n}$ is represented through its values on a fixed uniform grid of $d\in\{5,000,\,1,000,\,500,\,100,\,50,\,10\}$ quadrature points over $\Theta=(0,U_{\Theta})$, where $U_\Theta=\max\{\max\{Y_{1:n}\},\lceil Q_{n,0.99}+4\sqrt{\max\{Q_{n,0.99},1\}}\rceil\}$, with $Q_{n,0.99}=\text{Quantile}(Y_{1:n};0.99)$; ii) $G_{0}$ is Uniform over $\Theta$; iii) the learning rate is $\alpha_{n}=(1+n)^{-0.99}$. Table \ref{cau_tab_sens} reports the $\mathrm{E\text{-}mse}(\hat{G}_{\gamma,n}^{\text{\tiny{[Q-B]}}})$ and $\mathrm{E\text{-}regret}(G_{\gamma,n}^{\text{\tiny{[Q-B]}}})$ as the sample size $n$ and the grid resolution $d$ vary.

\begin{table}[ht]
\centering
\caption{\footnotesize{Square-root of half-Cauchy prior: $\mathrm{E\text{-}mse}(\hat{G}_{\gamma,n}^{\text{\tiny{[Q-B]}}})$ and $\mathrm{E\text{-}regret}(G_{\gamma,n}^{\text{\tiny{[Q-B]}}})$ as $n$ and $d$ vary.}}
{
\setlength{\tabcolsep}{0pt}
\begin{tabular}{@{}l@{\hspace{0.5cm}}*{6}{>{\centering\arraybackslash}p{2.1cm}}@{}}
\hline
\hline
 & $d=5{,}000$ & $d=1{,}000$ & $d=500$ & $d=100$ & $d=50$ & $d=10$ \\[0.1cm]
\hline
\multicolumn{7}{@{}l}{\underline{$n=50$}} \\[0.05cm]
$\mathrm{E\text{-}mse}(\hat{G}_{\gamma,n}^{\text{\tiny{[Q-B]}}})$     & 0.773 & 0.773 & 0.773 & 0.774 & 0.775 & 1.136 \\
$\mathrm{E\text{-}regret}(G_{\gamma,n}^{\text{\tiny{[Q-B]}}})$ & 0.108 & 0.108 & 0.108 & 0.108 & 0.109 & 0.470 \\[0.4cm]

\multicolumn{7}{@{}l}{\underline{$n=100$}} \\[0.05cm]
$\mathrm{E\text{-}mse}(\hat{G}_{\gamma,n}^{\text{\tiny{[Q-B]}}})$     & 0.515 & 0.514 & 0.513 & 0.509 & 0.545 & 10.613 \\
$\mathrm{E\text{-}regret}(G_{\gamma,n}^{\text{\tiny{[Q-B]}}})$ & 0.105 & 0.104 & 0.103 & 0.099 & 0.135 & 10.203 \\[0.4cm]

\multicolumn{7}{@{}l}{\underline{$n=200$}} \\[0.05cm]
$\mathrm{E\text{-}mse}(\hat{G}_{\gamma,n}^{\text{\tiny{[Q-B]}}})$     & 0.585 & 0.585 & 0.585 & 0.586 & 0.588 & 1.062 \\
$\mathrm{E\text{-}regret}(G_{\gamma,n}^{\text{\tiny{[Q-B]}}})$ & -0.014 & -0.014 & -0.014 & -0.013 & -0.011 & 0.463 \\[0.4cm]

\multicolumn{7}{@{}l}{\underline{$n=400$}} \\[0.05cm]
$\mathrm{E\text{-}mse}(\hat{G}_{\gamma,n}^{\text{\tiny{[Q-B]}}})$     & 0.594 & 0.594 & 0.594 & 0.593 & 0.592 & 1.798 \\
$\mathrm{E\text{-}regret}(G_{\gamma,n}^{\text{\tiny{[Q-B]}}})$ & 0.015 & 0.015 & 0.015 & 0.014 & 0.013 & 1.218 \\[0.1cm]
\hline
\hline
\end{tabular}
}
\label{cau_tab_sens}
\end{table}

With regards to the Bayes estimate $\hat{\theta}^{\text{\tiny{[B]}}}_{n}$, the implementation of Algorithm 8 is the same as in the synthetic-data analysis with the Weibull prior: we set the strength parameter $c=1$, use the same Gamma base probability measure, take $m=5$ auxiliary components, and use the same MCMC settings.

Table \ref{cau_tab_sens} shows that the empirical performance of Newton's algorithm is robust to the choice of the grid resolution $d\in\{5,000,\,1,000,\,500,\,100,\,50,\,10\}$. Thus, for the evaluation of $\hat{\theta}^{\text{\tiny{[Q-B]}}}_{n}$ we set $d=1,000$. Figure \ref{cau_compare1}-\ref{cau_compare2} display the quasi-Bayes, Bayes and oracle Bayes estimates. Figure \ref{cau_cpu} compares the quasi-Bayes and Bayes estimates in terms of both empirical performance, measured by the E-regret, and  computational cost, measured by the number of computational units and by CPU time. Finally, Figure \ref{cau_regret_comparison_1d} reports the empirical regret of the quasi-Bayes estimate $\hat{\theta}^{\text{\tiny{[Q-B]}}}_{\gamma,n}$ relative to the Bayes estimate $\hat{\theta}^{\text{\tiny{[B]}}}_{n}$, providing an empirical validation of the merging as the sample size $n$ grows.

\begin{figure}
\centering
\includegraphics[width=.75\textwidth,height=.55\textheight,keepaspectratio]{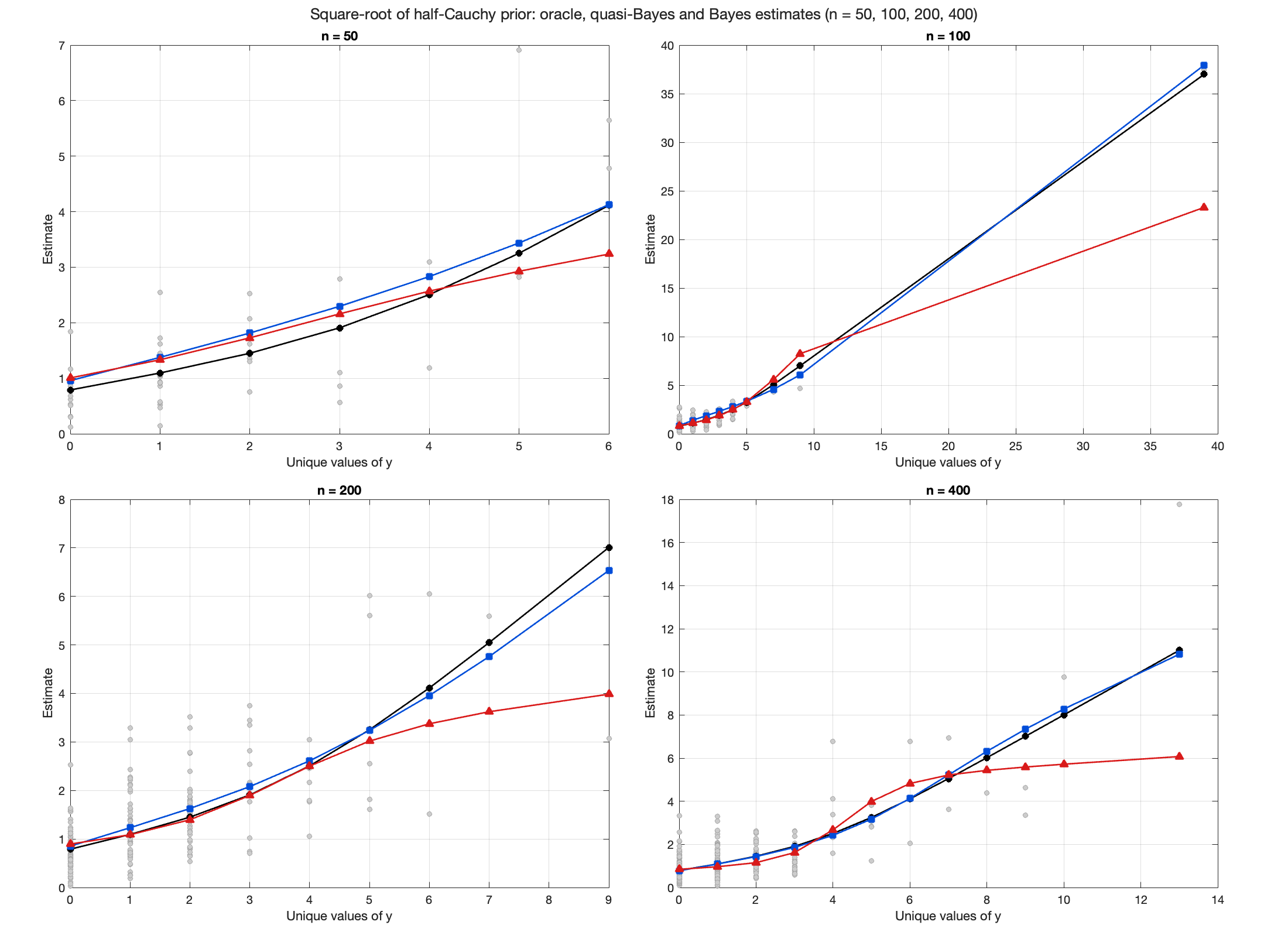}
\caption{\footnotesize{Square-root of half-Cauchy prior, $n\in\{50,\,100,\,200,\,400\}$: data points plotted against the ``true'' parameters (grey), together with the corresponding oracle Bayes (black), Bayes (red), and quasi-Bayes (blue) estimates.}}
\label{cau_compare1}
\end{figure}

\begin{figure}
\centering
\includegraphics[width=.75\textwidth,height=.55\textheight,keepaspectratio]{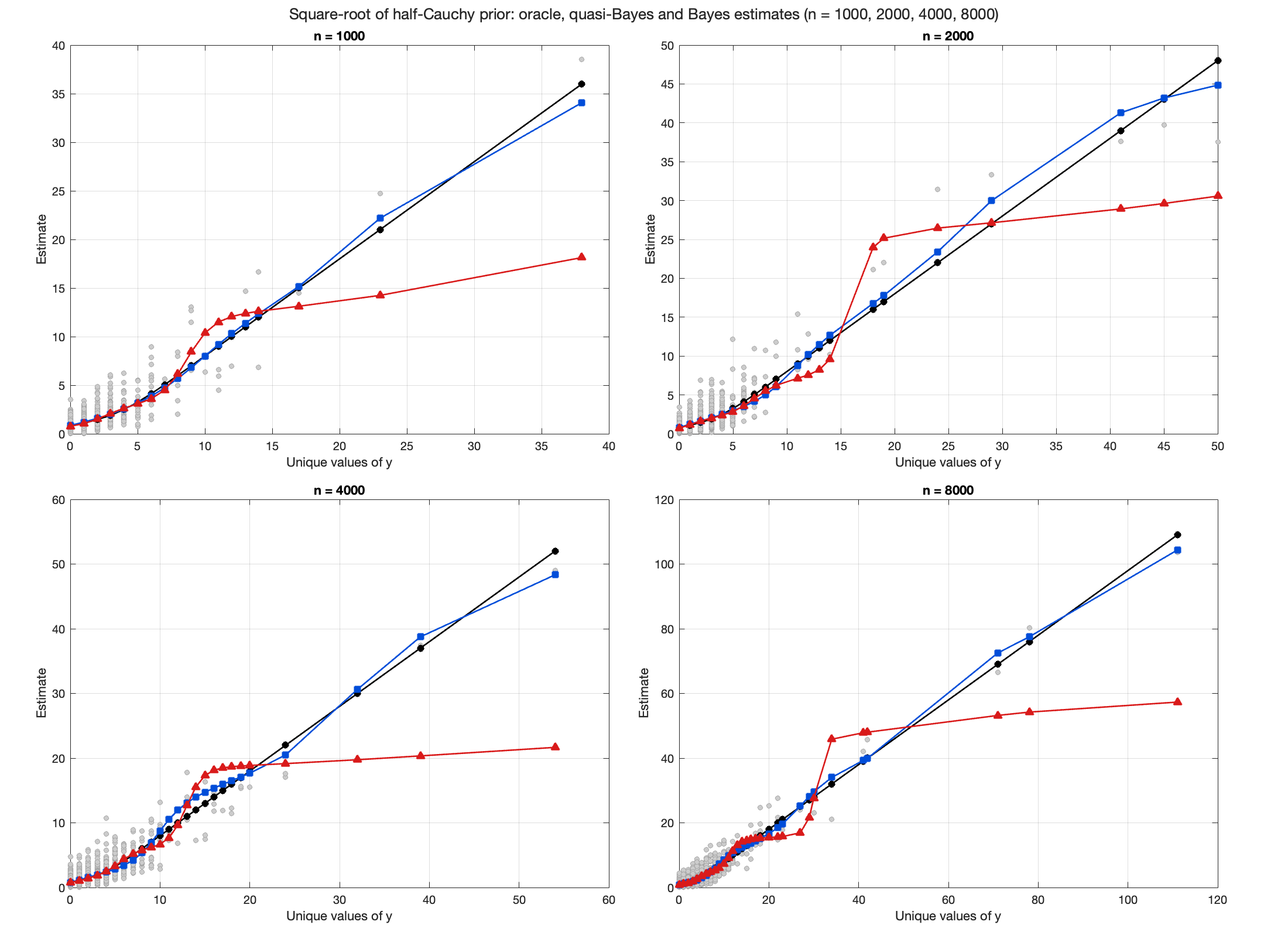}
\caption{\footnotesize{Square-root of half-Cauchy prior, $n\in\{1,000,\,2,000,\,4,000,\,8,000\}$: data points plotted against the ``true'' parameters (grey), together with the corresponding oracle Bayes (black), Bayes (red), and quasi-Bayes (blue) estimates.}}
\label{cau_compare2}
\end{figure}

\begin{figure}
\centering
\includegraphics[width=.95\textwidth,height=.85\textheight,keepaspectratio]{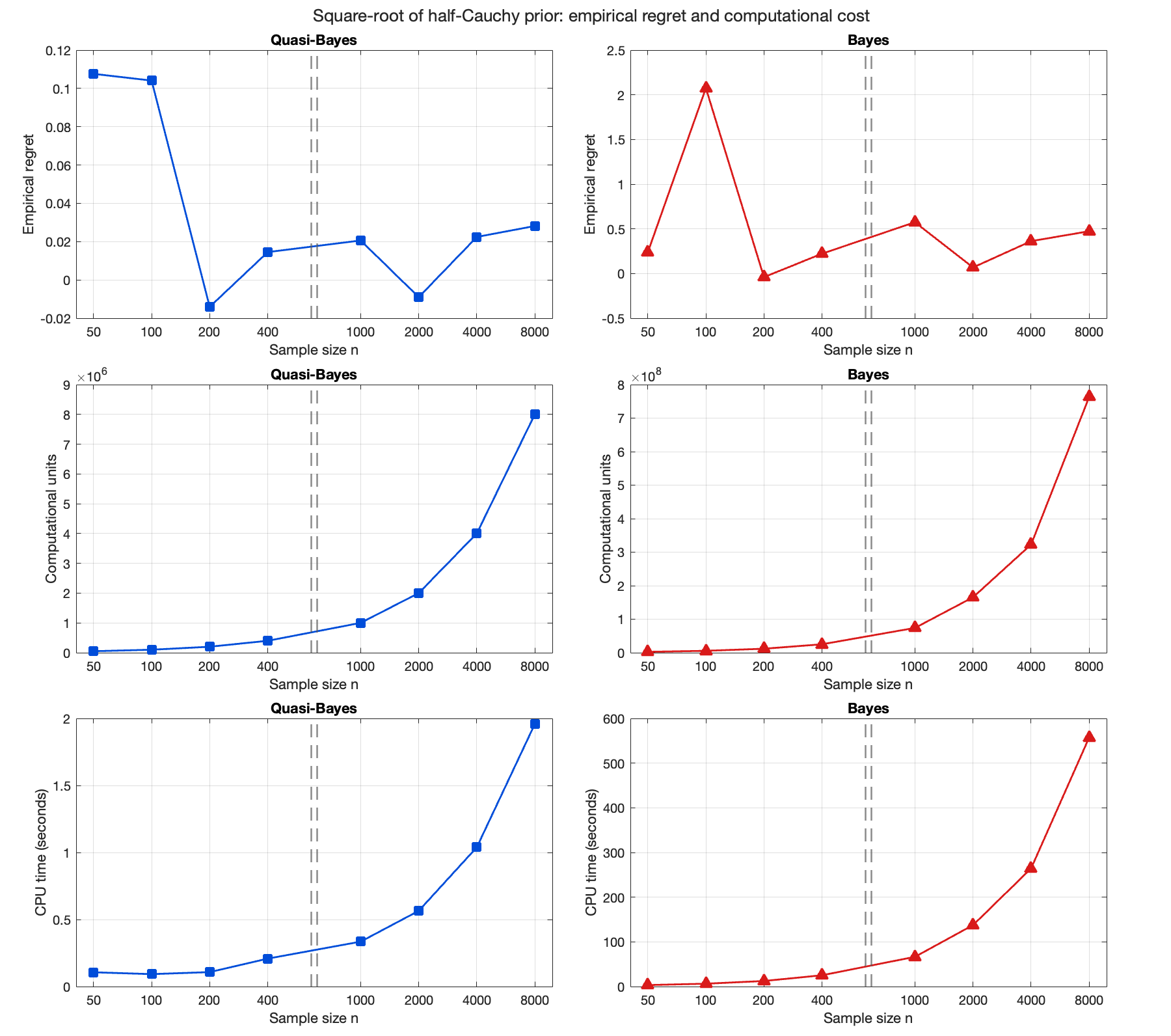}
\caption{\footnotesize{Square-root of half-Cauchy prior: quasi-Bayes (blue) and Bayes (red) estimates compared by E-regret (top panels), computational units (middle panels), and CPU time (bottom panels).}}
\label{cau_cpu}
\end{figure}

\begin{figure}
\centering
\includegraphics[width=.95\textwidth,height=.85\textheight,keepaspectratio]{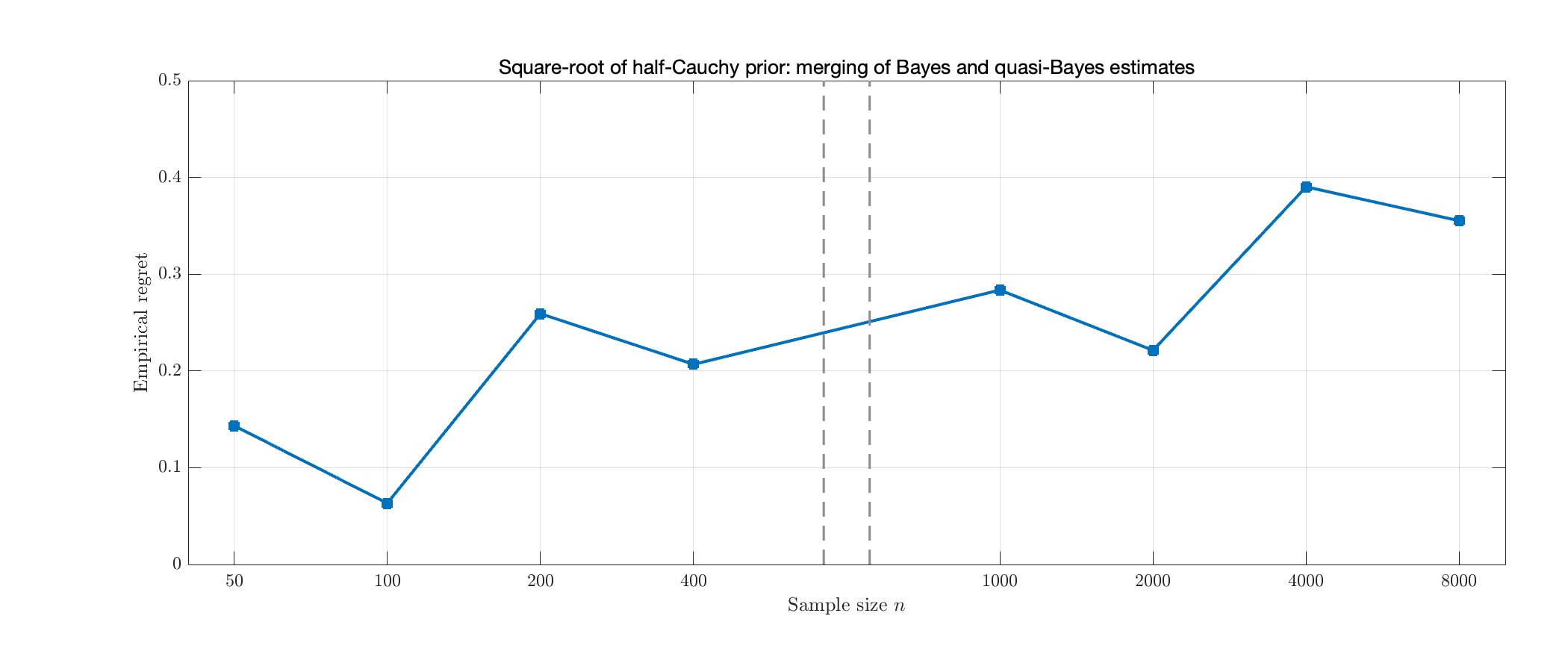}
\caption{\footnotesize{Square-root of half-Cauchy prior: E-regret incurred by using the quasi-Bayes estimate in place of the Bayes estimate.}}
\label{cau_regret_comparison_1d}
\end{figure}

\section{Additional synthetic-data illustrations: $d$-dimensional setting, $d>1$}\label{app4}

\subsection{Preliminaries}

We generate synthetic data from a $d$-dimensional Poisson mixture model, with dimension $d=2$, for various choices of the product prior (mixing) distribution $G=G_{1}\otimes G_{2}$. For $n\geq1$ let  $(\boldsymbol Y_1,\boldsymbol\theta_1),\ldots,(\boldsymbol Y_n,\boldsymbol\theta_n)$, with $\boldsymbol Y_i =(Y_{i,1},Y_{i,2}) \in\mathbb N_0^2$ and $\boldsymbol{\theta}_{i}=(\theta_{i,1},\theta_{i,2}))\in\mathbb{R}_{+}^{2}$, for $i=1,\ldots,n$, be distributed as follows:
\begin{equation}\label{eq:mixture_model_poisson_d}
\begin{aligned}
Y_{i,\ell}\mid \boldsymbol\theta_i
&\quad \simind \quad
\operatorname{Poisson}(\cdot\mid \theta_{i,\ell}),
\qquad \ell=1,2,\quad i=1,\ldots,n,\\[0.2cm]
\boldsymbol\theta_i
&\quad \simiid \quad G.
\end{aligned}
\end{equation}
where $G$ has independent and identically specified marginals $G_{\ell}$, $\ell=1,2$. We assume the $G_{\ell}$'s to be Uniform, Half-Cauchy and square-root of half-Cauchy distributions, as in the Poisson mixture model \eqref{eq:mixture_model_poisson}.

\subsection{Uniform product prior}

For sample sizes $n\in\{50,\,100,\,200,\,400,\,1,000,\,2,000,\,4,000,\,8,000\}$, we generate i.i.d. data $\boldsymbol Y_{1:n}=(\boldsymbol Y_1,\ldots,\boldsymbol Y_n)$, with $\boldsymbol Y_i=(Y_{i,1},Y_{i,2})\in\mathbb N_0^2$ from a $2$-dimensional Poisson mixture model \eqref{eq:mixture_model_poisson_d} with a product Uniform prior $G=G_1\otimes G_2$,  where $G_{\ell}$ is the Uniform distribution on $[0,3]$. We compare the quasi-Bayes estimate $\hat{\boldsymbol{\theta}}^{\text{\tiny{[Q-B]}}}_{\gamma,n}$ and the Bayes estimate $\hat{\boldsymbol{\theta}}^{\text{\tiny{[B]}}}_{n}$ with the oracle Bayes estimate $\hat{\boldsymbol{\theta}}^{\ast}$. In particular, the oracle $\hat{\boldsymbol{\theta}}^{\ast}$ is obtained from \eqref{eq:oracle_d} with $G^{\ast}=G_1^{\ast}\otimes G_2^{\ast}$ being the product Uniform prior distribution that generates the $\boldsymbol{\theta}_{i}$'s, and evaluating the marginal likelihood $p_{G^\ast}$ numerically through the trapezoidal rule.

With regards to the quasi-Bayes estimate $\hat{\boldsymbol\theta}^{\text{\tiny{[Q-B]}}}_{\gamma,n}$, the implementation of the $2$-dimensional Newton's algorithm is the same as in the synthetic-data analysis with the product Weibull prior: i) the density function of $G_n$ is represented through its values on a fixed tensor-product grid over $\boldsymbol{\Theta}=(0,U_{\Theta_{1}})\times(0,U_{\Theta_{2}})$ with $201$ quadrature points per coordinate, yielding $201^2$ grid points in total; ii) $G_0$ is Uniform over $\boldsymbol{\Theta}$; iii) the learning rate is $\alpha_{n}=(1+n)^{-0.99}$. With regards to the Bayes estimate  $\hat{\boldsymbol{\theta}}^{\text{\tiny{[B]}}}_{n}$, the implementation of the $2$-dimensional Algorithm 8 is the same as in the synthetic-data analysis with the product Weibull prior: we set the strength parameter $c=1$, use the same product Gamma base probability measure, take $m=5$ auxiliary components, and use the same MCMC settings.

Figure \ref{unif_compare1_d}-\ref{unif_compare2_d} display the quasi-Bayes, Bayes and oracle Bayes estimates. Figure \ref{unif_cpu_d} compares the quasi-Bayes and Bayes estimates in terms of empirical performance, measured by the E-regret, and computational cost, measured by the number of computational units and CPU time. Finally, Figure \ref{unif_regret_comparison_d} reportsthe empirical regret of the quasi-Bayes estimate $\hat{\boldsymbol{\theta}}^{\text{\tiny{[Q-B]}}}_{\gamma,n}$ relative to the Bayes estimate $\hat{\boldsymbol{\theta}}^{\text{\tiny{[B]}}}_{n}$, providing an empirical validation of the merging as the sample size $n$ grows.

\begin{figure}
\centering
\includegraphics[width=.75\textwidth,height=.55\textheight,keepaspectratio]{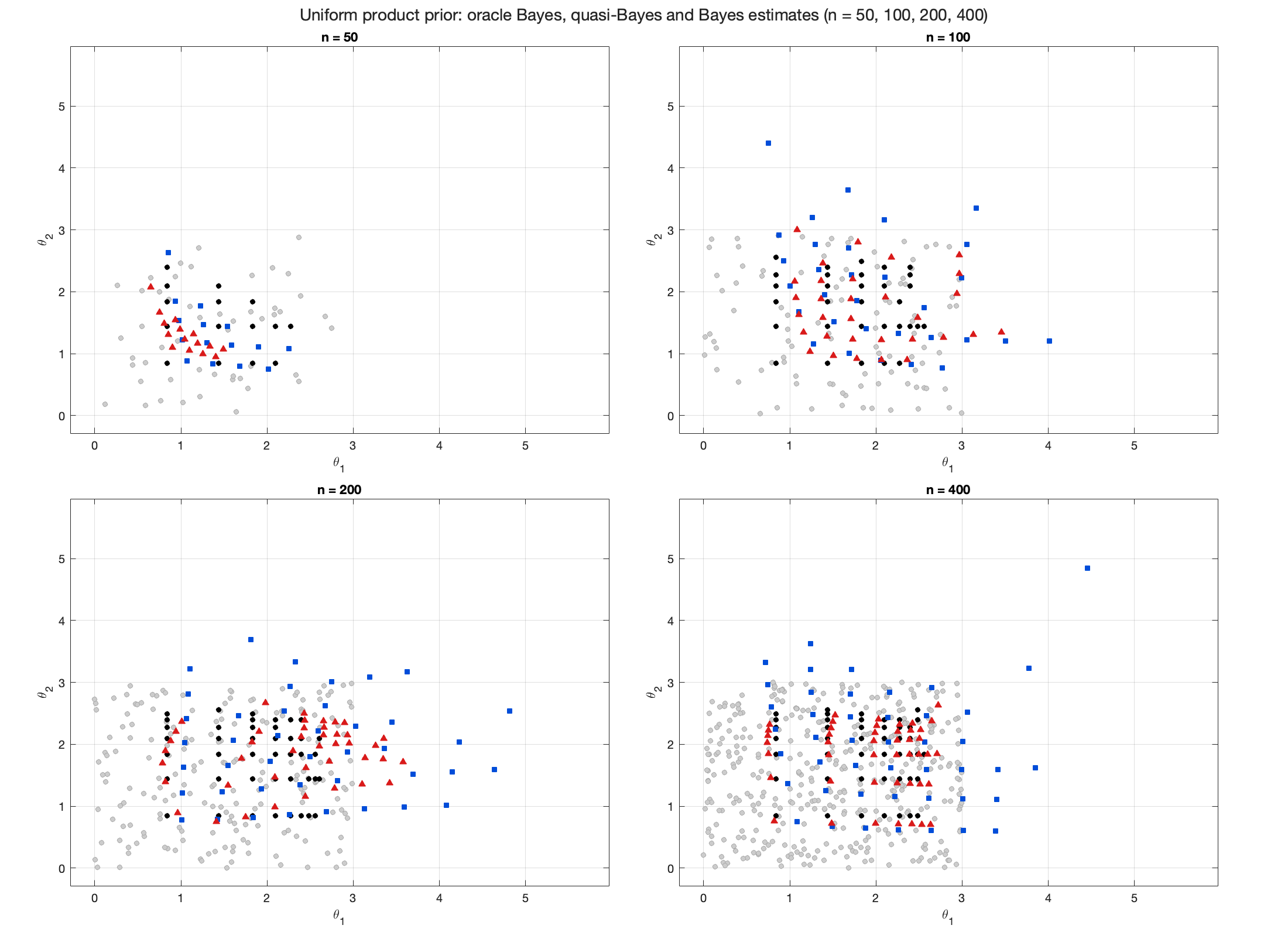}
\caption{\footnotesize{Uniform product prior, $n\in\{50,\,100,\,200,\,400\}$: data points plotted against the ``true'' parameters (grey), together with the corresponding oracle Bayes (black), Bayes (red), and quasi-Bayes (blue) estimates.}}
\label{unif_compare1_d}
\end{figure}

\begin{figure}
\centering
\includegraphics[width=.75\textwidth,height=.55\textheight,keepaspectratio]{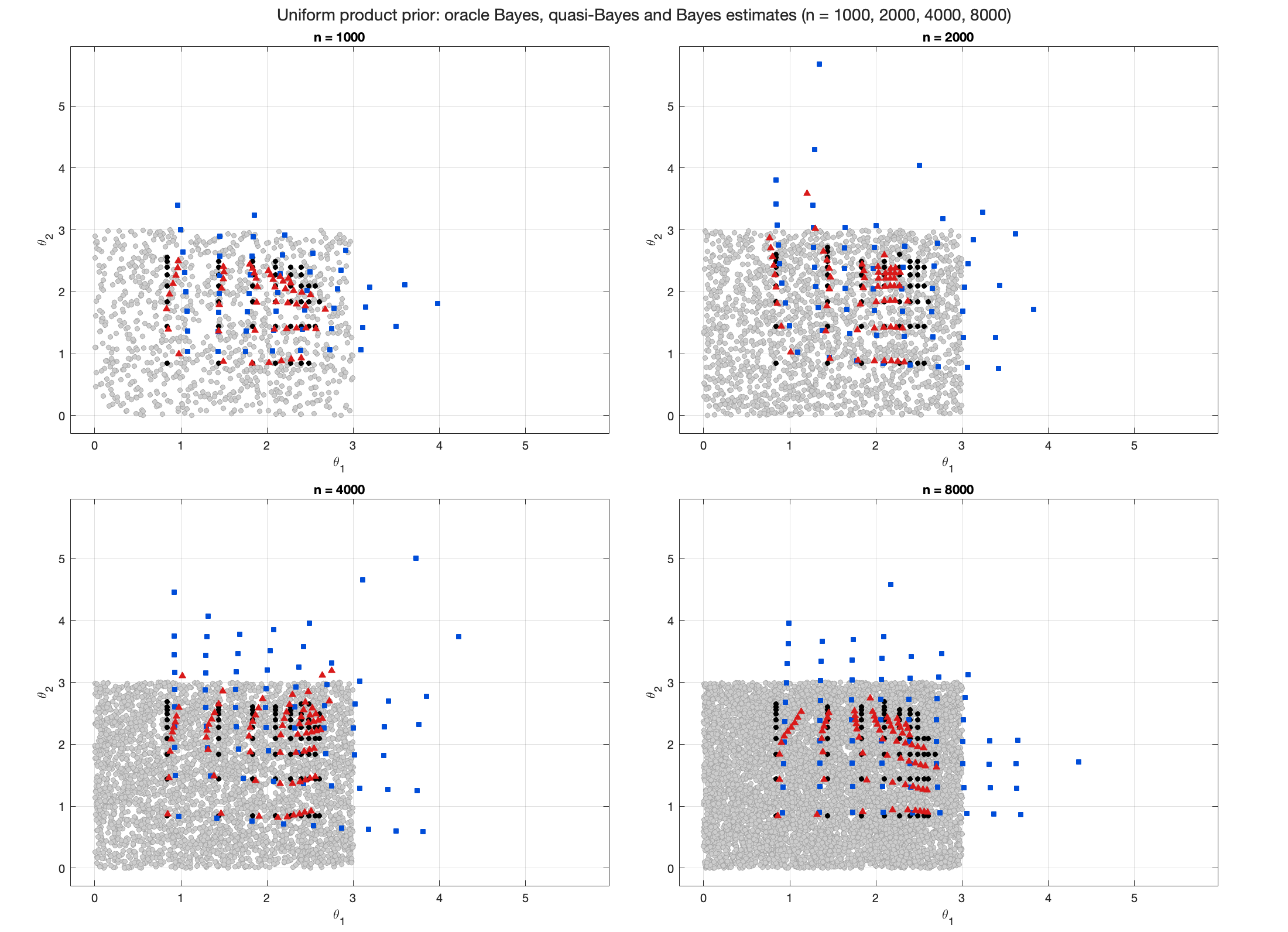}
\caption{\footnotesize{Uniform product prior, $n\in\{1,000,\,2,000,\,4,000,\,8,000\}$: data points plotted against the ``true'' parameters (grey), together with the corresponding oracle Bayes (black), Bayes (red), and quasi-Bayes (blue) estimates.}}
\label{unif_compare2_d}
\end{figure}

\begin{figure}
\centering
\includegraphics[width=.95\textwidth,height=.85\textheight,keepaspectratio]{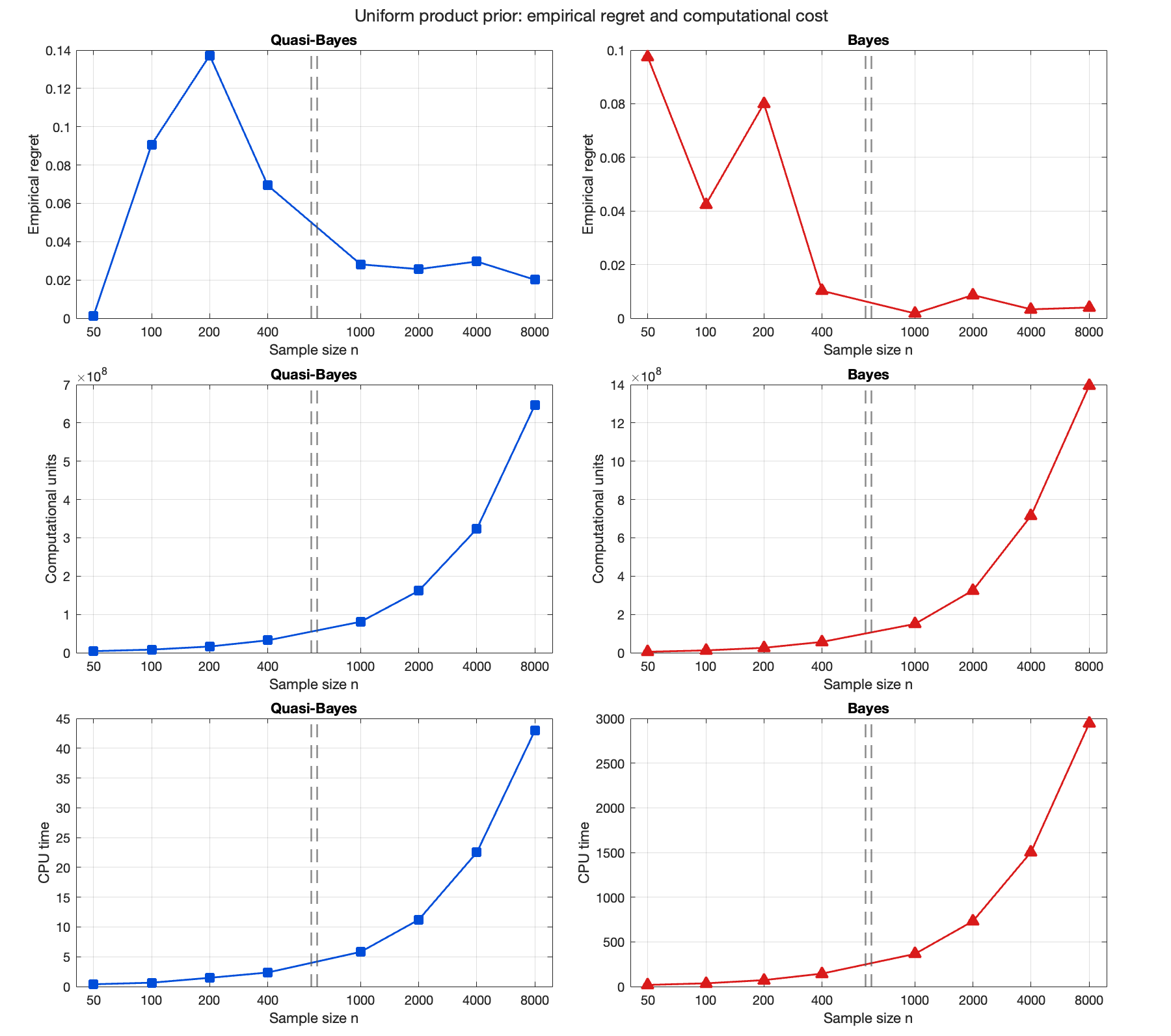}
\caption{\footnotesize{Uniform product prior: quasi-Bayes (blue) and Bayes (red) estimates compared by E-regret (top panels), computational units (middle panels), and CPU time (bottom panels).}}
\label{unif_cpu_d}
\end{figure}

\begin{figure}
\centering
\includegraphics[width=.95\textwidth,height=.85\textheight,keepaspectratio]{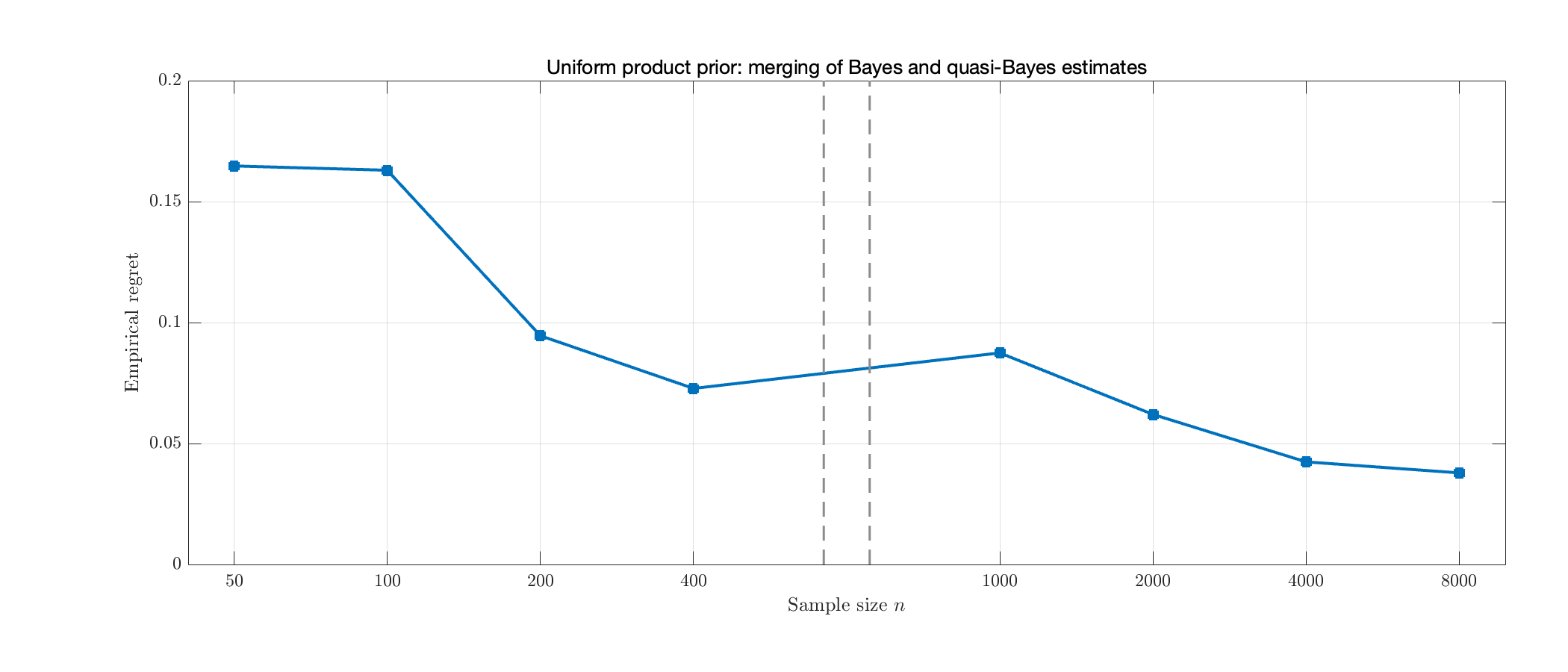}
\caption{\footnotesize{Uniform product prior: E-regret incurred by using the quasi-Bayes estimate in place of the Bayes estimate.}}
\label{unif_regret_comparison_d}
\end{figure}

\subsection{Half-Gaussian product prior}

For sample sizes $n\in\{50,\,100,\,200,\,400,\,1,000,\,2,000,\,4,000,\,8,000\}$, we generate i.i.d. data $\boldsymbol Y_{1:n}=(\boldsymbol Y_1,\ldots,\boldsymbol Y_n)$, with $\boldsymbol Y_i=(Y_{i,1},Y_{i,2})\in\mathbb N_0^2$ from a $2$-dimensional Poisson mixture model \eqref{eq:mixture_model_poisson_d} with a product half-Gaussian prior $G=G_1\otimes G_2$,  where $G_{\ell}$ is the half-Gaussian distribution with $\sigma=1$. We compare the quasi-Bayes estimate $\hat{\boldsymbol{\theta}}^{\text{\tiny{[Q-B]}}}_{\gamma,n}$ and the Bayes estimate $\hat{\boldsymbol{\theta}}^{\text{\tiny{[B]}}}_{n}$ with the oracle Bayes estimate $\hat{\boldsymbol{\theta}}^{\ast}$. In particular, the oracle $\hat{\boldsymbol{\theta}}^{\ast}$ is obtained from \eqref{eq:oracle_d} with $G^{\ast}=G_1^{\ast}\otimes G_2^{\ast}$ being the product half-Gaussian  prior distribution that generates the $\boldsymbol{\theta}_{i}$'s, and evaluating the marginal likelihood $p_{G^\ast}$ numerically through the trapezoidal rule.

With regards to the quasi-Bayes estimate $\hat{\boldsymbol\theta}^{\text{\tiny{[Q-B]}}}_{\gamma,n}$, the implementation of the $2$-dimensional Newton's algorithm is the same as in the synthetic-data analysis with the product Weibull prior: i) the density function of $G_n$ is represented through its values on a fixed tensor-product grid over $\boldsymbol{\Theta}=(0,U_{\Theta_{1}})\times(0,U_{\Theta_{2}})$ with $201$ quadrature points per coordinate, yielding $201^2$ grid points in total; ii) $G_0$ is Uniform over $\boldsymbol{\Theta}$; iii) the learning rate is $\alpha_{n}=(1+n)^{-0.99}$. With regards to the Bayes estimate  $\hat{\boldsymbol{\theta}}^{\text{\tiny{[B]}}}_{n}$, the implementation of the $2$-dimensional Algorithm 8 is the same as in the synthetic-data analysis with the product Weibull prior: we set the strength parameter $c=1$, use the same product Gamma base probability measure, take $m=5$ auxiliary components, and use the same MCMC settings.

Figure \ref{gauss_compare1_d}-\ref{gauss_compare2_d} display the quasi-Bayes, Bayes and oracle Bayes estimates. Figure \ref{gauss_cpu_d} compares the quasi-Bayes and Bayes estimates in terms of empirical performance, measured by the E-regret, and computational cost, measured by the number of computational units and CPU time. Finally, Figure \ref{gauss_regret_comparison_d} reportsthe empirical regret of the quasi-Bayes estimate $\hat{\boldsymbol{\theta}}^{\text{\tiny{[Q-B]}}}_{\gamma,n}$ relative to the Bayes estimate $\hat{\boldsymbol{\theta}}^{\text{\tiny{[B]}}}_{n}$, providing an empirical validation of the merging as the sample size $n$ grows.

\begin{figure}
\centering
\includegraphics[width=.75\textwidth,height=.55\textheight,keepaspectratio]{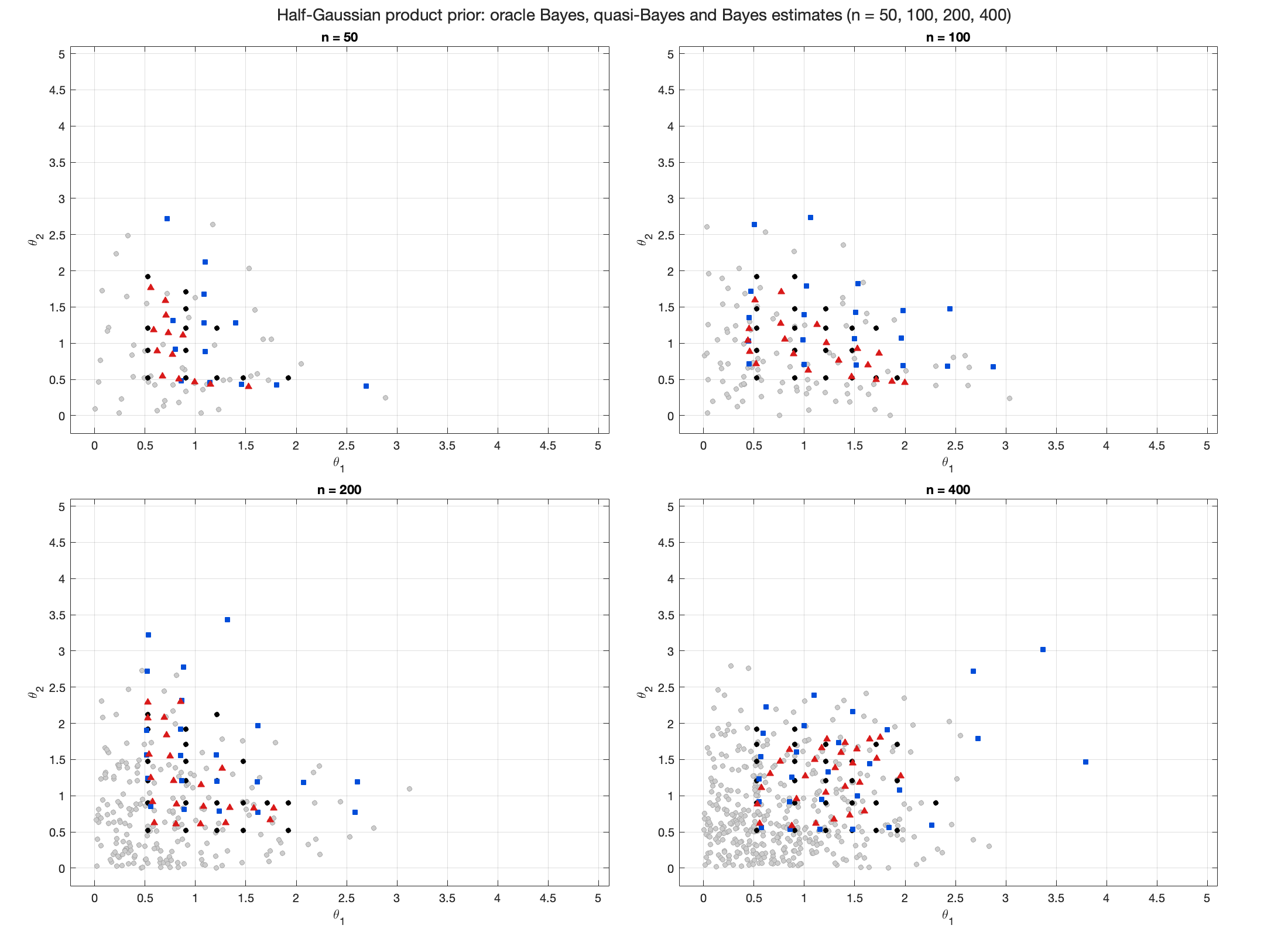}
\caption{\footnotesize{Half-Gaussian product prior, $n\in\{50,\,100,\,200,\,400\}$: data points plotted against the ``true'' parameters (grey), together with the corresponding oracle Bayes (black), Bayes (red), and quasi-Bayes (blue) estimates.}}
\label{gauss_compare1_d}
\end{figure}

\begin{figure}
\centering
\includegraphics[width=.75\textwidth,height=.55\textheight,keepaspectratio]{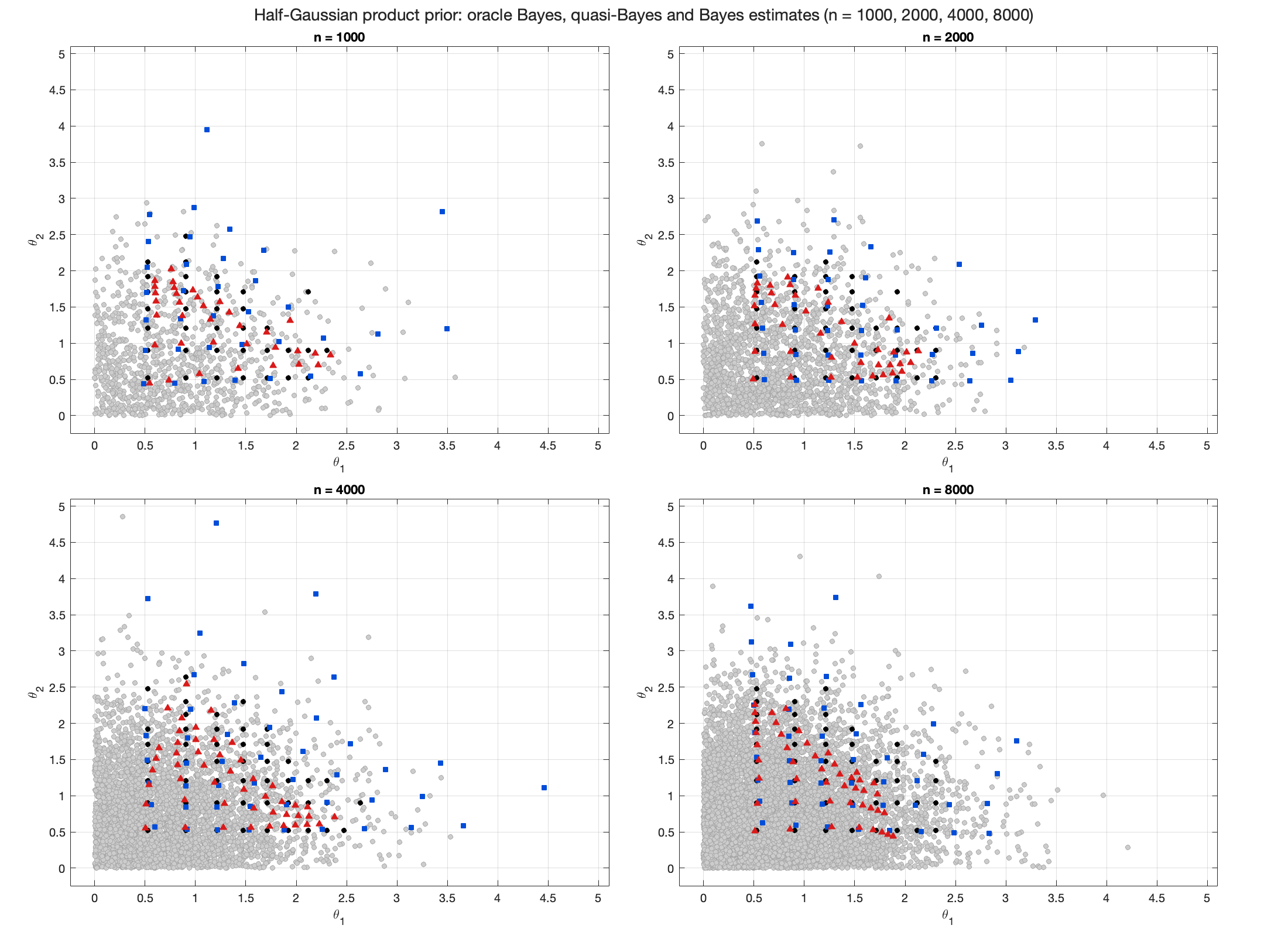}
\caption{\footnotesize{Half-Gaussian product prior, $n\in\{1,000,\,2,000,\,4,000,\,8,000\}$: data points plotted against the ``true'' parameters (grey), together with the corresponding oracle Bayes (black), Bayes (red), and quasi-Bayes (blue) estimates.}}
\label{gauss_compare2_d}
\end{figure}

\begin{figure}
\centering
\includegraphics[width=.95\textwidth,height=.85\textheight,keepaspectratio]{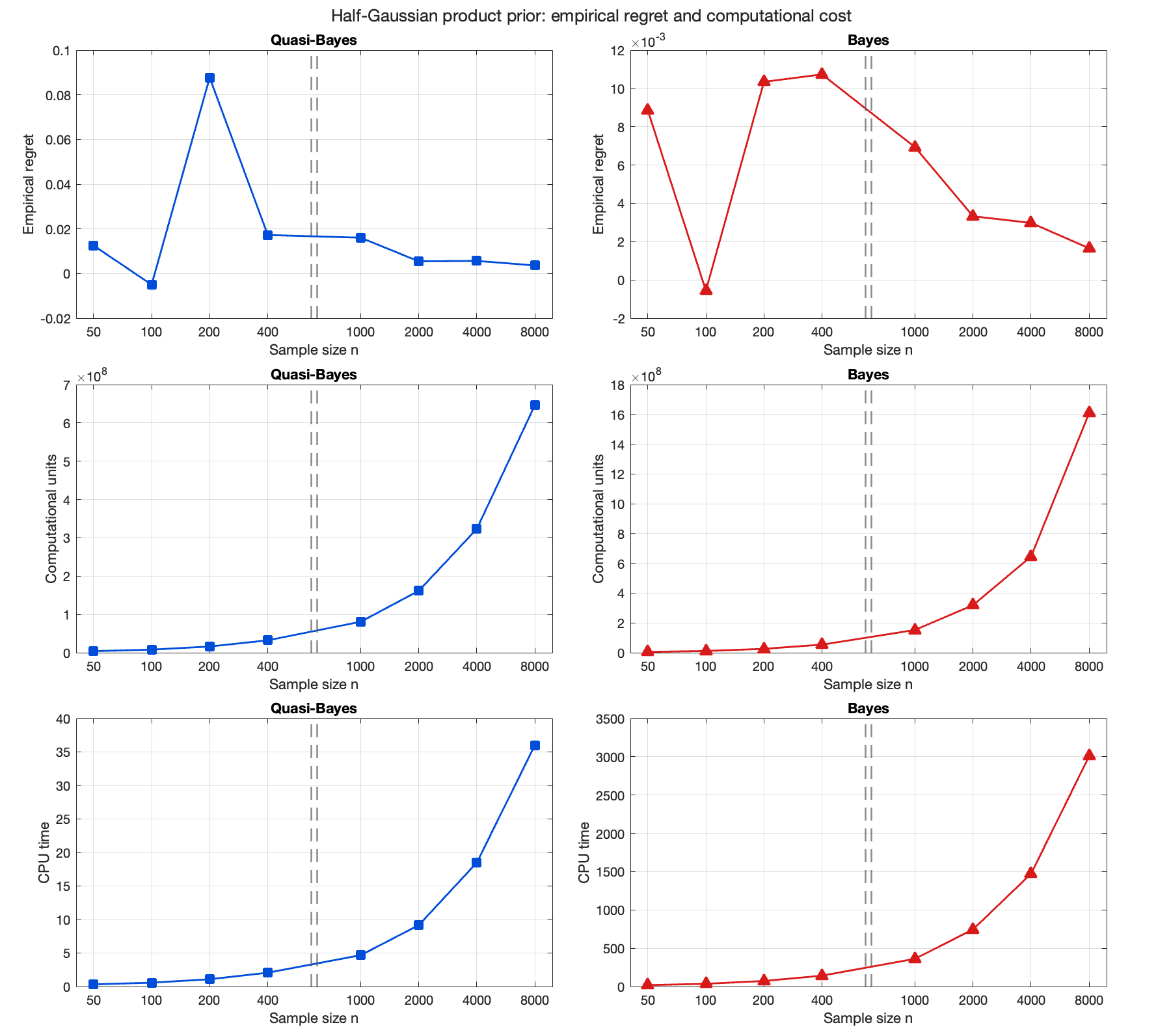}
\caption{\footnotesize{Half-Gaussian product prior: quasi-Bayes (blue) and Bayes (red) estimates compared by E-regret (top panels), computational units (middle panels), and CPU time (bottom panels).}}
\label{gauss_cpu_d}
\end{figure}

\begin{figure}
\centering
\includegraphics[width=.95\textwidth,height=.85\textheight,keepaspectratio]{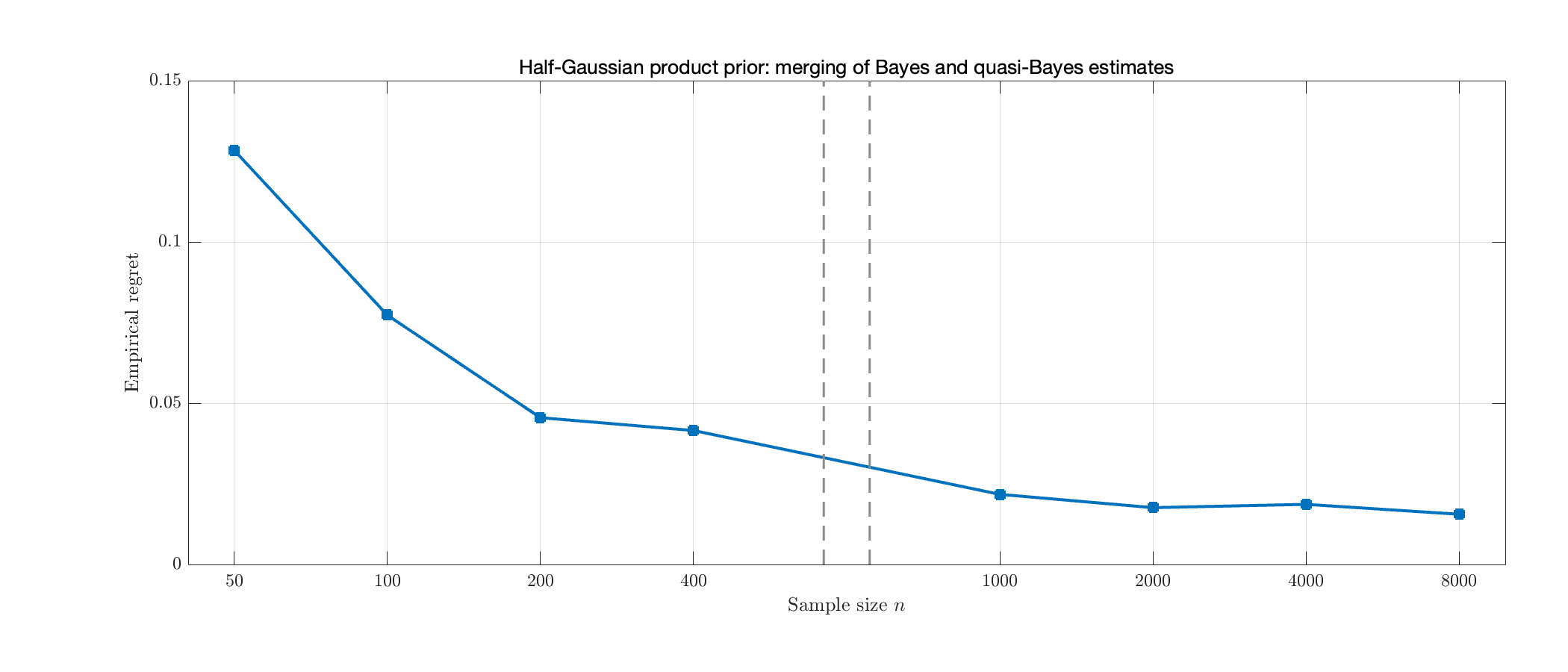}
\caption{\footnotesize{Half-Gaussian product prior: E-regret incurred by using the quasi-Bayes estimate in place of the Bayes estimate.}}
\label{gauss_regret_comparison_d}
\end{figure}

\subsection{Square-root of half-Cauchy product prior}

For sample sizes $n\in\{50,\,100,\,200,\,400,\,1,000,\,2,000,\,4,000,\,8,000\}$, we generate i.i.d. data $\boldsymbol Y_{1:n}=(\boldsymbol Y_1,\ldots,\boldsymbol Y_n)$, with $\boldsymbol Y_i=(Y_{i,1},Y_{i,2})\in\mathbb N_0^2$ from a $2$-dimensional Poisson mixture model \eqref{eq:mixture_model_poisson_d} with a product square-root of half-Cauchy prior $G=G_1\otimes G_2$, where $G_{\ell}$ is the square-root of half-Cauchy prior. We compare the quasi-Bayes estimate $\hat{\boldsymbol{\theta}}^{\text{\tiny{[Q-B]}}}_{\gamma,n}$ and the Bayes estimate $\hat{\boldsymbol{\theta}}^{\text{\tiny{[B]}}}_{n}$ with the oracle Bayes estimate $\hat{\boldsymbol{\theta}}^{\ast}$. In particular, the oracle $\hat{\boldsymbol{\theta}}^{\ast}$ is obtained from \eqref{eq:oracle_d} with $G^{\ast}=G_1^{\ast}\otimes G_2^{\ast}$ being the product square-root of half-Cauchy prior distribution that generates the $\boldsymbol{\theta}_{i}$'s, and evaluating the marginal likelihood $p_{G^\ast}$ numerically through the trapezoidal rule.
 
With regards to the quasi-Bayes estimate $\hat{\boldsymbol\theta}^{\text{\tiny{[Q-B]}}}_{\gamma,n}$, the implementation of the $2$-dimensional Newton's algorithm is the same as in the synthetic-data analysis with the product Weibull prior: i) the density function of $G_n$ is represented through its values on a fixed tensor-product grid over $\boldsymbol{\Theta}=(0,U_{\Theta_{1}})\times(0,U_{\Theta_{2}})$ with $201$ quadrature points per coordinate, yielding $201^2$ grid points in total; ii) $G_0$ is Uniform over $\boldsymbol{\Theta}$; iii) the learning rate is $\alpha_{n}=(1+n)^{-0.99}$. With regards to the Bayes estimate  $\hat{\boldsymbol{\theta}}^{\text{\tiny{[B]}}}_{n}$, the implementation of the $2$-dimensional Algorithm 8 is the same as in the synthetic-data analysis with the product Weibull prior: we set the strength parameter $c=1$, use the same product Gamma base probability measure, take $m=5$ auxiliary components, and use the same MCMC settings.

Figure \ref{cau_compare1_d}-\ref{cau_compare2_d} display the quasi-Bayes, Bayes and oracle Bayes estimates. Figure \ref{cau_cpu_d} compares the quasi-Bayes and Bayes estimates in terms of empirical performance, measured by the E-regret, and computational cost, measured by the number of computational units and CPU time. Finally, Figure \ref{cau_regret_comparison_d} reportsthe empirical regret of the quasi-Bayes estimate $\hat{\boldsymbol{\theta}}^{\text{\tiny{[Q-B]}}}_{\gamma,n}$ relative to the Bayes estimate $\hat{\boldsymbol{\theta}}^{\text{\tiny{[B]}}}_{n}$, providing an empirical validation of the merging as the sample size $n$ grows.

\begin{figure}
\centering
\includegraphics[width=.75\textwidth,height=.55\textheight,keepaspectratio]{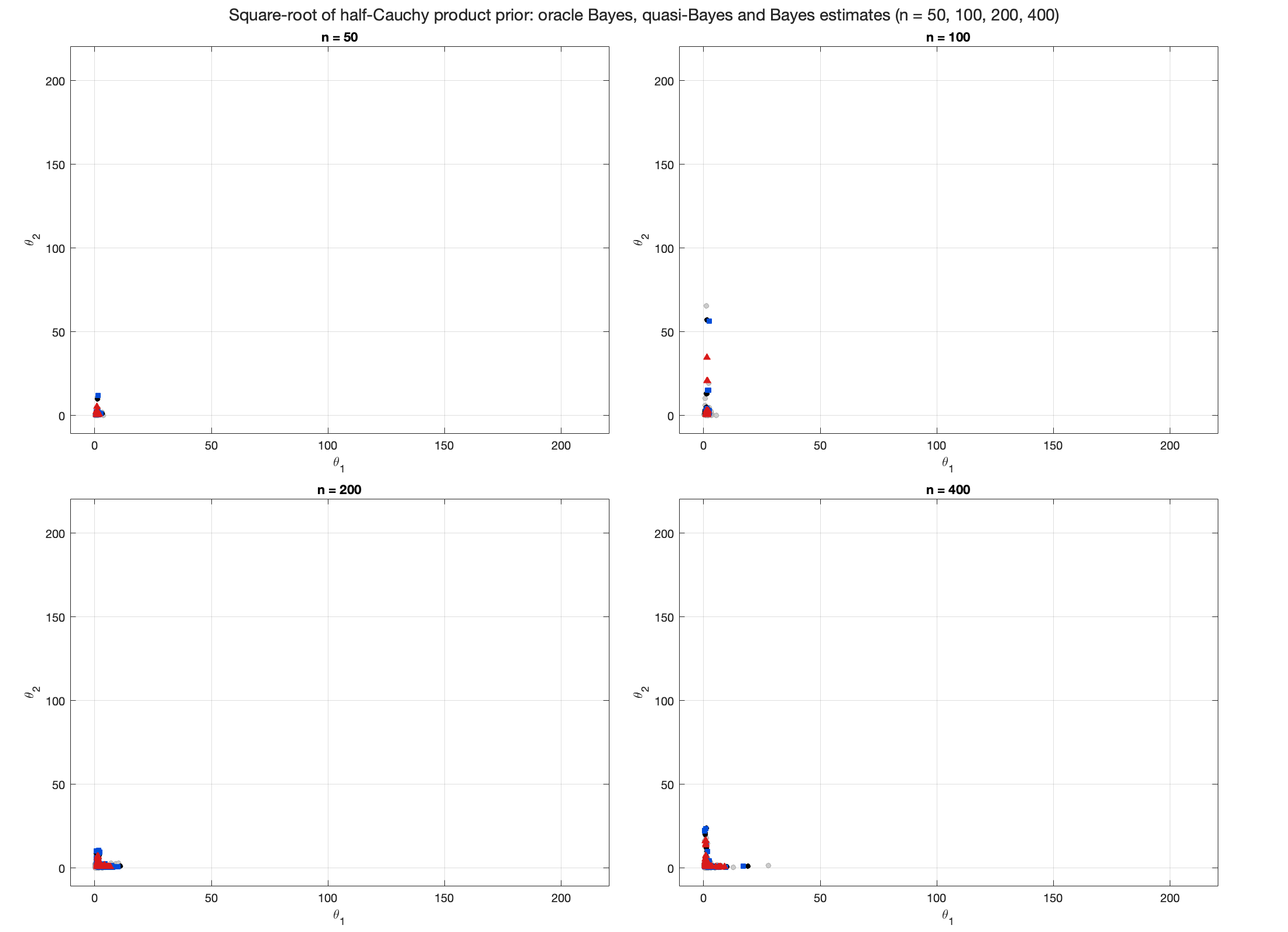}
\caption{\footnotesize{Square-root of half-Cauchy product prior, $n\in\{50,\,100,\,200,\,400\}$: data points plotted against the ``true'' parameters (grey), together with the corresponding oracle Bayes (black), Bayes (red), and quasi-Bayes (blue) estimates.}}
\label{cau_compare1_d}
\end{figure}

\begin{figure}
\centering
\includegraphics[width=.75\textwidth,height=.55\textheight,keepaspectratio]{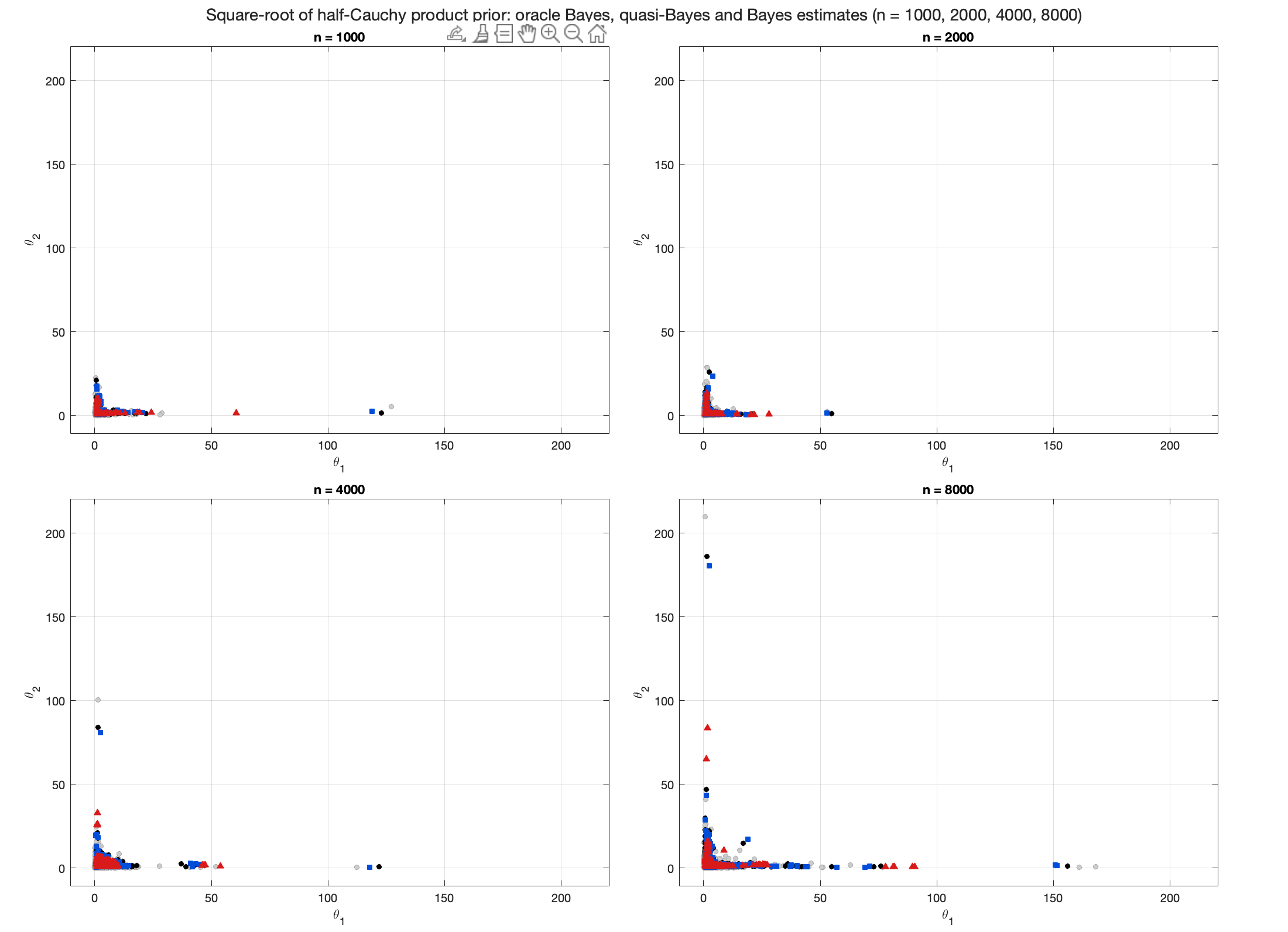}
\caption{\footnotesize{Square-root of half-Cauchy product prior, $n\in\{1,000,\,2,000,\,4,000,\,8,000\}$: data points plotted against the ``true'' parameters (grey), together with the corresponding oracle Bayes (black), Bayes (red), and quasi-Bayes (blue) estimates.}}
\label{cau_compare2_d}
\end{figure}

\begin{figure}
\centering
\includegraphics[width=.95\textwidth,height=.85\textheight,keepaspectratio]{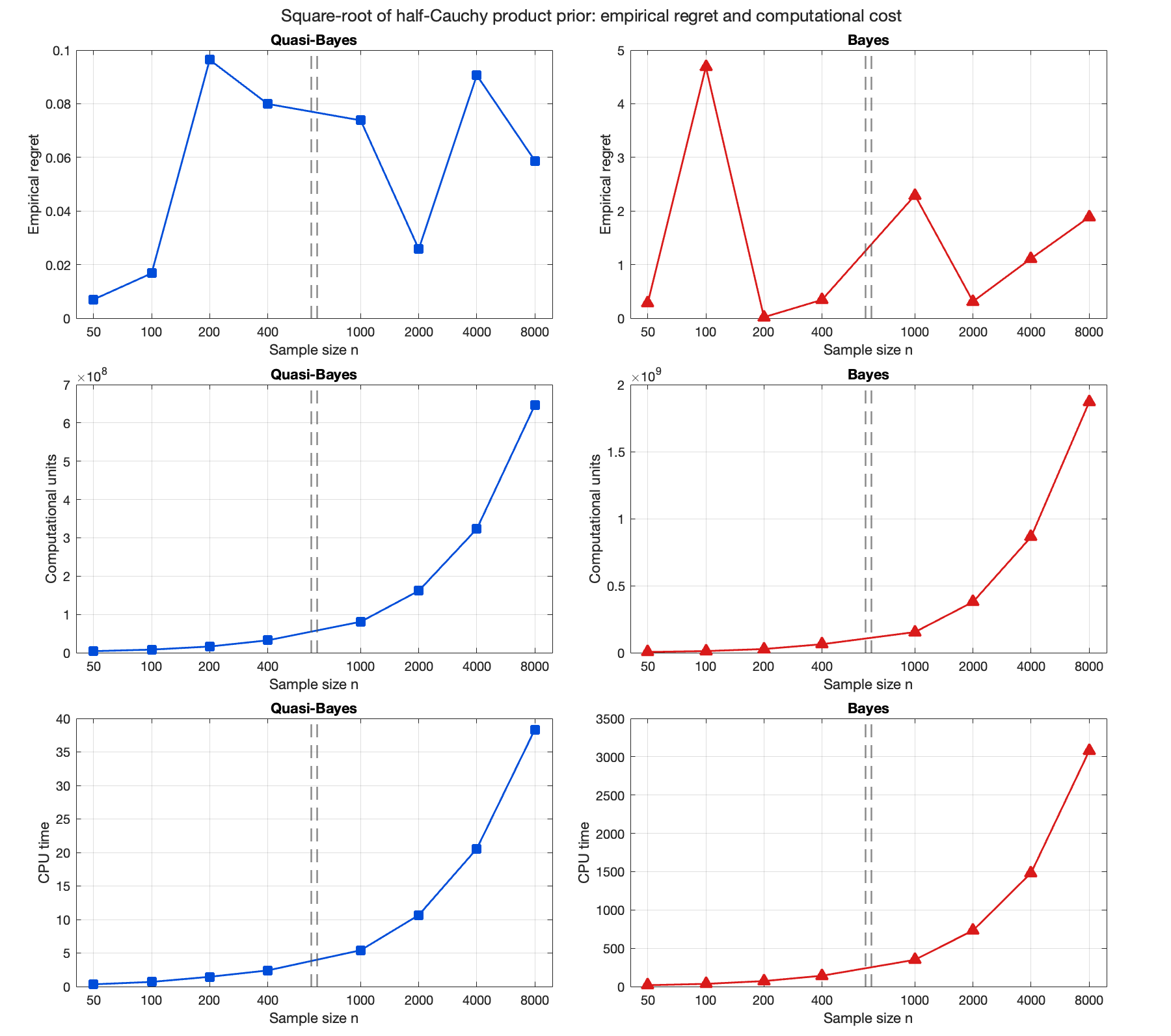}
\caption{\footnotesize{Square-root of half-Cauchy product prior: quasi-Bayes (blue) and Bayes (red) estimates compared by E-regret (top panels), computational units (middle panels), and CPU time (bottom panels).}}
\label{cau_cpu_d}
\end{figure}

\begin{figure}
\centering
\includegraphics[width=.95\textwidth,height=.85\textheight,keepaspectratio]{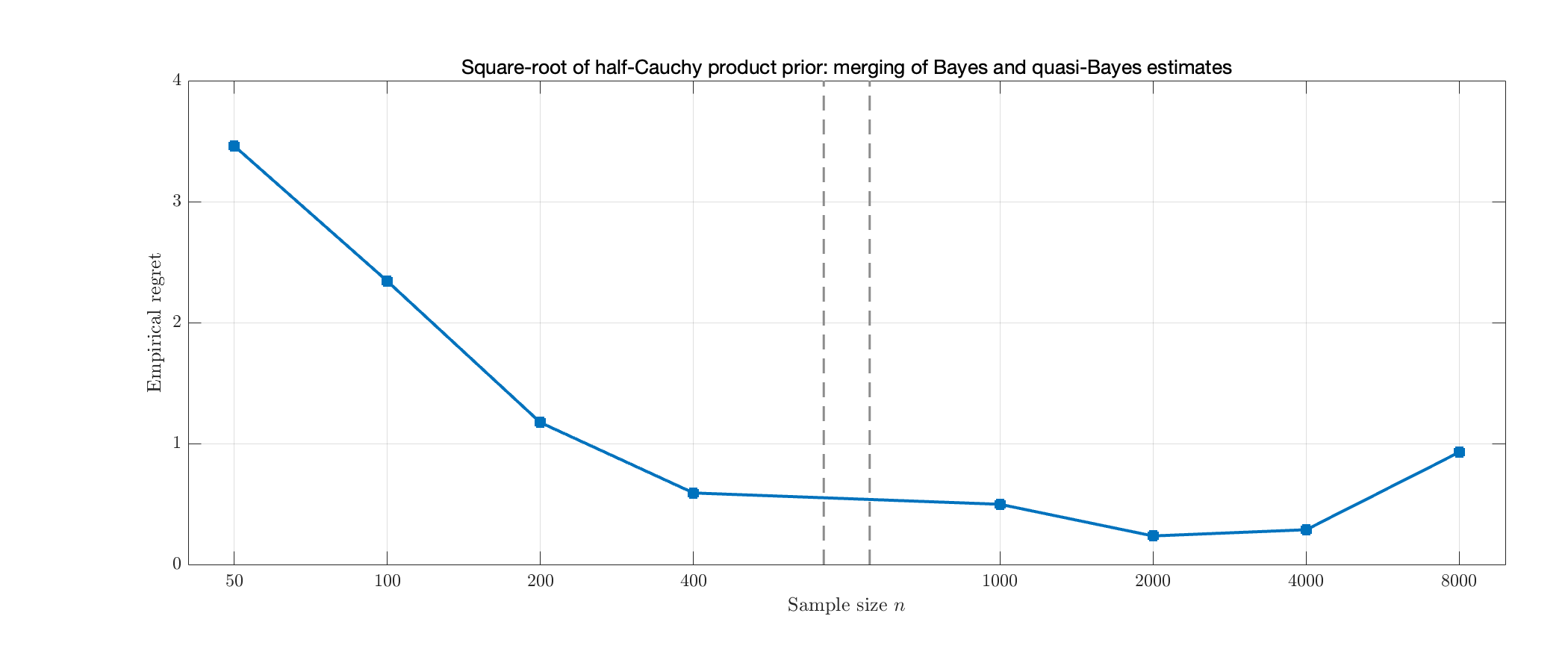}
\caption{\footnotesize{Square-root of half-Cauchy product prior: E-regret incurred by using the quasi-Bayes estimate in place of the Bayes estimate.}}
\label{cau_regret_comparison_d}
\end{figure}

%%%%%%%%%%%%%%%%%%%%%%%%%%%%%%%%
%%%%%%%%%%%%%%%%%%%%%%%%%%%%%%%%
%%%%%%%%%%%%%%%%%%%%%%%%%%%%%%%%
%%%%%%%%%%%%%%%%%%%%%%%%%%%%%%%%

\phantomsection


\addcontentsline{toc}{section}{References}
\begin{thebibliography}{9}

\bibitem[Blei and Jordan(2006)]{BleJor(06)} 
\textsc{Blei, D.M. and Jordan, M.I.} (2006). Variational inference for Dirichlet process mixtures. \textit{Bayesian Anal.} \textbf{1}, 121--144.

\bibitem[Brown and Farrell(1985)]{Bro(85)}
\textsc{Brown, L.D. and Farrell, R.H.} (1985). Complete class theorems for estimation of multivariate Poisson means and related problems. \textit{Ann. Statist.} \textbf{13}, 706--726.

\bibitem[Brown et al.(2013)]{Bro(13)}
\textsc{Brown, L.D., Greenshtein, E. and Ritov, Y.} (2013). The Poisson compound decision problem revisited. \textit{J. Am. Statist. Assoc.} \textbf{108}, 741--749.

\bibitem[Cannella et al.(2026)]{Can(26)}
\textsc{Cannella, N., Teh, A., Han Y. and Polyanskiy, Y.} (2026). Universal priors: solving empirical Bayes via Bayesian inference and pretraining. \textit{Preprint arXiv:2602.15136}.

\bibitem[Deely and Lindley(1981)]{Dee(81)}
\textsc{Deely, J.J. and Lindley, D.V.} (1981). Bayes empirical Bayes. \textit{J. Am. Statist. Assoc.} \textbf{76}, 833--841.

\bibitem[Efron(2014)]{Efr(14)}
\textsc{Efron, B.} (2014). Two modeling strategies for empirical Bayes estimation. \textit{Statist. Sci.} \textbf{29}, 285--301.

\bibitem[Efron(2019)]{Efr(19)}
\textsc{Efron, B.} (2019). Bayes, oracle Bayes and empirical Bayes. \textit{Statist. Sci.} \textbf{34}, 177--201.

\bibitem[Favaro and Fortini(2026)]{Fav(24)}
\textsc{Favaro, S. and Fortini, S.} (2026). Quasi-Bayes empirical Bayes: a sequential approach to the Poisson compound decision problem. \textit{Biometrika}, to appear.

\bibitem[Ferguson(1973)]{Fer(73)}
\textsc{Ferguson, T.S.} (1973). A Bayesian analysis of some nonparametric problems. \textit{Ann. Statist.} \textbf{1}, 209--230.

\bibitem[Fortini and Petrone(2020)]{For(20)}
\textsc{Fortini, S. and Petrone, S.} (2020). Quasi-Bayesian properties of a procedure for sequential learning in mixture models. \textit{J. R. Statist. Soc. B} \textbf{82}, 1087--1114.

\bibitem[Ghosal and van der Vaart(2001)]{Gho(01)}
\textsc{Ghosal, S. and van der Vaart, A.W.} (2001). Entropies and rates of convergence for maximum likelihood and Bayes estimation for mixtures of normal densities. \textit{Ann. Statist.} \textbf{29}, 1233--1263.

\bibitem[Ignatiadis and Kankanala(2026)]{Ign(26)}
\textsc{Ignatiadis, N. and Kankanala, S.} (2026). Compound decisions and empirical Bayes via Bayesian nonparametrics. \textit{Preprint arXiv:2602.20115}.

\bibitem[Ishwaran and James(2001)]{IshJam(01)} 
\textsc{Ishwaran, H. and James, L.F.} (2001). Gibbs sampling methods for stick-breaking priors. \textit{J. Amer. Statist. Assoc.} \textbf{96}, 161--173.

\bibitem[Jana et al.(2023)]{Jan(23)}
\textsc{Jana, S., Polyanskiy, Y., Teh, A. and Wu, Y.} (2023). Empirical Bayes via ERM and Rademacher complexities: the Poisson model. \textit{P. Mach. Learn. Res.} \textbf{195}, 1--37.

\bibitem[Jana et al.(2025)]{Jan(24)}
\textsc{Jana, S., Polyanskiy, Y. and Wu, Y.} (2025). Optimal empirical Bayes estimation for the Poisson model via minimum-distance methods. \textit{Inf. Inference} \textbf{14}, 1--42.

\bibitem[Johnstone(1986)]{Joh(86)}
\textsc{Johnstone, I.} (1986). Admissible estimation, Dirichlet principles and recurrence of birth-death chains on $\mathbb{Z}_{+}^{p}$.  \textit{Probab. Theory Related Fields} \textbf{71}, 231--269.

\bibitem[Lo(1984)]{Lo(84)}   
\textsc{Lo, A.Y.} (1984). On a class of Bayesian nonparametric estimates. I. Density estimates \textit{Ann. Statist.} \textbf{12}, 351--357.

\bibitem[Martin and Ghosh(2008)]{Mar(08)}
\textsc{Martin, R. and Ghosh, J.K.} (2008). Stochastic approximation and Newton’s estimate of a mixing distribution. \textit{Statist. Sci.} \textbf{23}, 365--382.

\bibitem[Martin and Tokdar(2009)]{MarTok(09)}
\textsc{Martin, R. and  Tokdar, S.T.} (2009) Asymptotic properties of predictive recursion: robustness and rate of convergence. \textit{Electron. J. Stat.} \textbf{3}, 1455--1472.

\bibitem[Neal(2000)]{Nea(00)} 
\textsc{Neal, R. M.} (2000). Markov chain sampling methods for Dirichlet process mixture models. \textit{J. Comput. Graph. Statist.} \textbf{9}, 249--265.

\bibitem[Newton et al.(1998)]{New(98)}
\textsc{Newton, M.A., Quintana, F.A. and Zhang, Y.} (1998). Nonparametric Bayes methods using predictive updating. In \textit{Practical Nonparametric and Semiparametric Bayesian Statistics}, Springer.

\bibitem[Papaspiliopoulos and Roberts(2008)]{PapRob(08)} 
\textsc{Papaspiliopoulos, O. and Roberts, G.O.} (2008). Retrospective Markov chain Monte Carlo methods for Dirichlet process hierarchical models. \textit{Biometrika} \textbf{95}, 169--186.

\bibitem[Polyanskiy and Wu(2021)]{Pol(21)}
\textsc{Polyanskiy, Y. and Wu, Y} (2021). Sharp regret bounds for empirical Bayes and compound decision problems. \textit{Preprint arXiv: 2109.03943}.

\bibitem[Robbins(1951)]{Rob(51)}
\textsc{Robbins, H.} (1951). Asymptotically subminimax solutions of compound decision problems. In \textit{Proc. Second Berkeley Symp} \textbf{2}, 131--148.

\bibitem[Robbins(1956)]{Rob(56)}
\textsc{Robbins, H.} (1956). An empirical Bayes approach to statistics. In \textit{Proc. Third Berkeley Symp.} \textbf{3}, 157--164.

\bibitem[Shen and Wu(2026)]{She(24)}
\textsc{Shen, Y. and Wu, Y.} (2026). Poisson empirical Bayes estimation: When does $g$-modeling beat $f$-modeling in theory (and in practice)? \textit{Ann. Statist.} \textbf{54}, 146--175.

\bibitem[Smith and Makov(1978)]{Smi(78)}
\textsc{Smith, A.F.M. and Makov, U.E.} (1978). A quasi-Bayes sequential procedure for mixtures. \textit{J. R. Statist. Soc. B} \textbf{40}, 106--112.

\bibitem[Teh et al.(2025)]{Teh(25)}
\textsc{Teh, A., Jabbour, M. and Polyanskiy, Y.} (2025). Solving empirical Bayes via transformers. \textit{Preprint arXiv:2502.09844}.

\bibitem[Walker(2007)]{Wal(07)}
\textsc{Walker, S.G.} (2007). Sampling the Dirichlet mixture model with slices. \textit{Comm. Statist. Simulation Comput.} \textbf{36}, 45--54.

\bibitem[Zhang(2003)]{Zha(03)}
\textsc{Zhang, C.-H.} (2003). Compound decision theory and empirical Bayes methods. \textit{Ann. Statist.} \textbf{31}, 379--390.

\bibitem[Zhang(2005)]{Zha(05)}
\textsc{Zhang, C.-H.} (2005). Estimation of sums of random variables: examples and information bounds. \textit{Ann. Statist.} \textbf{33}, 2022--2041.

\end{thebibliography}

\addcontentsline{toc}{section}{References}
\begin{thebibliography}{9}

\bibitem[Ghosal and van der Vaart(2001)]{Gho(01)}
\textsc{Ghosal, S. and van der Vaart, A.W.} (2001). Entropies and rates of convergence for maximum likelihood and Bayes estimation for mixtures of normal densities. \textit{Ann. Statist.} \textbf{29}, 1233--1263.

\bibitem[Neal(2000)]{Nea(00)} 
\textsc{Neal, R. M.} (2000). Markov chain sampling methods for Dirichlet process mixture models. \textit{J. Comput. Graph. Statist.} \textbf{9}, 249--265.

\bibitem[Shen and Wu(2026)]{She(24)}
\textsc{Shen, Y. and Wu, Y.} (2026). Poisson empirical Bayes estimation: When does $g$-modeling beat $f$-modeling in theory (and in practice)? \textit{Ann. Statist.} \textbf{54}, 146--175.

\bibitem[Wong and Shen(1995))]{Won(95)}
\textsc{Wong, H.W. and Shen, X.} (1995). Probability inequalities for likelihood ratios and convergence rates of sieve MLES \textit{Ann. Statist.} \textbf{23}, 339--362.

\end{thebibliography}
\end{document}